\documentclass[12pt]{article}

\usepackage{verbatim,color,amssymb}
\usepackage{amsmath}					
\usepackage{amsthm}					
\usepackage{natbib}
\usepackage{multirow}
\usepackage{setspace}
\usepackage[mathscr]{euscript}
\usepackage{fancyhdr}
\usepackage{enumitem}
\usepackage{graphicx}
\usepackage{lineno}
\usepackage{array,booktabs}
\usepackage{comment}
\usepackage{soul}
\usepackage{url}

\usepackage{multibib}
\newcites{latex}{References}

\usepackage{lscape}

\usepackage{setspace}
\usepackage{tikz}
\usetikzlibrary{arrows}
\usepackage{multirow}

\usepackage{subfigure}

\makeatletter
\renewcommand{\paragraph}{%
  \@startsection{paragraph}{4}%
  {\z@}{0.5ex \@plus 1ex \@minus 1ex}{-1em}%
  {\normalfont\normalsize\bfseries}%
}
\def\thm@space@setup{\thm@preskip=5pt
\thm@postskip=5pt}
\makeatother

\newtheorem{Thm}{\underline{\bf Theorem}}
\newtheorem{Assmp}{\underline{\bf Assumptions}}

\newtheorem*{Proof*}{Proof}

\newtheorem{Rem}{\underline{\bf Remark}}

\newtheorem{Lem}{\underline{\bf Lemma}}
\newtheorem{Cor}{\underline{\bf Corollary}}

\def\eE{\mathbb{E}}

\def\C{{\cal C}}

\def\S{{\cal S}}

\def\W{{\cal W}}

\def\simind{\stackrel{\mbox{\scriptsize{ind}}}{\sim}}
\def\simiid{\stackrel{\mbox{\scriptsize{iid}}}{\sim}}

\def\diag{\hbox{diag}}

\def\wt{\widetilde}

\def\diag{\hbox{diag}}

\def\var{\hbox{var}}
\def\cov{\hbox{cov}}
\def\corr{\hbox{corr}}

\def\trace{\hbox{trace}}

\def\vect{\hbox{vec}}

\def\Ga{\hbox{Ga}}

\def\MVN{\hbox{MVN}}

\def\Normal{\hbox{Normal}}

\def\P_25_ICML{{\it Proceedings of the 25th international conference on Machine learning}}

\def\bse{\begin{eqnarray*}}
\def\ese{\end{eqnarray*}}
\def\be{\begin{eqnarray}}
\def\ee{\end{eqnarray}}
\def\bq{\begin{equation}}
\def\eq{\end{equation}}

\def\vec{\hbox{vec}}

\usepackage{verbatim,color,amssymb}

\def\trans{^{\rm T}}

\def\th{^{th}}
\def\bone{{\mathbf 1}}

\def\ba{{\mathbf a}}
\def\bA{{\mathbf A}}
\def\bb{{\mathbf b}}
\def\bB{{\mathbf B}}
\def\bc{{\mathbf c}}
\def\bC{{\mathbf C}}

\def\bD{{\mathbf D}}

\def\b1e{{\mathbf e}}

\def\bg{{\mathbf g}}
\def\bG{{\mathbf G}}
\def\bI{{\mathbf I}}

\def\bL{{\mathbf L}}
\def\bM{{\mathbf M}}

\def\bP{{\mathbf P}}
\def\bq{{\mathbf q}}
\def\bQ{{\mathbf Q}}
\def\br{{\mathbf r}}
\def\bR{{\mathbf R}}

\def\bU{{\mathbf U}}

\def\bx{{\mathbf x}}
\def\bX{{\mathbf X}}
\def\by{{\mathbf y}}
\def\bY{{\mathbf Y}}

\def\bzero{{\mathbf 0}}

\newcommand{\etam}{\mbox{\boldmath $\eta$}}
\newcommand{\bmu}{\mbox{\boldmath $\mu$}}
\newcommand{\bDelta}{\mbox{\boldmath $\Delta$}}

\newcommand{\bepsilon}{\mbox{\boldmath $\epsilon$}}
\newcommand{\btheta}{\mbox{\boldmath $\theta$}}
\newcommand{\bTheta}{\mbox{\boldmath $\Theta$}}
\newcommand{\bbeta}{\mbox{\boldmath $\beta$}}
\newcommand{\bgamma}{\mbox{\boldmath $\gamma$}}
\newcommand{\bzeta}{\mbox{\boldmath $\zeta$}}
\newcommand{\bsigma}{\mbox{\boldmath $\sigma$}}
\newcommand{\bSigma}{\mbox{\boldmath $\Sigma$}}
\newcommand{\balpha}{\mbox{\boldmath $\alpha$}}

\newcommand{\blambda}{\mbox{\boldmath $\lambda$}}
\newcommand{\bLambda}{\mbox{\boldmath $\Lambda$}}
\newcommand{\bOmega}{\mbox{\boldmath $\Omega$}}

\newcommand{\bpsi}{\mbox{\boldmath $\psi$}}
\newcommand{\bGamma}{\mbox{\boldmath $\Gamma$}}

\newcommand{\abs}[1]{\left\vert#1\right\vert}

\renewcommand\footnoterule{\kern-3pt \hrule \textwidth 2in \kern 2.6pt}

\def\colred#1{\textcolor{red}{#1}}

\def\boxit#1{\vbox{\hrule\hbox{\vrule\kern6pt \vbox{\kern6pt \textcolor{blue}{#1}\kern6pt}\kern6pt\vrule}\hrule}}

\def\authorfootnote#1{{\let\thefootnote\relax\footnotetext{#1}}}

\allowdisplaybreaks

\begin{document}
\thispagestyle{empty}
\baselineskip=28pt

\begin{center}
{\LARGE{\bf 
Bayesian Semiparametric\\ 
\vskip -9pt 
Orthogonal Tucker Factorized Mixed Models\\ 
for 
Multi-dimensional Longitudinal Functional Data
}}
\end{center}
\baselineskip=12pt

\vskip 2mm
\begin{center}
 Arkaprava Roy\\
 arkaprava.roy@ufl.edu\\
 Department of Biostatistics,
 University of Florida\\
 2004 Mowry Road, Gainesville, FL  32611, USA\\
 \vskip 2mm%
 Abhra Sarkar\\
 abhra.sarkar@utexas.edu \\
 Department of Statistics and Data Sciences,
 The University of Texas at Austin\\
 2317 Speedway D9800, Austin, TX 78712-1823, USA\\
 \vskip 2mm%
 and\\
 The Alzheimer's Disease Neuroimaging Initiative
\end{center}

\begin{abstract}
\baselineskip=12pt
We introduce a novel longitudinal mixed model for analyzing complex multi-dimensional functional data, addressing challenges such as high-resolution, structural complexities, and computational demands. Our approach integrates dimension-reduction techniques, including basis function representation and Tucker tensor decomposition, to model complex functional (e.g., spatial and temporal) variations, group differences, and individual heterogeneity while drastically reducing model dimensions. The model accommodates 
multiplicative random effects whose marginalization yields a novel Tucker-decomposed covariance-tensor framework. To ensure scalability, we employ semi-orthogonal mode matrices implemented via a novel graph-Laplacian-based smoothness prior with low-rank approximation, leading to an efficient posterior sampling method. A cumulative shrinkage strategy promotes sparsity and enables semi-automated rank selection. We establish theoretical guarantees for posterior convergence and demonstrate the method's effectiveness through simulations, showing significant improvements over existing techniques. Applying the method to Alzheimer’s Disease Neuroimaging Initiative (ADNI) neuroimaging data reveals novel insights into local brain changes associated with disease progression, highlighting the method's practical utility for studying cognitive decline and neurodegenerative conditions.
\end{abstract}

\vskip 20pt 
\baselineskip=12pt
\noindent\underline{\bf Key Words}: 
B-spline mixtures, 
Fractional anisotropy, 
Higher-order singular value decomposition, 
Functional data, 
Longitudinal mixed models, 
Multi-group data, 
Multiplicative random effects, 
Neuroimaging data, 
Orthogonal matrices, 
Tensor factorization



\clearpage\pagebreak\newpage
\pagenumbering{arabic}
\newlength{\gnat}
\setlength{\gnat}{25pt}
\baselineskip=\gnat

\section{Introduction}

High-resolution functional data over multi-dimensional compact grids are routinely collected in many modern applications, including in medical imaging \citep{Bi2021}, gait analysis \citep{pham2017tensor}, and climate studies \citep{li2020tensor}. 
In medical imaging alone, such data arise in various forms, including diffusion-weighted MRI (DW-MRI), functional MRI, positron emission tomography (PET), computed tomography, and ultrasound \citep{Gandy2011, Li2017, Sun2017, shi2023tensor}. 
Clinically, these data aid in identifying differences between patients and healthy controls. 
Additionally, longitudinal data sets can help pinpoint regions of differential changes over time, providing insights into disease progression and valuable translational benefits.

Despite their potential usefulness, 
the ultra-high-resolutions of such data sets, 
their complex functional nature, 
their group and individual-level heterogeneity 
pose daunting statistical challenges. 
This article introduces a flexible Bayesian orthogonal tensor factor mixed model for high-resolution multi-group 
multi-dimensional longitudinal functional data, addressing many outstanding challenges in a principled and efficient manner.

\paragraph{Motivation.} 
We are motivated particularly by a neuroimaging application based on DW-MRI. 
We run our analysis on two types of DW-MRI-extracted features using the recently completed Alzheimer's Disease Neuroimaging Initiative (ADNI-3) (2016–2022) \citep{weiner2017alzheimer} data. 
Unlike past ADNI phases, ADNI-3 data is multi-shell high angular resolution diffusion imaging (HARDI) which allows us to apply the neurite orientation dispersion and density imaging (NODDI) \citep{zhang2012noddi} model along with the commonly used diffusion tensor imaging (DTI) \citep{soares2013hitchhiker} to extract different types of the microstructural features. 
Specifically, we analyze four different outcomes, as detailed in Section~\ref{sec: adni analysis}.
This led to several interesting findings.
Our dataset for each outcome includes 264 images from 101 participants from 3 groups, 
\ul{each image containing over 250,000 voxels}. 
For Alzheimer's (AD), there is substantial subject-level heterogeneity in these outcomes.
Additionally, the observation times are highly variable across subjects.
The key scientific goals here are to quantify differences in both baseline and longitudinal trajectories caused by neurodegeneration at both group and subject levels.

\paragraph{Existing Methods.} 
Since the literature on longitudinal, functional, and imaging data analyses is extensive, 
\ul{we review here a selection of works, with an additional emphasis on Bayesian methods, 
which are most relevant to our proposed ideas.}

To our knowledge, there are three primary strategies for analyzing longitudinal functional data. 
The first is to run a multi-step analysis: 
A first-stage mass-univariate analysis followed by  
a second-stage smoothing of the univariate estimates \citep{cui2022fast}. 
Subject heterogeneity is still hard to incorporate efficiently. 
The second approach is to apply spatially varying coefficient models, which would broadly come under the longitudinal functional data analysis framework \citep{greven2010longitudinal,csenturk2010functional,montagna2012bayesian,zipunnikov2014longitudinal,park2015longitudinal,zhu2019fmem,shamshoian2022bayesian,kang2023joint} along with some spatio-temporal approaches \citep{gossl2001bayesian,woolrich2004fully,penny2005bayesian,bowman2007spatiotemporal,guo2008predicting,bowman2008bayesian,derado2010modeling,hyun2016stgp,zhang2016spatiotemporal,abi2020monotonic} 
{These approaches may suffer from one or more of the following limitations: 
Insufficiently flexible parametric assumptions; 
inability to accommodate irregularly spaced observation points in imbalanced study designs; 
multi-stage implementation, which may result in inefficient inference and uncertainty quantification; 
lack of scalability to handle ultra-high-resolution data; 
and inability to account for the inherently multi-dimensional structure of data.} 
A third line of research relies on tensor-factorization-based approaches \citep{tucker:1966,carroll1970analysis,kolda2009tensor} 
and is becoming popular for 
their ability to inherently handle multi-dimensional functional data \citep{zhou2013tensor}. 

Specifically, the canonical polyadic (CP) or parallel factor (PARAFAC) decomposition represents a tensor into a sum of rank-1 tensors. 
The Tucker decomposition generalizes this by introducing a core tensor that interacts with factor matrices, allowing for a more expressive representation.
Research to date has often focused on the simpler CP decomposition \citep{zhou2013tensor,zhao2015bayesian,khan2016bayesian,hore2016tensor,sun2017store,lock2018tensor,zhang2019tensor,spencer2020joint,guhaniyogi2021bayesian,zhang2021bayesian,liu2023joint,liu2023integrative,jiang2024bamita,luo2024bayesian}, including works on modeling longitudinal neuroimaging data \citep{zhang2019tensor,kundu2023bayesian,niyogi2024tensor,sort2024latent}. 
Smooth functional versions of the CP decomposition 
have also been introduced \citep{yokota2016smooth,wang2017generalized,han2024guaranteed,guan2024smooth}. 
Despite the promise of substantially greater model flexibility and compressibility, 
works employing the Tucker approach have been limited \citep{xu2013bayesian,li2018tucker,spencer2024bayesian,stolf2024bayesian}, 
likely due to their additional modeling and computational challenges 
even when used in their simplest form. 
Adapting the Tucker approach to build longitudinal functional mixed models in ultra-high-resolution settings is therefore both promising and daunting at the same time. 

\paragraph{Our Proposed Approach.} 
We present a novel longitudinal mixed model for multi-dimensional functional data 
that integrates Tucker tensor decomposition, basis function expansion, and latent factor representation to address the aforementioned challenges. 
We adopt a Bayesian approach to estimation and inference, which supports 
hierarchical model construction, information sharing between model components, prior-based model regularization, and finite-sample uncertainty assessment via Markov chain Monte Carlo (MCMC). 

To begin with, we collect the images from each individual 
and form an expanded tensor indexed, in addition to the original voxel axes, by their group label and time points, 
enabling us to characterize group-level differences and temporal dynamics in a unified framework. 
Next, we decompose these tensors into two additive components: One capturing the baseline at the beginning of the data collection period, 
the other characterizing the subsequent longitudinal changes. 
The time-varying component is further modeled using flexible mixtures of B-spline bases 
with 
tensor coefficients. 
To address the massive dimensionality challenges, 
we represent both the baseline and basis coefficient tensors using a reduced rank multiplicative mixed effects Tucker tensor decomposition model, comprising a random effects core tensor specific to each subject and dimension-associated mode matrices shared between them. 

We impose a semi-orthogonal structure for the mode matrices via a novel prior distribution. 
The prior results in an easy-to-sample-from posterior for the mode matrices. 
Crucially, semi-orthogonal mode matrices make the core tensor elements a-posteriori conditionally independent, 
which allows them to also be updated efficiently in parallel, 
overcoming the daunting computational challenges associated with them as well. 
Furthermore, the prior's simple design allows us to adapt it to a graph-Laplacian-operator-based product Gaussian framework, incorporating low-rank approximations and column-wise cumulative shrinkage, all while preserving its computational advantages. 
This construction captures spatio-temporal dependencies across voxels and time, promotes sparsity, and facilitates automatic, data-adaptive tensor rank selection.

Overall, we believe our model is built on a clear rationale for its different components: 
Low-rank Tucker tensor decomposition-based representations achieve dimension reduction, 
splines model smooth longitudinal evolution, while allowing for irregular subject-specific observation times, 
random effects distributions on the core tensors accommodate individual heterogeneity, 
graph Laplacian infused sparse Gaussian process models for the columns of the mode matrices capture spatial dependence while automating rank selection, and 
their semi-orthogonality ensures efficient scalable computation. 

Our approach is highly data-adaptive, requiring very few prior hyper-parameters deep inside the model hierarchy. 
Theoretically, we establish some useful distributional properties of our model, large support of our proposed priors, and posterior consistency in recovering the `true' parameters under mild regularity conditions.
Notably, integrating out the random effects from the proposed mixed effects mean model 
also leads to an interesting novel composite covariance-tensor model for longitudinal multi-dimensional functional data settings.




\paragraph{Key Statistical Contributions.}
%
(a) First, we advance the field using a Tucker tensor decomposition paradigm which, despite its modeling and computational demands compared to simpler CP models, offers greater flexibility and efficient parameter compression, but has received limited attention, particularly in the Bayesian context. 
(b) In doing so, we introduce a novel prior that enforces semi-orthogonality on the mode matrices, facilitating efficient posterior exploration even in ultra-high-resolution settings. 
(c) We further incorporate several realistic enhancements of this prior, 
including 
a smoothness-inducing low-rank basis decomposition with a graph-Laplacian for the imaging axes and time, and 
semi-automated rank selection via stochastic ordering of the columns. 
%
The combination of (a), (b), and (c) substantially improves estimation and inference results.  
(d) 
We also account for individual heterogeneity in both baseline and longitudinal trajectories within a principled multiplicative mixed model framework, differing from the typical focus on population-level effects in longitudinal functional data analyses.
(e) Specifically, we model these trajectories as flexible smooth functions over time using novel mixtures of B-splines with individual-specific coefficients satisfying identifiability constraints that separate them from the baselines. 
(f) Our approach also unifies models for tensor-valued means with 
models for covariance tensors within a cohesive hierarchical framework, inherently satisfying positive definiteness without additional adjustments. 
Specifically, we show that flexible Tucker decomposed models for covariance tensors arise naturally from Tucker factorized mean models with random effects core components with greater flexibility in multi-subject settings. 
(g) Finally, 
the flexibility and computational efficiency of our basic tensor decomposition-based model construction and implementation make them highly adaptable to a wide range of other (static and longitudinal; fixed and mixed model) settings involving tensors and orthogonal matrices. 

There exist some interesting works on priors for orthonormal matrices \citep{hoff2007model,hoff2009simulation,jauch2020random,north2024flexible,jauch2025prior} 
which may not be amenable to easy computation in complex hierarchical settings, requiring spatial smoothness, cumulative shrinkage for automatic rank selection, etc. 
Our prior relies on a Gram-Smidth-type decomposition, similar to \cite{north2024flexible}, 
but avoids orthonormality-related unit-norm projections, 
allowing a more direct use of subspace-constrained multivariate distributions, 
facilitating further hierarchical refinements while maintaining easy posterior computation.


\paragraph{Outline of the Paper.}
Section \ref{sec: tensor factorization} reviews some tensor basics; 
Section \ref{sec: models} discusses the main models; 
Section \ref{sec: priors} introduces the new priors; 
Section \ref{sec: post comp} provides an overview of the posterior computation;
Section \ref{sec: asymptotics} presents asymptotic results; 
Section \ref{sec: adni analysis} discusses the results of the ADNI-3 analysis; 
Section \ref{sec: sim study} illustrates results of simulation experiments;
Section \ref{sec: discussion} concludes with a discussion. 
%
Substantive details, 
including additional information about ADNI-3, 
hyper-parameter selection, 
posterior computation, 
and auxiliary results and proofs of theoretical results, 
are presented in the Supplementary Materials (SM from here on).

\vspace*{-3ex}
\section{Tensor Factorization Basics} \label{sec: tensor factorization}
\vspace*{-1ex}

In this section, we provide a brief review of the different main types of tensor factorizations \citep{hitchcock1927expression,tucker:1966,de_lathauwer_etal:2000,kolda2009tensor}. 


A $d_{1} \times \dots \times d_{p}$ dimensional tensor $\btheta = \{\theta_{h_{1},\dots,h_{p}}: h_{j}=1,\dots,d_{j}, j=1,\dots,p\}$ 
admits a canonical polyadic (CP) or parallel factor (PARAFAC) decomposition 
with rank $r$ (Figure \ref{fig: CP}) 
if its elements can be written (not necessarily uniquely) as 
\vspace*{-7ex}\\
\be
\textstyle\theta_{h_{1},\dots,h_{p}} = \sum_{z=1}^{r} \eta_{z}\prod_{j=1}^{p} a_{z}^{(j)}(h_{j})~~~\text{for each}~(h_{1},\dots,h_{p}), \label{eq: parafac}
\ee
\vspace*{-7ex}\\
where $\ba_{z}^{(j)} = [a_{z}^{(j)}(1),\dots,a_{z}^{(j)}(d_{j})]\trans, z=1,\dots, r, j=1,\dots,p$ are $d_{j}$ dimensional vectors. 
The effective model size here reduces from $\prod_{j=1}^{p}d_{j}$ to $r\sum_{j=1}^{p}d_{j}$ after the factorization. 
The factorization can also be compactly represented as 
$\btheta = \sum_{z=1}^{r} \blambda \circ \ba_{z}^{(1)} \circ \dots \circ \ba_{z}^{(p)}$, 
where the symbol $\circ$ represents the vector outer product.

\begin{figure}[ht!]
	\centering
	\hspace*{-0.25cm}\includegraphics[width=0.60\linewidth, trim=1cm 1.25cm 1cm 1cm]{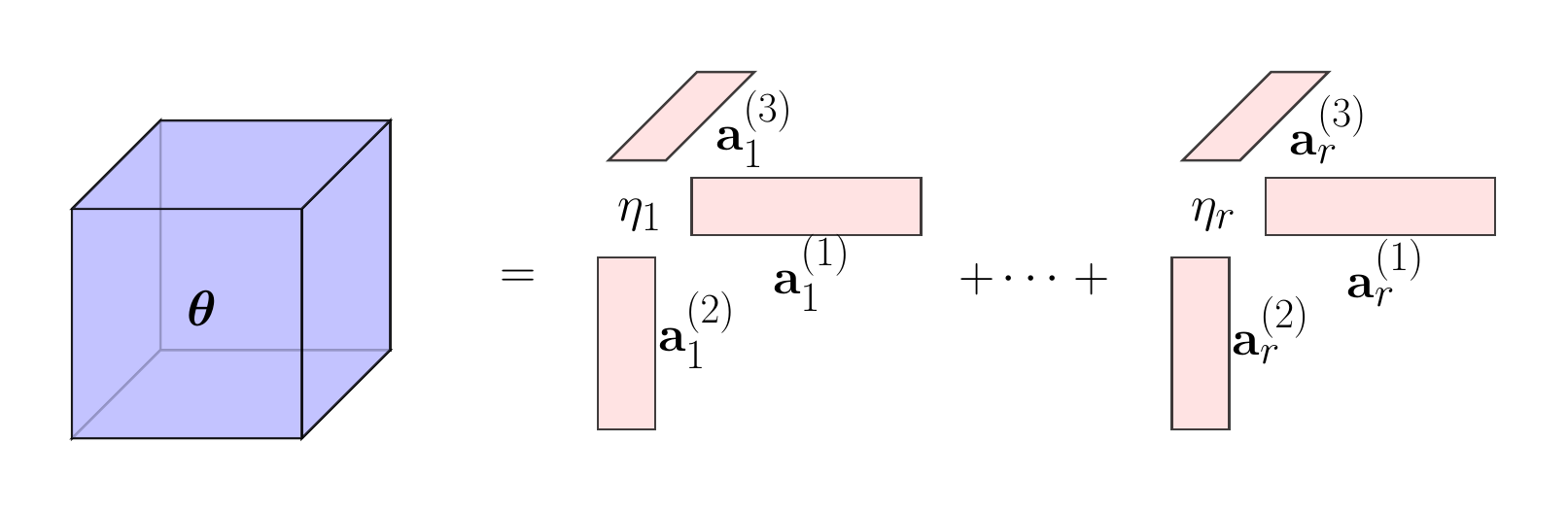}
	\caption{\baselineskip=10pt 
	Pictorial representation of CP decomposition of a three-dimensional tensor. 
	}
	\label{fig: CP}
\end{figure}

Alternatively, the tensor $\btheta$ admits a Tucker decomposition 
with multi-linear rank $(r_{1}, \dots,r_{p})$ (Figure \ref{fig: HOSVD}) 
if its elements can be represented (not necessarily uniquely) as 
\vspace*{-7ex}\\
\be
\textstyle\theta_{h_{1},\dots,h_{p}} = \sum_{z_{1}=1}^{r_{1}} \cdots\sum_{z_{p}=1}^{r_{p}} \eta_{z_{1},\dots,z_{p}} \prod_{j=1}^{p} a_{z_{j}}^{(j)}(h_{j})~~~\text{for each}~(h_{1},\dots,h_{p}), \label{eq: Tucker}
\ee
\vspace*{-7ex}\\
where $\etam = \{\eta_{z_{1},\dots,z_{p}}: z_{j}=1,\dots,r_{j}, j=1,\dots,p\}$ is a $r_{1} \times \dots \times r_{p}$ dimensional `core' tensor 
with $1 \leq r_{j} \leq d_{j}$ for each $j$, 
and $\bA^{(j)} = [\ba_{1}^{(j)},\dots,\ba_{r_{j}}^{(j)}]$, with $\ba_{z_{j}}^{(j)} = [a_{z_{j}}^{(j)}(1),\dots,a_{z_{j}}^{(j)}(d_{j})]\trans$, 
are $d_{j} \times r_{j}$ dimensional `mode' or `factor' matrices with full column rank $r_{j}$. 
The effective size of the model after the factorization is now $\prod_{j=1}^{p} r_{j} + \sum_{j=1}^{p}r_{j}d_{j} \approx \prod_{j=1}^{p} r_{j}$. 
A significant reduction in dimensions is therefore achieved by the decomposition when $\prod_{j=1}^{p} r_{j} \ll \prod_{j=1}^{p} d_{j}$, 
i.e., when the core $\etam$ is much smaller in size than the original tensor $\btheta$. 
The factorization can also be compactly represented as $\btheta = \etam \times_{1} \bA^{(1)} \dots \times_{p}\bA^{(p)}$, 
where $\etam \times_{j} \bA^{(j)}$ represents the $j$-mode tensor product such that $\etam \times_{j} \bA^{(j)}$ is an $r_{1} \times r_{j-1} \times d_{j} \times r_{j+1} \times \dots \times r_{p}$ dimensional tensor with elements $(\etam \times_{j} \bA^{(j)})_{z_{1},\dots,z_{j-1},\ell,z_{j+1},\dots,z_{p}} = \sum_{z_{j}=1}^{r_{j}} \eta_{z_{1},\dots,z_{p}} a_{z_{j}}^{(j)}(\ell)$.

\vskip -5pt
\begin{figure}[ht!]
	\centering
	\hspace*{-0.25cm}\includegraphics[width=0.50\linewidth, trim=1cm 1.25cm 1cm 1cm]{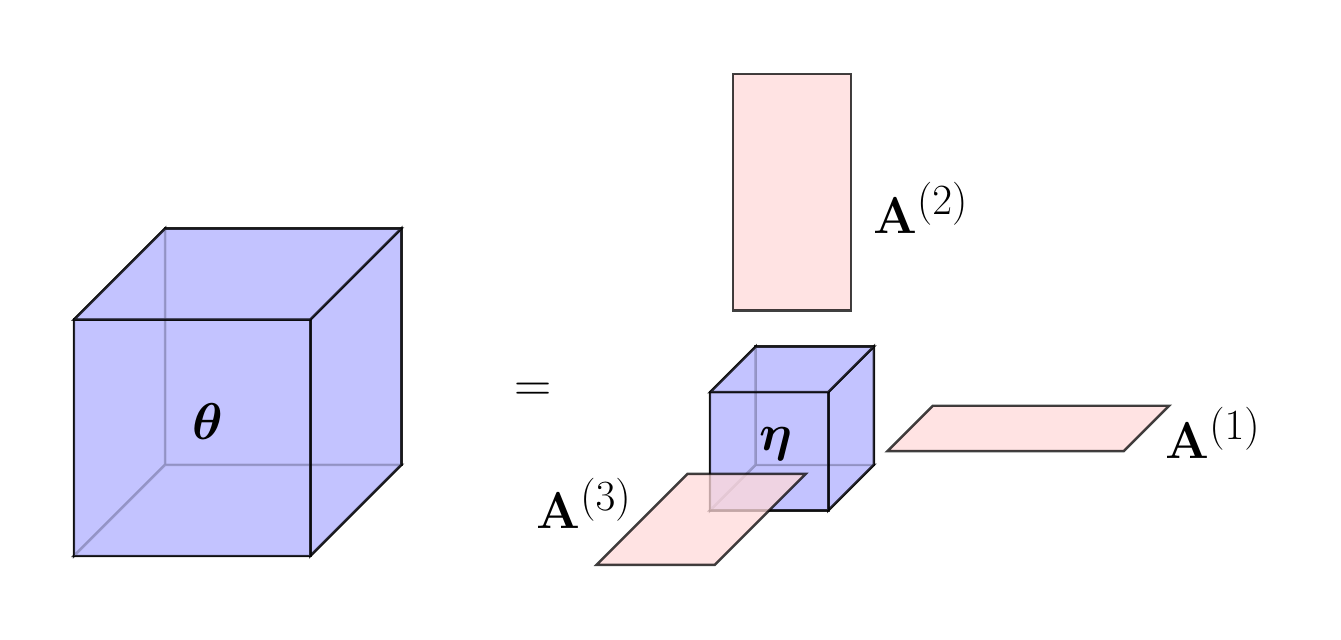}
	\caption{\baselineskip=10pt 
	Pictorial representation of HOSVD of a three-dimensional tensor. 
	}
	\label{fig: HOSVD}
\end{figure}
\vskip -10pt

The CP representation \eqref{eq: parafac} is obtained as a special case of the Tucker decomposition with 
a diagonal core $\eta_{z_{1},\dots,z_{p}}  = \eta_{z} 1\{z_{1} = \dots = z_{p}=z\}$ 
with $r_{1}=\dots=r_{p}=r$, and 
$z_{j}=z$ and $a_{z_{j}}^{(j)}(h_{j}) = a_{z}^{(j)}(h_{j})$ for all $j, h_{j}, z$. 
Conversely, the Tucker model in \eqref{eq: Tucker} has an equivalent CP presentation as 
$\theta_{h_{1},\dots,h_{p}} = \sum_{z=1}^{r} \vec(\etam)_{z(z_{1},\dots,z_{p})} \prod_{j=1}^{p}a_{z(z_{1},\dots,z_{p})}^{(j)}(h_{j})$, 
where $\vec(\etam)$ is the $r \times 1$ vectorized version of $\etam$ with $r=\prod_{j=1}^{p}r_{j}$, and 
$z(z_{1},\dots,z_{p})$ indexes the location of $\eta_{z_{1},\dots,z_{p}}$ in $\vec(\etam)$, 
and $a_{z(z_{1},\dots,z_{p})}^{(j)}(h_{j}) = a_{z_{j}}^{(j)}(h_{j})$ for all $j, h_{j}, z_{j}$ with many repeated entries. 

The compact higher order singular value decomposition (compact HOSVD) of a tensor 
is a special case of the Tucker decomposition, 
where the mode matrices $\bA^{(j)}$'s are all semi-orthogonal, 
i.e., $\bA^{(j)\hbox{\scriptsize T}} \bA^{(j)}$ is an $r_{j} \times r_{j}$ diagonal matrix for all $j$,  
resulting in a more interpretable form. 
%
The following simple but consequential result 
shows that the compact HOSVD is just as flexible as Tucker. 
The result is crucial in developing efficient computational strategies for our proposed approach, as described in Section \ref{sec: priors} below. 
The proof is straightforward and is presented in Section \ref{sec: sm proofs} in the Supplementary Materials. 

\begin{Lem} \label{lem: Tucker-HOSVD}
Every Tucker decomposition admits an equivalent compact HOSVD.
\end{Lem}

For a more comprehensive review of these methods and their adaptations and applications in statistics, especially in a Bayesian setting, see \cite{guhaniyogi2020bayesian,shi2023tensor}. 

\vspace*{-3ex}
\section{Longitudinal Multi-Group Tensor Factor Model} \label{sec: models}
\vspace*{-1ex}
For $h_{g}\in \{1,\dots,d_{g}\}$, $i \in \{i_{h_{g},1},\dots,i_{h_{g},N_{h_{g}}}\}$, $t \in \{t_{i,1},\dots,t_{i,n_{i}}\}$ and $h_{s}\in \{1,\dots,d_{s}\},s=1,2,3$, 
let $Y^{(i)}(h_{g},h_{1},h_{2},h_{3},t)$ denote the image value at the $(h_{1},h_{2},h_{3})\th$ voxel for the $i\th$ individual from the $h_{g}\th$ group at time point $t$. 
While $h_{g}=h_{g}(i)$, since every subject belongs to a unique group, 
throughout we keep the $i$ implicit in $h_{g}$ for simplicity. 
Let $N=\sum_{h_{g}=1}^{d_{g}}N_{h_{g}}$ be the total number of subjects. 
Also, let $\S = \{1,2,3\}, \S_{g}=\{g,1,2,3\}$ and $\S_{g,t}=\{g,1,2,3,t\}$.

To flexibly model the image trajectories, 
while accommodating spatial dependence between adjacent voxel locations, systematic differences between groups, and individual heterogeneity within groups, 
we let 
\vspace*{-6ex}\\
\be \label{eq: mainmodel}
\begin{aligned}
& Y^{(i)}(h_{g},h_{1},h_{2},h_{3},t)=\alpha^{(i)}(h_{g},h_{1},h_{2},h_{3})+\beta^{(i)}(h_{g},h_{1},h_{2},h_{3},t)+\epsilon^{(i)}(h_{g},h_{1},h_{2},h_{3},t),\\
& \epsilon^{(i)}(h_{g},h_{1},h_{2},h_{3},t) \simiid \Normal(0,\sigma_{\epsilon}^{2}),
\end{aligned}
\ee
\vspace*{-5ex}\\
where $\alpha^{(i)}(h_{g},h_{1},h_{2},h_{3})$ are the group-subject-and-voxel-specific baseline, 
$\beta^{(i)}(h_{g},h_{1},h_{2},h_{3},t)$ are the group-subject-and-voxel-specific time-varying deviations from the baseline, 
and $\epsilon^{(i)}(h_{g},h_{1},h_{2},h_{3},t)$ are random errors.
With identifiability constraints discussed later, \ul{the proposed separation of the baseline and time-varying deviations allows capturing their distinct spatial sparsity patterns}.

\paragraph{Baseline Tensors $\balpha$.} 
To flexibly characterize the baseline tensors, we propose the following reduced rank Tucker decomposed multiplicative mixed model 
\vspace*{-4ex}\\
\be \label{eq: baseline effects}
\begin{aligned}
& \textstyle\alpha^{(i)}(h_{g},h_{1},h_{2},h_{3}) = 
\sum_{m \in \S_{g}}\sum_{z_{m}=1}^{r_{\alpha,m}} \eta_{\alpha,z_{g},z_{1},z_{2},z_{3}}^{(i)} \prod_{s \in \S_{g}} a^{(s)}_{\alpha,z_{s}}(h_{s}), \\ 
& \eta_{\alpha,z_{g},z_{1},z_{2},z_{3}}^{(i)} \simind \Normal(c_{\alpha,z_{g},z_{1},z_{2},z_{3}}, \tau^{2}_{\alpha}\sigma_{\alpha,z_{g},z_{1},z_{2},z_{3}}^{2}),
\end{aligned}
\ee
\vspace*{-4ex}\\
where $\etam_{\alpha}^{(i)} = ((\eta_{\alpha,z_{g},z_{1},z_{2},z_{3}}^{(i)}))$ is an $r_{\alpha,g} \times r_{\alpha,1} \times r_{\alpha,2}\times r_{\alpha,3}$ subject-specific random effects core tensor; 
$\bA_{\alpha}^{(g)} = ((a^{(g)}_{\alpha,z_{g}}(h_{g})))$ is an $ d_{g} \times r_{\alpha,g}$ mode matrix associated with the group labels, 
and $\bA_{\alpha}^{(s)} = ((a^{(s)}_{\alpha,z_{s}}(h_{s}))), s\in \S$ are $d_{s} \times r_{\alpha,s}$ mode matrices associated with the voxel directions. 
Integrating out the random effects, 
we arrive at a Tucker decomposed subject-averaged mean tensor, denoted $\alpha(h_{g},h_{1},h_{2},h_{3})$, as   
\vspace*{-7ex}\\
\be \label{eq: mean baseline effects}
\textstyle \alpha(h_{g},h_{1},h_{2},h_{3}) = \eE_{\eta_{\alpha}}\{\alpha^{(i)}(h_{g},h_{1},h_{2},h_{3})\} = \sum_{m \in \S_{g}}\sum_{z_{m}=1}^{r_{\alpha,m}} c_{\alpha,z_{g},z_{1},z_{2},z_{3}} \prod_{s \in \S_{g}} a^{(s)}_{\alpha,z_{s}}(h_{s}), 
\ee
\vspace*{-7ex}\\
where $\bC_{\alpha}=((c_{\alpha,z_{g},z_{1},z_{2},z_{3}}))$ is an $r_{\alpha,g} \times r_{\alpha,1} \times r_{\alpha,2} \times r_{\alpha,3}$ core tensor. 

\paragraph{Time-Varying Tensors $\bbeta$.} 
Additional challenges arise here from the time-varying nature of the deviation tensors, 
the fact that data are observed on individual specific irregular grids, 
and that they need to satisfy some constraints to be separately identifiable from the baselines.
To address these issues, we first represent these deviations as flexible mixtures of basis functions as  
\vspace*{-7ex}\\
\be
\textstyle \beta^{(i)}(h_{g},h_{1},h_{2},h_{3},t) = \sum_{h_{t}=1}^{d_{t}} \wt\beta^{(i)}(h_{g},h_{1},h_{2},h_{3},h_{t}) b_{q,h_{t}}(t), 
\ee
\vspace*{-6ex}\\
where $b_{q,j}(t)$ are 
B-spline bases of degree $q$
defined at the end of this subsection below. 
%

We then propose a Tucker decomposed mixed effects model for the basis coefficients as 
\vspace*{-6ex}\\
\be 
\begin{aligned}\label{eq: regression coeffs}
& \textstyle \wt\beta^{(i)}(h_{g},h_{1},h_{2},h_{3},h_{t}) = \sum_{m \in \S_{g,t}}\sum_{z_{m}=1}^{r_{\beta,m}} \eta_{\beta,z_{g},z_{1},z_{2},z_{3},z_{t}}^{(i)} \prod_{s \in \S_{g,t}} a_{\beta,z_{s}}^{(s)}(h_{s}),\\
%
&\eta_{\beta,z_{g},z_{1},z_{2},z_{3},z_{t}}^{(i)} \simind \Normal(c_{\beta,z_{g},z_{1},z_{2},z_{3},z_{t}},\tau_{\beta}^{2} \sigma_{\beta,z_{g},z_{1},z_{2},z_{3},z_{t}}^{2}),
\end{aligned}
\ee
\vspace*{-4ex}\\
where $\etam_{\beta}^{(i)} = ((\eta_{\beta,z_{g},z_{1},z_{2},z_{3},z_{t}}^{(i)}))$ is an $r_{\beta,g} \times r_{\beta,1} \times r_{\beta,2}\times r_{\beta,3} \times r_{\beta,t}$ subject-specific random effects core tensor; 
$\bA_{\beta}^{(g)} = ((a_{\beta,z_{g}}^{(g)}(h_{g})))$ is a $d_{g} \times r_{\beta,g} $ mode matrix associated with the group labels, 
$\bA_{\beta}^{(s)} = ((a_{\beta,z_{s}}^{(s)}(h_{s})))$, $s\in \S$, are $d_{s} \times r_{\beta,s} $ mode matrices associated with the voxel directions, 
and $\bA_{\beta}^{(t)} = ((a_{\beta,z_{t}}^{(t)}(h_{t})))$ is a $d_{t} \times r_{\beta,t} $ mode matrix associated with the B-spline bases, 
all mode matrices again being semi-orthogonal. 
Integrating out the random effects 
as before, we arrive again at a compact representation for the subject-averaged mean basis coefficient tensors, denoted $\wt\beta(h_{g},h_{1},h_{2},h_{3},h_{t})$, as   
\vspace*{-6ex}\\
\be
\begin{aligned} \label{eq: mean regression coeffs}
& \textstyle \wt\beta(h_{g},h_{1},h_{2},h_{3},h_{t}) = \eE_{\eta_{\beta}}\{\wt\beta^{(i)}(h_{g},h_{1},h_{2},h_{3},h_{t})\} \\
& ~~~~~ \textstyle = \sum_{m \in \S_{g,t}}\sum_{z_{m}=1}^{r_{\beta,m}} c_{\beta,z_{g},z_{1},z_{2},z_{3},z_{t}} \prod_{s \in \S_{g,t}} a_{\beta,z_{s}}^{(s)}(h_{s}),  
\end{aligned} \label{eq: mean regression coeffs}
\ee
\vspace*{-4ex}\\
where $\bC_{\beta}=((c_{\beta,z_{g},z_{1},z_{2},z_{3},z_{t}}))$ is an $r_{\beta,g} \times r_{\beta,1} \times r_{\beta,2}\times r_{\beta,3} \times r_{\beta,t}$ mean core tensor, leading to
\vspace*{-7ex}\\
\be
\textstyle \beta(h_{g},h_{1},h_{2},h_{3},t) = \sum_{h_{t}=1}^{d_{t}} \wt\beta(h_{g},h_{1},h_{2},h_{3},h_{t}) b_{q,h_{t}}(t). 
\ee

\vskip -15pt 
\begin{figure}[!h] 
\centering\includegraphics[trim=1.1cm 1.1cm 1.0cm 0cm, clip, scale=0.6]{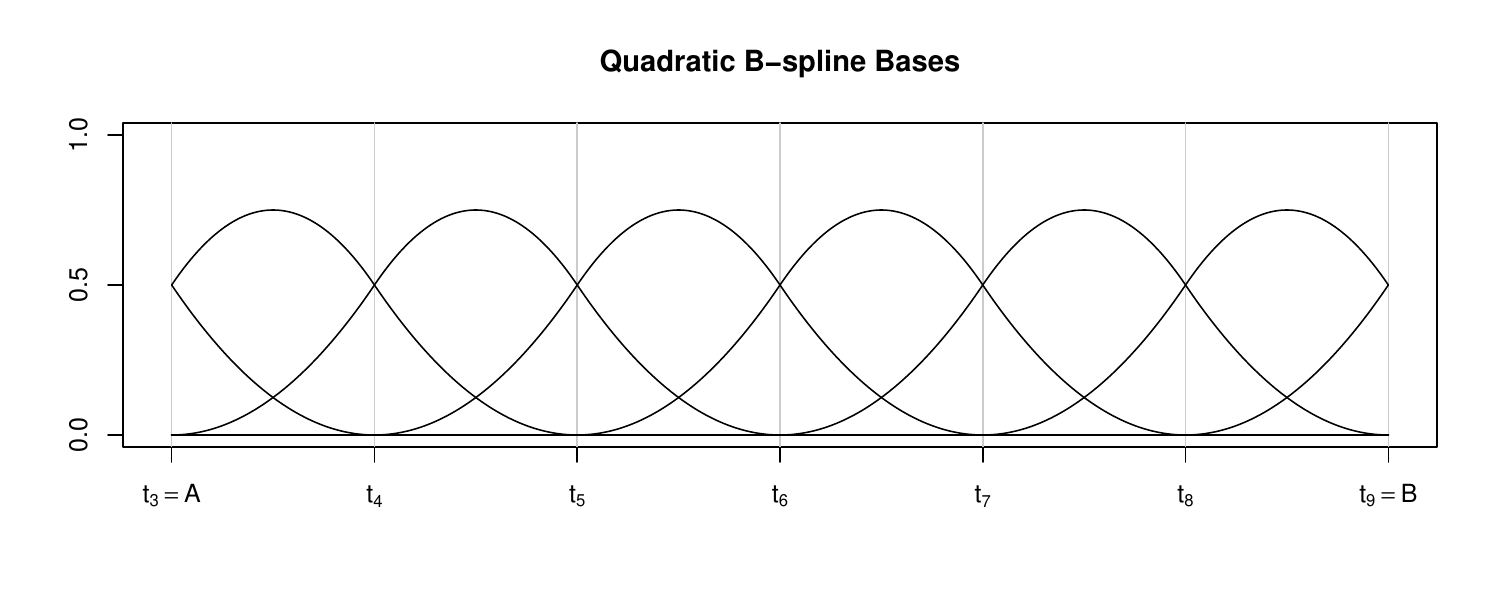}
\vspace*{-5pt}
\caption{Quadratic B-splines defined with $11$ equidistant knot-points partitioning the interval $[A,B]$ into $6$ equal sub-intervals.} \label{fig: Quadratic B-splines}
\vspace*{-5pt}
\end{figure}

\paragraph{Separation of $\balpha$ and $\bbeta$.} 
Next, to separately identify the $\alpha^{(i)}(h_{g},h_{1},h_{2},h_{3})$'s as the baseline effects, 
we set 
$\beta^{(i)}(h_{g},h_{1},h_{2},h_{3},0)=0 ~~\text{for all}~~ (i,h_{g},h_{1},h_{2},h_{3})$. 
We exploit local support and boundary properties of B-spline bases to achieve this. 
First, we define these bases by partitioning $[A,B]$ into $k_{t}$ sub-intervals using the knot points
$t_{1}=\dots=t_{q+1} = A < t_{q+2} < t_{q+3} < \dots < t_{q+k_{t}} < t_{q+k_{t}+1} = \dots = t_{2q+k_{t}+1} = B$.
Using these knot points, $(q+k_{t})=d_{t}$ bases of degree $q$, 
namely $\{b_{q,1},b_{q,2},\dots,b_{q,d_{t}}\}$, can be defined through the recursion relation given on page 90 of de Boor (2000). 
%
%
Note that, for any $t$,
$b_{q,h_{t}}$ is positive only inside the interval $[t_{h_{t}},t_{h_{t}+q+1}]$.

For our purposes $[A,B]=[0,1]$ and we use quadratic B-splines (i.e., $q=2$) with equidistant knot points (Figure \ref{fig: Quadratic B-splines}). 
This implies that only the first two bases, namely $b_{2,1}(t)$ and $b_{2,2}(t)$, are active at $t=0$ with $b_{2,1}(0) = b_{2,2}(0) = 0.5$. 
Hence, $\beta^{(i)}(h_{g},h_{1},h_{2},h_{3},0) = 0.5 \{\wt\beta^{(i)}(h_{g},h_{1},h_{2},h_{3},1)+\wt\beta^{(i)}(h_{g},h_{1},h_{2},h_{3},2)\}$. 
The condition $\beta^{(i)}(h_{g},h_{1},h_{2},h_{3},0) = 0$ can therefore be imposed by setting $a_{\beta,z_{t}}^{(t)}(1)+a_{\beta,z_{t}}^{(t)}(2)=0$ 
which is equivalent to reparametrizing $\bA_{\beta}^{(t)}\bM = \bA_{\beta,1}^{(t)}$, 
where $\bA_{\beta,1}^{(t)}$ is an $(d_{t}-1) \times r_{\beta,t}$ matrix and $\bM$ is $d_{t} \times (d_{t}-1)$ dimensional such that $\bM_{-1,-1}=\bI_{d_t-1}$ and $\bM_{1,1}=-1$.

\vskip 10pt
Here we restrict the mode matrices $\bA_{\alpha}^{(s)}, s \in \S_{g}$ and $\bA_{\beta}^{(s)}, s \in \S_{g,t}$ to be semi-orthogonal, 
ensuring the decompositions in \eqref{eq: baseline effects}, \eqref{eq: mean baseline effects}, \eqref{eq: regression coeffs} and \eqref{eq: mean regression coeffs} to all be compact HOSVD (Lemma \ref{lem: Tucker-HOSVD}). 
Additionally, in \eqref{eq: baseline effects} and \eqref{eq: regression coeffs}, we allowed individual heterogeneity through random-effects core tensors. 
Below, we discuss these modeling choices' motivations, implications, and challenges. 
The semi-orthogonality of the mode matrices with piece-wise smooth functional characteristics is achieved via a novel prior discussed separately in Section \ref{sec: priors} below.

\vspace*{-2ex}
\subsection{Semi-orthogonal Mode Matrices} \label{sec: semi-orthogonality}
\vspace*{-1ex}
Semi-orthogonal mode matrices endow our model with two key distributional properties, which we discuss below. 
Let $\bY^{(i)}(h_{g},\cdot,\cdot,\cdot,t)$ denote the $d_{1} \times d_{2} \times d_{3}$ dimensional image for the $i\th$ individual from the $h_{g}\th$ group at time $t$, 
and $\bB^{(i)}=\{b_{q,h_{t}}(t)\}_{t_{i,1}\leq t\leq t_{i,n_{i}}, 1\leq h_{t}\leq d_{t}}$ be the $n_{i} \times d_{t}$ B-spline basis matrix for subject $i$. 
Also, let $(\bA_{\alpha}^{(s)})\trans\bA_{\alpha}^{(s)} = \bDelta_{\alpha}^{(s)} = \diag[\delta_{\alpha,1}^{(s)},\dots,\delta_{\alpha,r_{\alpha,j}}^{(s)}]$ for $s \in \S$ and $\wt{\etam}^{(i)}_{\alpha} = ((\wt{\eta}^{(i)}_{\alpha}(h_{g},z_{1},z_{2},z_{3}))) ={\etam}^{(i)}_{\alpha}\times_{g} \bA_{\alpha}^{(g)}$; 
and similarly, $(\bA_{\beta}^{(s)})\trans\bA_{\beta}^{(s)} = \bDelta_{\beta}^{(s)} = \diag[\delta_{\beta,1}^{(s)},\dots,\delta_{\beta,r_{\beta,j}}^{(s)}]$ for $s \in \S$ and $\wt{\etam}^{(i)}_{\beta} = ((\wt{\eta}^{(i)}_{\beta}(h_{g},z_{1},z_{2},z_{3},t_{ij}))) = {\etam}^{(i)}_{\beta}\times_{g} \bA_{\beta}^{(g)}\times_t (\bB^{(i)}\bA^{(t)}_{\beta})$.
Finally, define the $r_{\alpha,1} \times r_{\alpha,2} \times r_{\alpha,3}$ and $r_{\beta,1} \times r_{\beta,2} \times r_{\beta,3}$ dimensional tensors 
\vspace*{-7ex}\\
\bse
& \bR^{(i)}(h_{g},\cdot,\cdot,\cdot,0) = \bY^{(i)}(h_{g},\cdot,\cdot,\cdot,0)\times_{1} (\bA_{\alpha}^{(1)})\trans\times_{2} (\bA_{\alpha}^{(2)})\trans\times_{3} (\bA_{\alpha}^{(3)})\trans, ~~~\text{and}\\ 
& \bR^{(i)}(h_{g},\cdot,\cdot,\cdot,t_{ij}) = \{\bY^{(i)}(h_{g},\cdot,\cdot,\cdot,t_{ij})-\bR^{(i)}(h_{g},\cdot,\cdot,\cdot,0)\}\times_{1} (\bA_{\beta}^{(1)})\trans\times_{2} (\bA_{\beta}^{(2)})\trans\times_{3} (\bA_{\beta}^{(3)})\trans.
\ese
\vspace*{-7ex}

\begin{Lem} \label{lem: distprop alpha}
$R^{(i)}(h_{g},z_{1},z_{2},z_{3},0) \overset{ind}{\sim} \Normal\left\{\wt{\eta}_{\alpha}^{(i)}(h_{g},z_{1},z_{2},z_{3}) \prod_{s\in \S}\delta_{\alpha,z_{s}}^{(s)}, \sigma^{2}_{\epsilon}\prod_{s\in \S}\delta_{\alpha,z_{s}}^{(s)}\right\}$ across $(z_{1},z_{2},z_{3})$. 
\end{Lem}

\begin{Lem} \label{lem: distprop beta}
$R^{(i)}(h_{g},z_{1},z_{2},z_{3},t_{ij}) \overset{ind}{\sim} \Normal\left\{\wt{\eta}_{\beta}^{(i)}(h_{g},z_{1},z_{2},z_{3},t_{ij}) \prod_{s\in \S}\delta_{\beta,z_{s}}^{(s)}, \sigma^{2}_{\epsilon}\prod_{s\in \S}\delta_{\beta,z_{s}}^{(s)}\right\}$ across $(z_{1},z_{2},z_{3})$. 
\end{Lem}

\paragraph{Principal Tensor Components.} 
These results demonstrate that, akin to SVD or PCA for independently distributed vector-valued 
data, semi-orthogonal mode matrices in our proposed compact HOSVD model enable decorrelation of the baseline and the longitudinal changes for spatially and spatiotemporally correlated tensor-valued data.
{However, unlike standard SVD or PCA, 
the columns of the mode matrices now 
have varying (rather than unit) lengths.} 
Rank selection is now achieved through a complementary strategy of ordering these columns with stochastically diminishing variances, 
as described in Section \ref{sec: priors}. 
The $\bR^{(i)}(h_{g},\cdot,\cdot,\cdot,0)$'s and $\bR^{(i)}(h_{g},\cdot,\cdot,\cdot,t_{ij})$'s are thus broadly interpretable as \emph{principal tensor components} in the lower-dimensional tensor space obtained by transforming the observed images $\bY^{(i)}(h_{g},\cdot,\cdot,\cdot,0)$ and $\bY^{(i)}(h_{g},\cdot,\cdot,\cdot,t_{ij})$. 
This construction is different from other approaches such as \cite{allen2012sparse} due to our random core model which is more along the lines of probabilistic PCA \citep{tipping1999probabilistic}.


\paragraph{Challenges.} 
Significant challenges remain in defining appropriate prior distributions that enforce semi-orthogonality, simultaneously inducing distinct smoothness properties for different domains (groups, voxels, and time), 
and efficiently sampling the mode matrices from the posterior. 
We tackle these issues in Section \ref{sec: priors} and Section \ref{sec: sm post comp} in the SM.

\paragraph{Benefits.} 
\ul{Despite their added challenges, 
semi-orthogonal mode matrices actually markedly reduce the overall computational burden by facilitating efficient sampling of the individual-specific core tensors $\etam_{\alpha}^{(i)}$ and $\etam_{\beta}^{(i)}$, which contribute the majority of model parameters.} 
Exploiting 
the Lemmas \ref{lem: distprop alpha} and \ref{lem: distprop beta}, 
we enable fast and efficient parallel sampling of the elements of $\etam_{\alpha}^{(i)}$ and $\etam_{\beta}^{(i)}$, as detailed in Section \ref{sec: sm post comp} in the SM.


\vspace*{-2ex}
\subsection{Random Effects Core Tensors} \label{sec: random cores}
\vspace*{-1ex}
The literature on tensor methods with random effects components that accommodate subject heterogeneity is sparse. 
Existing approaches have considered structuring the data with the subjects constituting an additional tensor dimension, followed by a decomposition with a low-dimensional mode-matrix for the subjects \citep{hore2016tensor,guan2024smooth}. 
The next remark shows the connection between this strategy and our proposed subject-specific core tensor model, 
highlighting the greater flexibility the latter approach offers.

\begin{Rem}
We can stack the $\alpha^{(i)}(h_{g},h_{1},h_{2},h_{3})$'s from \eqref{eq: mainmodel} to form an $N\times d_{g} \times d_{1}\times d_{2}\times d_{3}$ tensor which can then be factorized as $\etam^{(0)}_{\alpha}\times_{N} \bA^{(N)}_{\alpha}\times_{g}\bA_{\alpha}^{(g)}\times_{1}\bA_{\alpha}^{(1)}\times_{2}\bA_{\alpha}^{(2)}\times_{3}\bA_{\alpha}^{(3)}$ with a mode matrix $\bA^{(N)}_{\alpha}$ with individual-specific rows. 
This can be viewed as a special case of our proposed model with individual-specific core tensors with $\etam_{\alpha}^{(N)}=\etam^{(0)}_{\alpha}\times_{N} \bA^{(N)}_{\alpha}$, 
where $\etam_{\alpha}^{(N)}$ is an $N\times r_{\alpha,g} \times r_{\alpha,1}\times r_{\alpha,2}\times r_{\alpha,3}$ tensor obtained by stacking the $\eta_{\alpha,z_{g},z_{1},z_{2},z_{3}}^{(i)}$'s from \eqref{eq: baseline effects}. 
A similar result holds for the $\bbeta^{(i)}$'s as well. 
\label{rem:decomposubmode}
\end{Rem}


{Furthermore, decomposing the subject-specific core-tensor as $\eta^{(i)}_{\alpha,z_{g},z_{1},z_{2},z_{3}}=c_{\alpha,z_{g},z_{1},z_{2},z_{3}} + \eta^{(i)}_{\alpha,diff,z_{g},z_{1},z_{2},z_{3}}$, where $\eta^{(i)}_{\alpha,diff,z_{g},z_{1},z_{2},z_{3}}\sim\Normal(0, \sigma_{\alpha,z_{g},z_{1},z_{2},z_{3}})$, we can get the population-level baseline and subject-specific differences as $\bC_{\alpha}\times_{g}\bA_{\alpha}^{(g)}\times_{1}\bA_{\alpha}^{(1)}\times_{2}\bA_{\alpha}^{(2)}\times_{3}\bA_{\alpha}^{(3)}$ and $\etam_{\alpha,diff}^{(i)}\times_{g}\bA_{\alpha}^{(g)}\times_{1}\bA_{\alpha}^{(1)}\times_{2}\bA_{\alpha}^{(2)}\times_{3}\bA_{\alpha}^{(3)}$, respectively. While both models share the same mode matrices, the following Lemma ensures that the construction is still highly flexible. 

\begin{Lem}
If the mode matrices of the closest Tucker approximations of the population-level baseline and the associated individual-specific differences under a given set of rank tuple are in the column spaces of a set of common mode matrices, 
the random-effects core model  \eqref{eq: baseline effects} can flexibly represent both the population-level and individual-specific difference tensors.
A similar result also holds for the $\wt\bbeta^{(i)}$'s.
\label{lem:indicompo}
\end{Lem}

The proof is straightforward, relying on the fact that  $\bC_{\alpha}$ and $\etam_{\alpha,diff}^{(i)}$ can flexibly down-weigh any irrelevant column from the shared mode matrices 
and hence omitted. 
On a related note, the special case discussed in Remark~\ref{rem:decomposubmode} with  $\etam_{\alpha}^{(N)}=\etam^{(0)}_{\alpha}\times_{N} \bA^{(N)}_{\alpha}$ implies that the population-level baseline and subject-specific differences are $\etam^{(0)}_{\alpha}\times_{N} \bmu_{\alpha,a}\times_{g}\bA_{\alpha}^{(g)}\times_{1}\bA_{\alpha}^{(1)}\times_{2}\bA_{\alpha}^{(2)}\times_{3}\bA_{\alpha}^{(3)}$ and $\etam^{(0)}_{\alpha}\times_{N} (\ba^{(N)}_{\alpha,i}-\bmu_{\alpha,a})\times_{g}\bA_{\alpha}^{(g)}\times_{1}\bA_{\alpha}^{(1)}\times_{2}\bA_{\alpha}^{(2)}\times_{3}\bA_{\alpha}^{(3)}$, respectively, 
where the rows $\ba^{(N)}_{\alpha,i}$ of $\bA^{(N)}_{\alpha}$ are assumed to be random with a common mean $\bmu_{\alpha,a}$.
This is again clearly less flexible in capturing complex population-level baseline and subject-specific differences, requiring stronger low-rank conditions to hold for both of these quantities.

\vspace*{-2ex} 
\subsection{Factor Model Representation}
\label{sec: factor model}
\vspace*{-1ex} 

With some abuse of notation, 
let  $\wt\bY^{(i)} (h_{g},t) = \vec\{\bY^{(i)}(h_{g},\cdot,\cdot,\cdot,t)\}$. 
Let the Kronecker product of two matrices be denoted by $\otimes$. 
We can write the model for vectorized images as a functional latent factor model as 
\vspace*{-7ex}\\
\be
\wt\bY^{(i)}(h_{g},t) = \bLambda_{\alpha}(h_{g}) \vec(\etam_{\alpha}^{(i)}) + \bLambda_{\beta}(h_{g},t) \vec(\etam_{\beta}^{(i)}) + \vec\{\bepsilon^{(i)}(h_{g},t)\}, \label{eq: latent factor}
\ee
\vspace*{-7ex}\\
where $\bLambda_{\alpha}(h_{g}) = \ba^{(g)}_{\alpha,h_{g}} \otimes \bA^{(1)}_{\alpha}\otimes \bA^{(2)}_{\alpha} \otimes \bA^{(3)}_{\alpha}$ and 
$\bLambda_{\beta}(h_{g},t) = \ba^{(g)}_{\beta,h_{g}} \otimes \bA^{(1)}_{\beta} \otimes \bA^{(2)}_{\beta} \otimes\bA^{(3)}_{\beta} \otimes \{\bA^{(t)}\bb(t)\}$, 
where $\ba^{(g)}_{\alpha,h_{g}}$ and $\ba^{(g)}_{\beta,h_{g}}$ are the $h_{g}\th$ column in $\bA_{\alpha}^{(g)}$ and $\bA_{\beta}^{(g)}$, respectively, 
with  
$\vec(\etam_{\alpha}^{(i)})\sim\Normal[\vec(\bC_{\alpha}),\diag\{\vec(\bSigma_{\alpha})\}]$ and 
$\vec(\etam_{\beta}^{(i)})\sim\Normal[\vec(\bC_{\beta}),
\diag\{\vec(\bSigma_{\beta})\}]$, 
where $\bSigma_{\alpha}=((\sigma^{2}_{\alpha,z_{g},z_{1},z_{2},z_{3}}))$, $\bSigma_{\beta}=((\sigma^{2}_{\beta,z_{g},z_{1},z_{2},z_{3},z_{t}}))$ are $r_{\alpha,g} \times r_{\alpha,1} \times r_{\alpha,2}\times r_{\alpha,3}$ and $r_{\beta,g} \times r_{\beta,1} \times r_{\beta,2}\times r_{\beta,3}\times r_{\beta,t}$ tensors. 
The formulation helps to derive some additional distributional properties of our model, which we utilize in our large sample analysis in Section~\ref{sec: asymptotics}.

Further marginalizing in \eqref{eq: latent factor}, we get $\wt{\bY}^{(i)}(h_{g},t)
\sim \MVN\{\bmu(h_{g},t),\bSigma(h_{g},t)\}$, where $\bmu(h_{g},t) = \bLambda_{\alpha}(h_{g})\vec(\bC_{\alpha})+\bLambda_{\beta}(h_{g},t)\vec(\bC_{\beta})$ and $\bSigma(h_{g},t) = \bLambda_{\alpha}(h_{g})\diag\{\vec(\bSigma_{\alpha})\}\bLambda_{\alpha}(h_{g})\trans + \bLambda_{\beta}(h_{g},t)\diag\{\vec(\bSigma_{\beta})\}\bLambda_{\beta}(h_{g},t)\trans + \diag(\sigma_{\epsilon}^{2}\bone_{d_{1}d_{2}d_{3}})$. 
Furthermore, if the observation times are the same for all subjects, each being observed a total of $n$ times, 
then for the $i\th$ subject in the $h_{g}\th$ group, the marginal density becomes $\vec(\wt{\bY}^{(i)})\sim \MVN (\bmu_{h_{g}},\bSigma_{h_{g}})$, where $\bmu_{h_{g}}=\bLambda_{\alpha,h_{g}}\vec(\bC_{\alpha})+\bLambda_{\beta,h_{g}}\vec(\bC_{\beta})$ and $\bSigma_{h_{g}}=\bLambda_{\alpha,h_{g}}\bSigma_{\alpha,\eta}\bLambda_{\alpha,h_{g}}^T+\bLambda_{\beta,h_{g}}\bSigma_{\beta,\eta}\bLambda_{\beta,h_{g}}^T+\sigma^{2}_{\epsilon}\bI_{nd_{1}d_{2}d_{3}}$.

\vspace*{-2ex} 
\subsection{Covariance Tensor Model Derivation} \label{sec: sm covaritensor}
\vspace*{-1ex} 
Integrating out the random effects $\etam_{\alpha}^{(i)}$ and $\etam_{\beta}^{(i)}$ from our proposed model, we also arrive at an interesting model for the covariance tensor. 
Specifically, for $h_{g} = h_{g} (i) = h_{g}^{\prime}(i) = h_{g}^{\prime}$ (by design) but arbitrary $(h_{1},h_{2},h_{3},t)$ and $(h_{1}^{\prime},h_{2}^{\prime},h_{3}^{\prime},t^{\prime})$, 
the marginal variance-covariance at baseline, 
conditional on 
$\bA_{\alpha},\bsigma_{\alpha}^{2},\tau_{\alpha}^{2}$, is given by 
\vspace*{-7ex}\\
\be
\begin{aligned}
& \textstyle \var\{\alpha^{(i)}(h_{g},h_{1},h_{2},h_{3})\} = \tau_{\alpha}^{2}\sum_{m \in \S_{g}}\sum_{z_{m}=1}^{r_{\alpha,m}} \sigma_{\alpha,z_{g},z_{1},z_{2},z_{3}}^{2} \left\{\prod_{s \in \S_{g}} a^{(s)}_{\alpha,z_{s}}(h_{s})\right\}^{2}, \\
&\cov\{\alpha^{(i)}(h_{g},h_{1},h_{2},h_{3}),\alpha^{(i)}(h_{g}^{\prime},h_{1}^{\prime},h_{2}^{\prime},h_{3}^{\prime})\} \\ 
&\textstyle~~~~~~~= \tau_{\alpha}^{2}\sum_{m \in \S_{g}}\sum_{z_{m}=1}^{r_{\alpha,m}} \sigma_{\alpha,z_{g},z_{1},z_{2},z_{3}}^{2} \left\{\prod_{s \in \S_{g}} a^{(s)}_{\alpha,z_{s}}(h_{s})\right\} \left\{\prod_{s \in \S_{g}} a^{(s)}_{\alpha,z_{s}}(h_{s}^{\prime})\right\}.
\end{aligned} \label{eq: cov tensor alpha}
\ee
\vspace*{-5ex}\\
Similarly, the marginal variance-covariance 
at time $t$ in addition to the baseline is given by 
\vspace*{-7ex}\\
\be
\begin{aligned}
& \textstyle \var\{\beta^{(i)}(h_{g},h_{1},h_{2},h_{3},t)\} = \sum_{h_{t}=1}^{d_{t}}\sum_{m \in \S_{g,t}}\sum_{z_{m}=1}^{r_{\beta,m}} \sigma_{\beta,z_{g},z_{1},z_{2},z_{3},z_{t}}^{2} \left\{ b_{h_{t}}(t) \prod_{j \in \S_{g,t}} a^{(s)}_{\mu,z_{s}}\right\}^{2} \\
&\cov\{\beta^{(i)}(h_{g},h_{1},h_{2},h_{3},t),\beta^{(i)}(h_{g}^{\prime},h_{1}^{\prime},h_{2}^{\prime},h_{3}^{\prime},t^{\prime})\} \\ 
&~~~~~~~= \tau_{\beta}^{2} \sum_{h_{t}=1}^{d_{t}}\sum_{h_{t}^{\prime}=1}^{d_{t}} \sum_{m \in \S_{g,t}}\sum_{z_{m}=1}^{r_{\alpha,m}} \sigma_{\beta,z_{g},z_{1},z_{2},z_{3},z_{t}}^{2} b_{h_{t}}(t) b_{h_{t}^{\prime}}(t') \left\{\prod_{s \in \S_{g,t}} a^{(s)}_{\beta,z_{s}}(h_{s})\right\} \left\{\prod_{s \in \S_{g,t}} a^{(s)}_{\beta,z_{s}}(h_{s}^{\prime})\right\}. 
\end{aligned} \label{eq: cov tensor beta}
\ee
\vspace*{-5ex}

The literature on higher-order moment tensors is extremely sparse. 
A CP decomposition-based model for correlation tensors for two-dimensional spatial domains in static settings 
was recently introduced in \cite{deng2023correlation}. 
Our proposal, while for covariance tensors, is substantially more general as it relies on the more versatile HOSVD decomposition and is designed for higher-dimensional spaces. 
Our model is also naturally {\it semi-symmetric}, 
i.e.,
\vspace*{-7ex}\\
\bse
\cov\{\alpha^{(i)}(h_{g},h_{1},h_{2},h_{3}),\alpha^{(i)}(h_{g}^{\prime},h_{1}^{\prime},h_{2}^{\prime},h_{3}^{\prime})\} = \cov\{\alpha^{(i)}(h_{g}^{\prime},h_{1}^{\prime},h_{2}^{\prime},h_{3}^{\prime}),\alpha^{(i)}(h_{g},h_{1},h_{2},h_{3})\}. 
\ese
\vspace*{-7ex}\\
Unlike the probabilistic Tucker model of \cite{chu2009probabilistic}, ours incorporates subject-specific cores and allows for non-identically distributed core-tensor elements.
Following Lemma~\ref{lem:indicompo}, varying variances control the magnitudes of $\etam_{\alpha,diff}^{(i)}$'s, thereby the subject-specific variations from the population average.
Jointly, \eqref{eq: cov tensor alpha}-\eqref{eq: cov tensor beta} thus leads to a very flexible tensor covariance decomposition model for longitudinal functional data.

For the rest of the article, to simplify notation, 
we will often use concise symbols without 
subscripts to represent pertinent sets of parameters and variables, e.g., $\alpha$ will denote $\alpha^{(i)}(h_{g},h_{1},h_{2},h_{3})$ and $\alpha(h_{g},h_{1},h_{2},h_{3})$ for different possible values of $(i,g,h_{1},h_{2},h_{3})$, etc.

\vspace*{-3ex}
\section{Prior Specification} \label{sec: priors}
\vspace*{-1ex}
\begin{enumerate}[wide, labelwidth=0pt, labelindent=0pt]
\item {\bf Mode Matrices:} 
Defining appropriate prior distributions for semi-orthogonal mode matrices, especially when their columns span general functional (discrete, spatial, and temporal) domains, presents a significant challenge. 
While there exists some literature for priors for orthonormal matrices 
\citep[e.g.,][]{hoff2007model,hoff2009simulation,jauch2020random,north2024flexible,jauch2025prior}, 
these methods, with the exception of \cite{north2024flexible}, 
are not amenable to easy refinements or efficient posterior computation. 
One key innovation of our proposed research is a cumulative subspace-constrained (CSC) product-independent Gaussian (PING) process that addresses the above-mentioned issues. 
Before we introduce this new prior, 
we first discuss the CSC distribution for general semi-orthogonal matrices, 
and then briefly review the PING process as originally introduced in \cite{roy2021spatial}. 

{\bf CSC Distribution for Semi-Orthogonal Matrices:} 
Let $\bX = [\bx_{1},\dots,\bx_{r}]$ be a $d \times r$ semi-orthogonal matrix $(d \geq r)$, 
i.e., $\bX\trans\bX$ is diagonal. 
An alternative specification of the orthogonality constraint may be via a column-wise sequence of linear constraints as $\{\bGamma_{2}\bx_{2}=\bzero,\bGamma_{3}\bx_{3}=\bzero,\ldots,\bGamma_{r}\bx_{r}=\bzero\}$, where $\bGamma_{\ell}=[\bx_{1},\dots,\bx_{\ell-1}]$ is a $d \times (\ell-1)$ dimensional matrix for $\ell=2,\dots,r$.
Thus, we can specify a distribution on $\bX$ as $P(\bX) = P(\bx_1)P(\bx_2\mid\bx_{1})\cdots P(\bx_{r}\mid \bx_{r-1},\ldots\bx_{1})$ 
with $P(\bx_{\ell}\mid\bx_{\ell-1},\ldots,\bx_{1})=\Normal(\bmu_{\ell},\bSigma_{\ell} \mid \bx_{\ell}\trans\bx_{\ell'}=0, \ell<\ell') = \Normal(\bmu_{\ell},\bSigma_{\ell} \mid \bGamma_{\ell}\bx_{\ell}=\bzero)$.
This mimics the well-known Gram-Schmidt-type construction, similar to \cite{north2024flexible}, 
but avoids unit norm constraints, and hence computationally more tractable. 
%
The full conditional for $\bx_{\ell}$ is then given by $P(\bx_{\ell}\mid \bX_{-\ell})=\Normal(\bmu_{\ell},\bSigma_{\ell} \mid \bX_{-\ell}\bx_{\ell}=\bzero)$, where $\bX_{-\ell}=[\bx_{1},\dots,\bx_{\ell-1},\bx_{\ell+1},\ldots,\bx_{r}]$.
We use this strategy in Section \ref{sec: sm post comp} to efficiently sample the mode matrices from the posterior. 

A significant additional advantage of the CSC 
is that the marginal distributions are multivariate normal, 
which makes them amenable to 
desirable structural constraints via judicious choices of their covariance matrices. 
We exploit this feature to induce smoothness in the columns of the spatial and temporal mode matrices using the PING process \citep{roy2021spatial} 
as well as to induce cumulative shrinkage with increasing column index for semi-automated rank selection \citep{bhattacharya2011sparse}, as described next.

{\bf PING for Piece-wise Smooth Processes:} 
The PING process is constructed as the location-wise product of Gaussian processes (GPs) with a smoothness-inducing covariance kernel. 
Let $\theta_{1}(v),\dots,\theta_{q}(v)$ be $q$ iid GPs with mean $\eE\{\theta_{k}(v)\} = 0$, 
variance $\var\{\theta_{k}(v)\} = 1$, and 
a covariance kernel $\cov\{\theta_{k}(v),\theta_{k}(v')\} = K(v,v')$. 
The PING process is then defined as $\theta(v) = \sigma \prod_{k=1}^{q}\theta_{k}(v)$. 
For $q=1$, the PING reduces to a simple GP. 
For $q>1$, the distribution of $\theta(v)$ has 
a higher concentration around zero, promoting sparsity, and 
heavier tails, preserving large signals. 
Also, if $K(v,v')$ is smooth, $\theta(v)$ will also be smooth, effectively capturing spatial dependence. 
The parameter $\sigma$ characterizes the overall scale of $\theta(v)$.

A major advantage of the PING prior is that it is computationally easily tractable. 
For any model leading to a conditionally normal likelihood function, the posterior full conditionals of the PING component processes are also all multivariate normal, 
allowing them to be easily recursively sampled from the posterior using efficient Gibbs sampling schemes. 

{\bf CSC-PING Prior for Semi-Orthogonal Matrices:}
To adapt the CSC distribution and the PING process to our setting, 
for $s \in \S_{g,t}$, 
we let $a^{(s)}_{z_{s}}(h_{s})=\sigma_{z_{s}}^{(s)}\prod_{k=1}^{q_{s}} a^{(s)}_{z_{s},k}(h_{s})$, 
where $a^{(s)}_{z_{s},1}(h_{s}),\dots,a^{(s)}_{z_{s},q_{s}}(h_{s})$ are iid GPs 
with mean $0$, variance $1$, and covariance kernel $K^{(s)}(h_{s},h_{s}')$. 
We then set $P(\ba_{z_{s}}^{(s)} \mid \ba_{z_{s}-1}^{(s)},\dots,\ba_{1}^{(s)})=\Normal(\bzero,\bSigma^{(s)} \mid \ba_{\ell}^{(s) \text{T}}\ba_{z_{s}}^{(s)}=0, \ell<z_{s})=\Normal(\bzero,\bSigma^{(s)} \mid \bGamma_{z_{s}}^{(s)}\ba_{z_{s}}^{(s)}=\bzero)$, where $\bGamma_{z_{s}}^{(s)} = [\ba_{1}^{(s)},\dots,\ba_{z_{s}-1}^{(s)}]\trans$ is an $(z_{s}-1)\times d_{s}$ matrix, and $\bSigma^{(s)}$ is the covariance kernel of the PING process induced on $\ba^{(s)}_{z_{s}}$. 

{With different choices of $q_{s}$ and $K^{(s)}$, the CSC-PING prior can be conditioned to adapt to different sparsity and smoothness levels as appropriate for functional domains. 
The treatment of the mode matrices $\bA_{\alpha}^{(s)}$ and $\bA_{\beta}^{(s)}$ for $s\in \S$, whose columns span spatial domains, 
that of $\bA_{\beta}^{(t)}$, 
whose columns span a temporal domain, 
as well as that of $\bA_{\alpha}^{(g)}$ and $\bA_{\beta}^{(g)}$, whose columns span a discrete space, 
can all be covered in a unified manner, 
as we describe below.}

{\bf Temporal Domain:} 
For $s=t$, since we do not anticipate sudden changes over time, 
we do not induce strong shrinkage but set $q_{t}=1$ whereby the CSC-PING reduces to a CSC-GP.

{\bf Discrete Domain:} 
For the case $s=g$, where the columns span the categorical group labels, 
we set $q_{g}=1$ with the GP covariance kernel set to be an identity matrix.

{\bf Spatial Domains:} 
For the spatially varying cases $s\in \S$, 
where the columns are expected to be smooth but sparse, we let $q_{s}>1$. 
However, instead of letting the components $a^{(s)}_{z_{s},k}(h_{s})$ follow a simple GP, as in the original PING, 
we consider a further low-rank approximation, as described below, 
which helps speed up the computation. 

To simplify notation, we omit the suffixes $\alpha$ and $\beta$ in this paragraph. 
Specifically, we let $a^{(s)}_{z_{s},k}(h_{s})=\sum_{m=1}^{m_{s}}\gamma^{(s)}_{z_{s},k,m} g_{m}^{(s)}(h_{s}) = \bg^{(s)}(h_{s})\trans\bgamma^{(s)}_{z_{s},k}$, where
$\bg^{(s)}(h_{s})=[g_{1}^{(s)}(h_{s}),\dots,g_{m_{s}}^{(s)}(h_{s})]\trans$ constitute a basis, and 
$\bgamma^{(s)}_{z_{s},k}=[\gamma^{(s)}_{z_{s},k,1},\dots,\gamma^{(s)}_{z_{s},k,m_{s}}]\trans\sim\Normal(\bzero, \bQ^{(s)})$ are the associated basis coefficients.
Let $\bU^{(s)}=((u_{\ell,\ell'}))$ be an $m_{s}\times m_{s}$ adjacency matrix, where $u^{(s)}_{\ell,\ell-1}=u^{(s)}_{\ell,\ell+1}=1$ and $u_{\ell,\ell'}=0$ otherwise.
Also, let the corresponding diagonal degree matrix be $\bD^{(s)}$. 
We set the graph Laplacian 
$\bQ^{(s)}=(\bD^{(s)}-\eta^{(s)}\bU^{(s)})^{-1}$
with $\eta^{(s)}=0.99$ 
which makes the $\bQ^{(s)}$'s valid covariance matrices with banded inverses. 
%
One may set a prior on $\eta^{(s)}$ for additional flexibility \citep{nychka2015multiresolution}. 
However, due to the orthogonality restriction on the mode matrices, 
we do not pursue it here but use a fixed value in favor of computational simplicity and stability. 
We then set   
$g_{m}^{(s)}(h_{s}) = \wt{g}_{m}^{(s)}(h_{s})/\wt{w}_{m}^{(s)}(h_{s})$, with 
$\wt{g}_{m}^{(s)}(h_{s}) = \exp\left(-\frac{h_{s}^{2}}{2h_{s}^{2}}\right)1\{\abs{h_{s}}\leq 3h_{s}\}$ 
and $\wt{w}_{m}^{(s)}(h_{s})=\{\wt\bg^{(s)}(h_{s})\trans\bQ^{(s)}\wt\bg^{(s)}(h_{s})\}^{1/2}$, where
$\wt\bg^{(s)}(h_{s})=[\wt{g}_{1}^{(s)}(h_{s}),\dots,\wt{g}_{m_{s}}^{(s)}(h_{s})]\trans$.
%
%
%
If we ignore the orthogonality restrictions, 
we now have $a^{(s)}_{z_{s}, k}(h_{s}) \sim \Normal\{0,\bg^{(s)}(h_{s})\trans\bQ^{(s)}\bg^{(s)}(h_{s})\} = \Normal(0,1)$. 
Fixing the variances of the $a^{(s)}_{z_{s}, k}(h_{s})$'s this way ensures that the parameters $\sigma_{z_{s}}^{(s)}$ maintain their interpretation as the overall scale of the $a^{(s)}_{z_{s}}(h_{s})$'s. 
Also, $\cov\{a^{(s)}_{z_{s}, k}(h_{s,1}),a^{(s)}_{z_{s}, k}(h_{s,2})\} = \corr\{a^{(s)}_{z_{s}, k}(h_{s,1}),a^{(s)}_{z_{s}, k}(h_{s,2})\} = \bg^{(s)}(h_{s,1})\trans\bQ^{(s)}\bg^{(s)}(h_{s,2})$. 
Let the resulting covariance kernel for the $\ba_{z_{s}}^{(s)}$'s be denoted by $\bSigma^{(s)}$. 
Combining with the CSC structure, we can now assign CSC-PING priors on the $\ba_{z_{s}}^{(s)}$'s 
that simultaneously incorporate spatial sparsity and column-wise orthogonality.

To deal with the cases $s=g,t$ later under a common framework, we set, for these cases, $m_{s}=d_{s}$ and $\bG^{(s)}=\bI_{d_{s}}$, thereby defining $\gamma^{(s)}_{z_{s},1,h_{s}}=a_{z_{s},1}^{(s)}(h_{s})$ with $\bQ^{(g)} = \bI_{d_{g}}$, 
and $\bQ^{(t)}$ an appropriately chosen GP covariance kernel.

{\bf Cumulative Shrinkage on the Columns:} 
Finally, we impose a cumulative shrinkage prior on the scale parameters $\sigma^{(s)}_{z_{s}}$'s. 
Specifically, we set $\sigma^{(s)}_{z_{s}}=\prod_{k=1}^{z_{s}}\varsigma^{(s)}_{k}$ and $\varsigma^{(s)}_{1}\sim\Ga(\kappa_{1},1), \varsigma^{(s)}_{k}\sim\Ga(\kappa_{2},1)$ for $k>1$. 
Here $\Ga(a,b)$ denotes a Gamma distribution with mean $a/b$ and variance $a/b^{2}$. 
The $\sigma^{(s)}_{z_{s}}$'s are stochastically increasing when $\kappa_{2}>1$, which favors more shrinkage as the column index $z_{s}$ increases, implying their diminishing importance \citep{bhattacharya2011sparse,durante2017note}. 
This stochastic ordering enables semi-automated rank selection whereby we start with liberally large numbers of columns, eliminating the redundant ones by shrinking them toward zero. 

{In recent Bayesian Tucker decomposition-based models \citep[e.g.,][etc.]{spencer2024bayesian, stolf2024bayesian}, rank selection is achieved via sparsity-inducing priors on the core tensors. 
In light of Lemma~\ref{lem: Tucker-HOSVD}, here we perform this task using the mode matrices instead.}

The CSC-PING prior, as introduced above, inherits the computational advantages of both the CSC distribution and the PING process. 
Additionally, the component processes $\ba_{z_{s},k}^{(s)}$ admit closed from multivariate normal posterior full conditionals 
and hence can be easily recursively sampled via an efficient Gibbs sampler described in Supplementary Section \ref{sec: sm post comp}.

%

    \item {\bf Core Tensors:} We let $c_{\alpha,z_{g}, z_{1}, z_{2}, z_{3}} \sim \Normal(0, \sigma_{\alpha}^{2})$ with $\sigma^{-2}_{\alpha}\sim\Ga(0.1,0.1)$. 
    Similarly, $c_{\beta,z_{g}, z_{1}, z_{2}, z_{3}, z_{4}} \sim \Normal(0, \sigma_{\beta}^{2})$ with $\sigma^{-2}_{\mu}\sim\Ga(0.1,0.1)$.

    \item {\bf Variance Parameters:} Finally, we let $\sigma_{\epsilon}^{-2}\sim\Ga(a_{\epsilon},b_{\epsilon})$ for the error variance, and $\sigma^{-2}_{\alpha,z_{g},z_{1},z_{2},z_{3}}\sim\Ga(a_{s,\alpha},b_{s,\alpha})$, $\sigma^{-2}_{\beta,z_{g},z_{1},z_{2},z_{3},z_{t}}\sim\Ga(a_{s,\beta},b_{s,\beta})$, and $\tau^{-2}_{\alpha},\tau^{-2}_{\beta}\sim \Ga(a_{\tau},b_{\tau})$.

    


\end{enumerate}

\section{Additional Discussions on Hyper-parameters} \label{sec: sm hyper-parameters1}
\vspace*{-1ex} 
With some repetition from the paper for completeness and readability, we recall that our implementation involves only a small number of prior hyperparameters: $\kappa_1, \kappa_2$ for cumulative shrinkage; $q_s$ for the PING prior; and $a_s, b_s$, $a_\tau, b_\tau$ for inverse Gamma distributions. We use $\kappa_1 = 2.1$, $\kappa_2 = 3.1$, $a_s = 10$, $b_s = 0.1$ (for all $s$), and $a_\tau = b_\tau = 0.1$. 
For the components of $\bbeta$ ($s = 1, 2, 3$), we set $q_s = 3$, while for $\balpha$, $q_s = 1$ to reflect reduced spatial sparsity in the baseline.
The $\text{Ga}(10, 0.1)$ priors encourage shrinkage, while the $\text{Ga}(0.1, 0.1)$ priors on the global precisions $\tau_\alpha^{-2}$ and $\tau_\beta^{-2}$ are diffuse. We also set the number of B-spline bases $d_t$ to approximate the maximum number of patient visits, and the number of Gaussian bases $m_s = \lfloor d_s / 2 \rfloor$ with spacing $v_s = d_s / m_s$.

The choices for $a_{\tau}, b_{\tau}$ are non-informative.
The choices of $\kappa_{1}$ and $\kappa_{2}$ are as recommended by \cite{durante2017note}. 
The choice for $a_{s}$ and $b_{s}$ induces shrinkage that offers numerical stability and is partly motivated by \cite{roy2021perturbed}. 
The choice $q = 3$ for the PING was found to perform best across various settings in \cite{roy2021spatial}. 
Although \cite{roy2021spatial} also explored cross-validation methods for setting $q$, this approach was computationally prohibitive for our purposes. 
Instead, we conducted a preliminary exploration and again found $q_{s} = 3$ to yield very robust performances. 
Our choice of the number of B-splines is supported by \citet{ruppert2002selecting}, who showed that beyond a minimum threshold, the fit quality is largely insensitive to the exact number of splines -- a finding consistent with our own numerical experiments. 
The choice for the number of Gaussian bases 
is also motivated by the numerical experiments from \cite{roy2021spatial} using low-rank approximations, and also found to be working well in our current numerical experiments. 

\vspace*{-5ex}
\section{Posterior Computation} \label{sec: post comp}
\vspace*{-1ex} 
For finite sample inference, we use MCMC samples drawn from the posterior. 
A major advantage of our proposed approach is that all parameters have closed-form posterior full conditionals, enabling efficient Gibbs sampling. 
Moreover, {\ul{the core tensor elements can be updated in parallel (for $\alpha$) or semi-parallel (for $\beta$)} leveraging the distributional properties induced by semi-orthogonal mode matrices (Lemmas \ref{lem: distprop alpha} and \ref{lem: distprop beta}). 
\ul{Together, these aspects overcome major computational bottlenecks in exploring the full posterior of Tucker tensor factorization models, even in complex ultra-high-resolution settings with over $250,000$ voxels per image like the ADNI-3 data set.}
We thus view this as a major contribution of our work. 
Due to space constraints, details are presented in Section \ref{sec: sm post comp} in the SM.

\vspace*{-3ex} 
\section{Large-sample Properties} \label{sec: asymptotics}
\vspace*{-1ex} 
We study prior support and posterior consistency in recovering the marginal mean parameters $\alpha(h_{g}, h_{1}, h_{2}, h_{3})$ and $\beta(h_{g},h_{1},h_{2},h_{3},t)$ as $N$, the total number of subjects, goes to $\infty$. 
For simplicity, we assume that the true model is the one in (3) and that the observation times are the same $\{t_1,\ldots,t_n\}$ with $n_{i}=n$, 
and do not change with $N$. 
For any vector $\ba$, let $\|\ba\|_{2} = \sqrt{\sum_{j}a_{j}^{2}}$ denote the Euclidean norm, and $\|\ba\|_{\infty} = \max_{j} \abs{a_{j}}$ the supremum norm. 
Likewise, for any matrix $\bA$, let $\|\bA\|_{\infty} = \|\vect(\bA)\|_{\infty}$ denote the supremum norm, and $\|\bA\|_{F} = \|\vect(\bA)\|_{2}$ the Frobenius norm, where $\vect(\bA)$ is the vectorization of $\bA$. 
The symbol $\lesssim$ will stand for domination by a constant multiple for two sequences. 

Let $\bTheta=\{\etam_{\alpha},\bA^{(s)}_{\alpha}, s \in \S_{g}\} \cup \{\etam_{\beta},\bA^{(s)}_{\beta}, s \in \S_{g,t}\}\cup\{\bC_{\alpha},\bC_{\beta},\bSigma_{\alpha},\bSigma_{\beta},\sigma_{\epsilon}^{2}\}$ denote the complete set of parameters, $\bTheta_{0}$ be the set of corresponding true parameters, 
and $\eE_{0}$ be the expectation taken under the true probability law. 


To present a simplified posterior consistency proof, 
we invoke the latent factor representation of the proposed model from Section~\ref{sec: factor model} and define some simplified notations. 
Let $\wt\bY^{(i)}=[\wt\bY^{(i)} (h_{g},t_{1})\trans,\ldots,\wt\bY^{(i)} (h_{g},t_{n})\trans]\trans$. 
Then, $\wt\bY^{(i)} \sim \MVN(\bmu^{(i)},\bSigma^{(i)})$, where, if the $i\th$ subject belongs to the $h_{g}\th$ group, 
$\bmu^{(i)}=\bmu_{h_{g}}=\bone_{n}\otimes\{\bLambda_{\alpha,h_{g}}\vec(\bC_{\alpha})\}+\bLambda_{\beta,h_{g}}\vec(\bC_{\beta})$ and $\bSigma^{(i)} =\bSigma_{h_{g}}=(\bone_{n}\bone_{n}\trans)\otimes \{\bLambda_{\alpha,h_{g}}\bSigma_{\alpha,\eta}\bLambda_{\alpha,h_{g}}\trans\}+\bLambda_{\beta,h_{g}}\bSigma_{\beta,\eta}\bLambda_{\beta,h_{g}}\trans+\sigma^{2}_{\epsilon}\bI_{nd_{1}d_{2}d_{3}}$.
Then, $\sum_{i}\|\bmu^{(i)}-\bmu_{0}^{(i)}\|_{2}^{2}=n\sum_{h_{g}=1}^{d_{g}} N_{h_{g}}\|\balpha_{h_{g}}-\balpha_{0,h_{g}}\|^{2}_{F} + \sum_{h_{g}=1}^{d_{g}}N_{h_{g}} \sum_{k=1}^{n}\|\bbeta_{h_{g},t_{k}}-\bbeta_{0,h_{g},t_{k}}\|^{2}_{F}$. 
Here $\bone_{n}$ stands for the vector of all ones with length $n$, 
$\balpha_{h_{g}}=\balpha(h_{g},\cdot,\cdot,\cdot)$, 
$\bbeta_{h_{g},t_{k}} = \bbeta(h_{g},\cdot,\cdot,\cdot,t_{k})$, 
$\bLambda_{\beta,h_{g}}=[\bLambda_{\beta}(h_{g},t_{1})\trans,\ldots,\bLambda_{\beta}(h_{g},t_{n})\trans]\trans$.

\begin{Assmp}
The true coefficients admit a Tucker form and the mode matrices satisfy:
\begin{itemize} \setlength\itemsep{-0.5em}
    \item[(1)] (Bounded Modes) 
    $\|\bA^{(s)}_{\alpha,0}\|_{\infty}, \|(\bA^{(s)}_{\alpha,0})\trans\bA^{(s)}_{\alpha,0}\|_{\infty} \leq C_1$ for all $s \in \S_{g}$,\\ 
    and $\|\bA^{(s)}_{\beta,0}\|_{\infty}, \|(\bA^{(s)}_{\beta,0})\trans\bA^{(s)}_{\beta,0}\|_{\infty}\leq C_2$ for all $s \in \S_{g,t}$.
    \item[(2)] (Smooth Modes) 
    The functions $a_{\alpha,z_{s},0}^{(s)}(\cdot)$ and $a_{\beta,z_{s},0}^{(s)}(\cdot)$ for all $s \in \S$ belong to the H\"older class of regularity level $\iota$ on $[0,1]$. 
    Additionally, the functions $a_{\beta,z_{t},0}^{(t)}(\cdot)$ in the temporal mode-matrix belongs to the H\"older class of regularity level $\iota'$ on $[0,1]$. 
    \item[(3)] (Growing Dimensions) 
    $\log d_{s} \lesssim \log N$ for all $s \in \S$. 
    \item[(4)] (Known Growing Ranks) The Tucker ranks are known and satisfy $r_{\alpha,s}, r_{\beta,s} \le d_{s}$ with $\log r_{\alpha,s}, \log r_{\beta,s} \lesssim \log N$. 
\end{itemize}
\label{assmp::reg_coef}
\end{Assmp}
\vspace*{-2ex}
Assumption (2) is standard in functional analysis \citep{shen2015adaptive}.
Assumptions (1), (3), and (4) are standard in the asymptotics for tensor-decomposition guided models \citep{guhaniyogi2017bayesian, sun2017store, deng2023correlation}. 
However, we require one extra condition on the cross-product term in Assumption (1) to control the eigenvalues of $\bSigma_{h_{g},0}$. 
It can be relaxed by controlling true random effect variances but is omitted here for simplicity. 
Also, due to assumption (4), we drop the cumulative shrinkage component from the prior, 
further simplifying the analysis. 
We recall that a function $f$ 
is in H\"older space following the Definition C.4 of \cite{GhosalBook}. 
The rank inequality in the last assumption is not restrictive since $r_{\alpha,s},r_{\beta,s}\ll d_{s}$ are generally needed to leverage tensor decomposition models while being typically sufficient to achieve strong results. 
The second part of the assumption allows the ranks to grow with $N$.
On the number $d_{t}$ of B-spline bases and $m_{1},m_{2},m_{3}$ of Gaussian bases, we assume a prior with a Poisson-like tail $e^{-x \log x}$, 
as commonly done in posterior consistency analysis for such models \citep{shen2015adaptive}.
In practice, with data-adaptive smoothness, 
the number of bases is usually not crucial beyond a minimum threshold \citep{ruppert2002selecting}. 
Therefore, to simplify posterior computation, 
we fix $m_{1},m_{2},m_{3}$ and $d_t$ at sufficiently large values 
that ensure efficient function approximation.
The number of components in the PING prior is also kept fixed. 

Let $\Pi(\cdot)$ and $\Pi_{N}(\cdot \vert \bY^{(1:N)})$ be generic notations for the priors and the posterior, respectively. 
The following two lemmas establish the large sup-norm supports of the PING and CSC-PING priors, respectively, highlighting their theoretical flexibility, 
followed by our main Theorem on the consistency of the posterior. 

\begin{Lem}[Large Sup-Norm Support for PING] \label{lem: sup-norm PING}
For any $\bbeta_{\star}$ with $\|\bbeta_{\star}\|_{\infty}< C$, for a constant $C>0$ and any $\epsilon>0$, $\bbeta\sim$PING$(q)$ satisfies $\Pi(\big\|\bbeta-\bbeta_{\star}\big\|_{\infty}<\epsilon)>0$. 
\label{thm::largesupping}
\end{Lem}

\begin{Lem}[Large Sup-Norm Support for CSC-PING]\label{lem: sup-norm CSC-PING}
Let a $p\times q$ ($p>q$) semi-orthogonal matrix $\bA$ follow CSC-PING. 
Then, $\Pi(\big\|\bA-\bA_{\star}\big\|_{\infty}<\epsilon)>0$ for any semi-orthogonal $\bA*$ of same dimension with $\|\bA_{\star}\|_{\infty}<C$ for a constant $C>0$ and any $\epsilon>0$.  
\label{thm::largesupcsc}
\end{Lem}

\begin{Thm}
Under Assumptions~\ref{assmp::reg_coef}, 
for every $\epsilon>0$,
\vspace*{-7ex}\\
\bse
\textstyle \eE_{0}\left[\Pi_{N}\left(\frac{1}{N}\sum_{i}\|\bmu^{(i)}-\bmu_{0}^{(i)}\|_{2}^{2}>\epsilon \mid \bY^{(1:N)} \right)\right] \rightarrow 0 \text{ as } N \rightarrow \infty.
\ese
\vspace*{-8ex}
\label{thm: consistency}
\end{Thm}

Since $\{n\|\balpha_{h_{g}}-\balpha_{0,h_{g}}\|^{2}_{F}+\sum_{k=1}^{n}\|\bbeta_{h_{g},t_{k}}-\bbeta_{0,h_{g},t_{k}}\|^{2}_{F}\} \lesssim \frac{1}{N}\sum_{i}\|\bmu^{(i)}-\bmu_{0}^{(i)}\|_{2}^{2}$ 
when $N_{h_{g}}/N\rightarrow p_{h_{g}}$, 
Theorem~\ref{thm: consistency} straightforwardly extends to the following Corollary, 
establishing consistent recovery separately for each group as well.

\begin{Cor}
Furthermore, if $N_{h_{g}}/N\rightarrow p_{h_{g}}$, where $p_{h_{g}}\in(0,1)$ as $N \rightarrow \infty$, we have 
\vspace*{-7ex}\\
\bse
\textstyle \eE_{0}\left[\Pi_{N}\left(n\|\balpha_{1,g}-\balpha_{0,g}\|^{2}_{F}+ \sum_{k=1}^{n}\|\bbeta_{1,t_{k},g}-\bbeta_{0,t_{k},g}\|^{2}_{F}>\epsilon\mid\bY^{(1:N)} \right)\right] \rightarrow 0.
\ese
\vspace*{-8ex}
\end{Cor}

The proofs of the Lemmas \ref{lem: sup-norm PING}, \ref{lem: sup-norm CSC-PING}, and Theorem \ref{thm: consistency} are 
presented in Section \ref{sec: sm proofs} in the SM. 
It is based on applying Theorem~6.39 of \cite{GhosalBook} which is an extension of
Schwartz's original result \citep{schwartz1965bayes} with both the truth and the priors depending on the sample size. 
Here, we verify the required testing conditions directly. 

\vspace*{-3ex} 
\section{ADNI-3 Data Analysis} \label{sec: adni analysis}
\vspace*{-1ex} 
Here we discuss the analysis of ADNI-3 DW-MRI data using our method. 
Additional details about the scientific background and the data set can be found in Section \ref{sec: sm sci bg} in the SM.

The raw DW-MRI data were processed, as described in detail in Section \ref{sec: sm sci bg} in the SM, producing $90 \times 108 \times 90$ initial dimensional images for each case. 
Since statistical analyses of neuroimaging data can be adversely influenced by empty voxels outside the brain, 
we followed \cite{feng2021brain} and 
cropped out these parts, 
resulting in final images of dimension $60\times 70\times 60$. 
The observation times are scaled to [0,1] by dividing observed times with the highest observation time of our study cohort. 

While our method generates detailed voxel-level estimates for the whole brain, 
voxel-wise visualizations are not easily discernible. 
We thus present here whole-brain-averaged and specific-region-averaged estimates for 
summarized inference. 
Specifically, we focus on five white-matter dense regions in the main paper: corpus callosum (CC), right corticospinal tract (CST), left CST, right frontal-parietal-temporal (FPT), and left FPT. 
Among these, CC is the largest, with around 13 thousand voxels. The other regions on average consist of $4$-$7$ thousand voxels.
In Section S.3 in the SM, we provide results similar to Tables~\ref{tab: arc-length} and \ref{tab: CGD} for a total of 21 tracts including these 5.

\begin{figure}[htbp]
\centering
\includegraphics[width = 0.45\textwidth, clip=true]{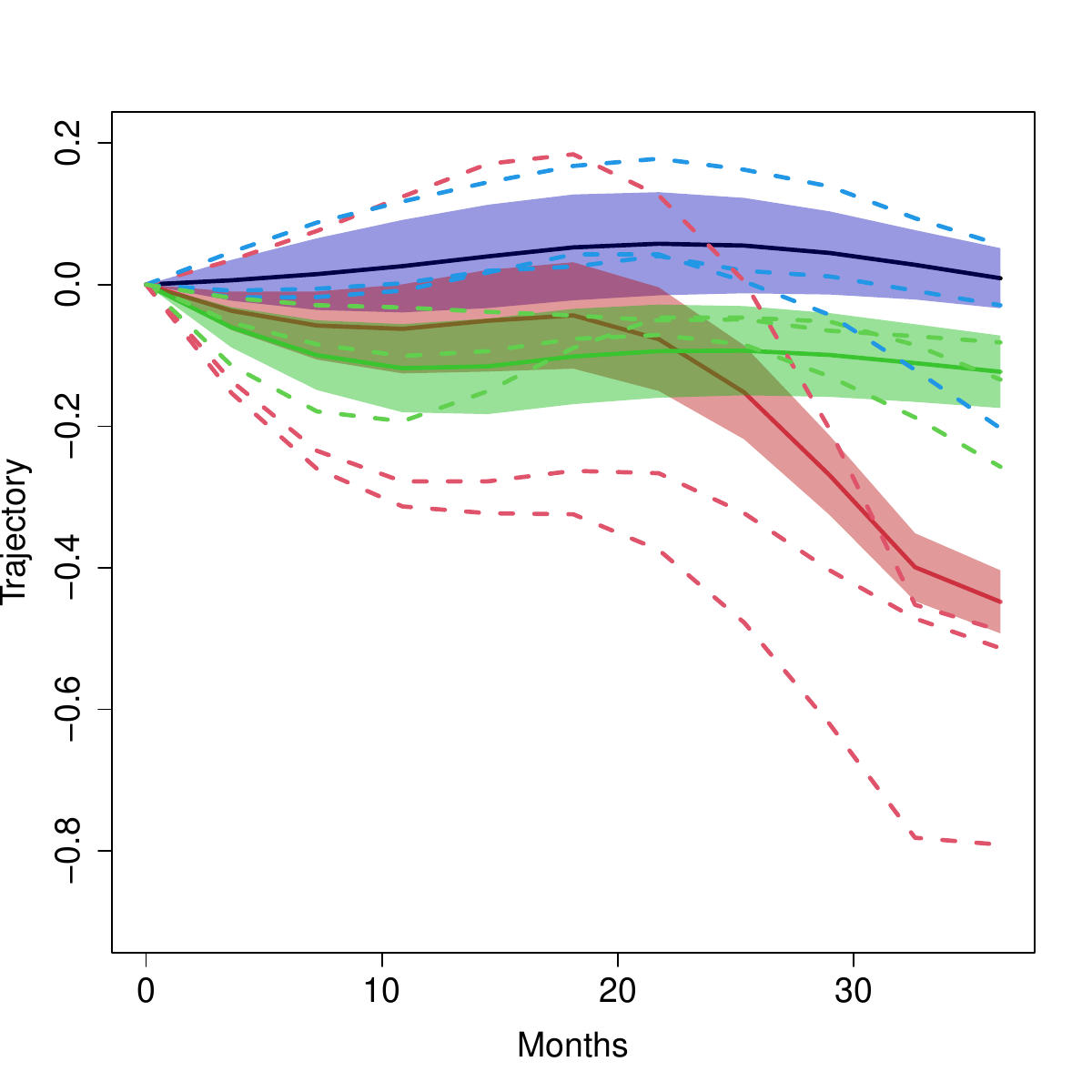}\quad\quad
\includegraphics[width = 0.45\textwidth, clip=true]{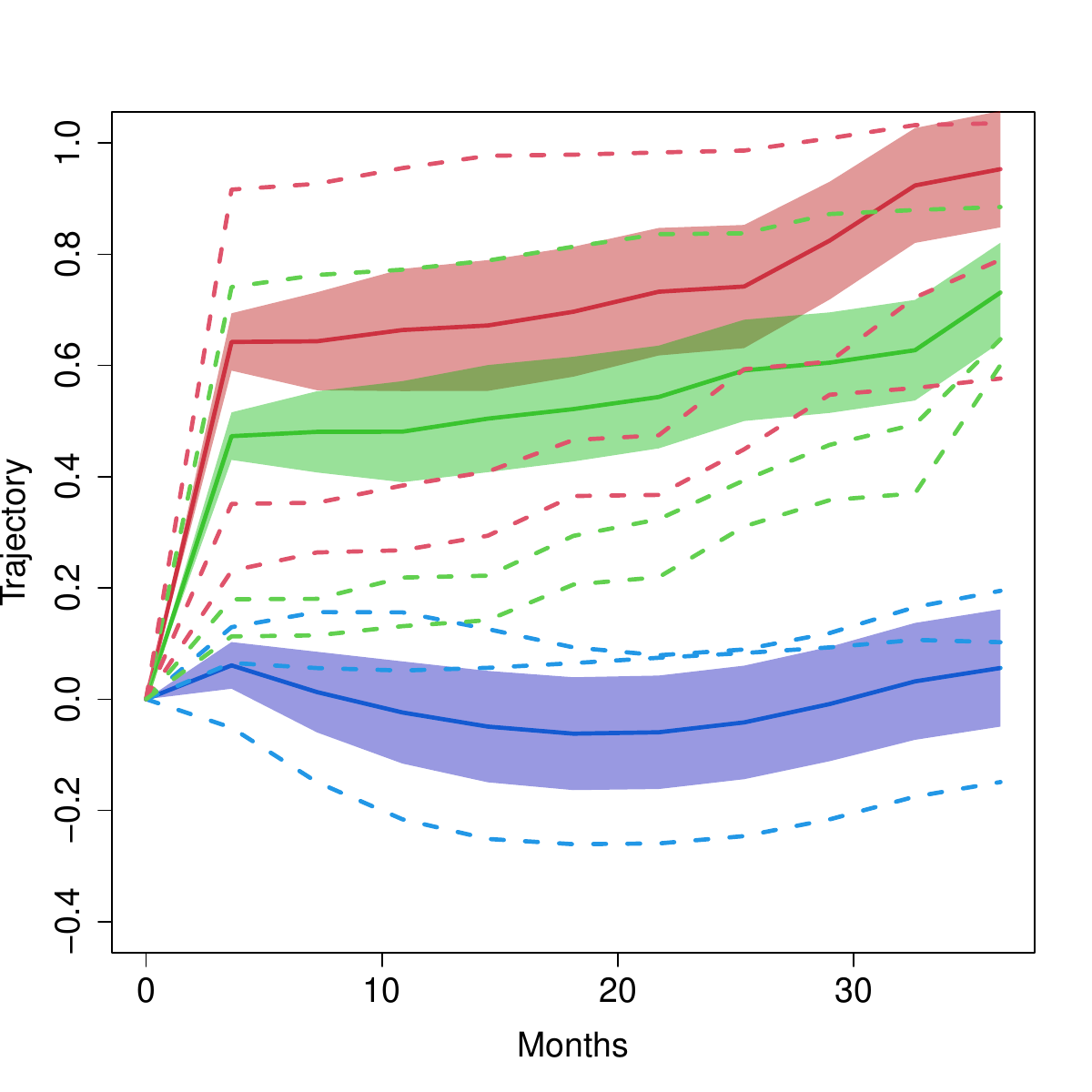}
\vspace*{-5pt}
\caption{Inference for $\beta$'s: Estimated median trajectories of FA (left panel) and ODI (right panel) in the corpus callosum (CC) region for normal cognition (NC, blue),  mild cognitive impairment (MCI, green), and Alzheimer's (AD, red) subjects. 
The solid lines show the population-level trajectories; 
the shaded regions show the corresponding 90\% point-wise credible intervals; 
the dotted lines show three randomly selected subjects from each group. 
\vspace*{-10pt}
}
\label{fig: realFA2 and realODI2}
\end{figure}

Analyses for the baselines, captured by the $\alpha$'s, are in Section \ref{sec: sm add figs real} in the SM, where we find the ordering of the groups based on baseline estimates in the CC region to broadly align with their respective degrees of cognitive decline that define these groups, specifically for ODI (Figure~\ref{fig: realbaseFA2 and realbaseODI2}). 
For the rest of this section, we focus on the longitudinal changes over the study period, 
captured by the $\beta$'s, 
where our main interests lie. 
They all start at zero due to the separability condition imposed on them (Section \ref{sec: models}). 
Figure~\ref{fig: realFA2 and realODI2} illustrates the estimated median trajectories of the CC region evaluated by taking the median over the approximately 14 thousand voxels in CC. 
Figure~\ref{fig: realFWF2 and realNDI2} in the SM provide similar results for FWF and NDI. 
We plot these trajectories up to 36 months since observations beyond that point are extremely sparse. 
Figure \ref{fig: realFA_and_ODI} shows the summarized differences across different tracts in FA and ODI. 
Figure \ref{fig: realFWF_and_NDI} in the SM provides similar figures for FWF and NDI. 
To facilitate comparison across the four markers, these figures all use the same scale. 
We also report the arc length of their estimated population-level longitudinal trajectories and the cross-group differences (CGD) for each pair of groups over the study period. 
These markers are computed as 
\vspace*{-7ex}\\
\bse
& \text{Arc-Length}_{\beta}(h_{g}) = \frac{1}{Td_{1}d_{2}d_{3}}\sum_{h_{1},h_{2},h_{3}}\sum_{\ell=1}\trans\abs{\frac{d}{dt} \beta(h_{g},h_{1},h_{2},h_{3},t)}_{t=\ell/T},\\
& \text{CGD}_{\beta}(h_{g},h_{g}^{\prime}) = \frac{1}{Td_{1}d_{2}d_{3}}\sum_{h_{1},h_{2},h_{3}} \sum_{\ell=1}\trans\{\beta(h_{g},h_{1},h_{2},h_{3},\ell/T)-\beta(h_{g}^{\prime},h_{1},h_{2},h_{3},\ell/T)\}^{2}. 
\ese
\vspace*{-7ex}\\
Tables \ref{tab: arc-length} and \ref{tab: CGD} report these values computed with $T=20$ equidistant time points. 

Our analysis suggests that (1) the longitudinal changes are highly non-linear with considerable individual heterogeneity, justifying a nonlinear mixed effects model. 
(2) 
FA, ODI, and FWF show greater variations than NDI; here for brevity, we focus mainly on the first two markers and the changes in the CC region. 
(3) 
The largest changes are observed in the CC brain region.
(4) 
The trajectories also change the most for AD patients; they do not change much overall for the NC subjects; 
and for the MCI subjects, the change is positive but smaller than the AD individuals in magnitude.
(5)
The AD patients' FA trajectories show faster decay than the other two groups; conversely, their ODI trajectories increase, and they increase faster than the others. 
This aligns with the existing understanding of AD-related differences compared to NC and MCI subjects. 
(6) 
While the ODI trajectories show slow but persistent changes over time, FA shows drastic changes after a period of early slow trend.
(7) 
The ODI trajectories for different disease groups are consistently separated throughout the entire period. 
In contrast, the FA trajectories between the groups are similar in the initial part and then exhibit large differences.
(8) The FWF trajectory shows an increasing pattern in CC as shown in Figure S.1. Although the pattern is very similar to ODI, FWF seems to increase moderately even for the NC subjects, showing some effect of aging. 
However, the changes in NDI are highly non-linear and show an overall decreasing pattern among the MCI and AD subjects.  
(9) 
In general, as with the case of the initial values, the AD group is well-separated from the NC and MCI groups in terms of their longitudinal changes. 
The separations between the NC and AD trajectories and the AD and MCI trajectories are comparable but are considerably smaller for the NC and MCI groups across all four markers. 
(10) Cross-group differences are in general monotonic as shown in Table~\ref{tab: CGD}. 
(11) Additional results for a total of 21 white matter tracts similar to Tables~\ref{tab: arc-length} and \ref{tab: CGD}, presented in Section S.3 in the SM (Tables S.1-S.42), 
also exhibit similar patterns. 

\vspace*{-10pt}
\begin{table}[htbp]
\centering
\caption{$\text{Arc-Length}_{\beta}(h_{g})$ for different disease status groups $h_{g}$ across the whole brain.}
\begin{tabular}{rrrr}
  \hline
 & NC & MCI & AD \\ 
  \hline
FA & 0.589 & 0.642 & 0.909 \\ 
  ODI & 0.750 & 0.853 & 1.012 \\ 
  FWF & 0.687 & 0.733 & 0.990 \\ 
  NDI & 0.494 & 0.568 & 0.718 \\ 
   \hline
\end{tabular}
\label{tab: arc-length}
\end{table}
\vspace*{-15pt}

\begin{table}[htbp]
\centering
\caption{$\text{CGD}_{\beta}(h_{g},h_{g}^{\prime})$ between pairs of disease status groups $(h_{g},h_{g}^{\prime})$  across the whole brain.}
\begin{tabular}{rrrr}
  \hline
  & NC-AD & AD-MCI & NC-MCI \\ 
  \hline
  FA  & 0.061 & 0.059 & 0.030 \\ 
  ODI & 0.062 & 0.054 & 0.019 \\ 
  FWF & 0.064 & 0.054 & 0.020 \\ 
  NDI & 0.043 & 0.037 & 0.014 \\  
   \hline
\end{tabular}
\label{tab: CGD}
\end{table}


\begin{figure}[!ht]
\centering
\begin{minipage}{0.47\textwidth}
\subfigure[AD vs NC]{\includegraphics[width = 0.4\textwidth, trim=1cm 1cm 1cm 2cm, clip=true]{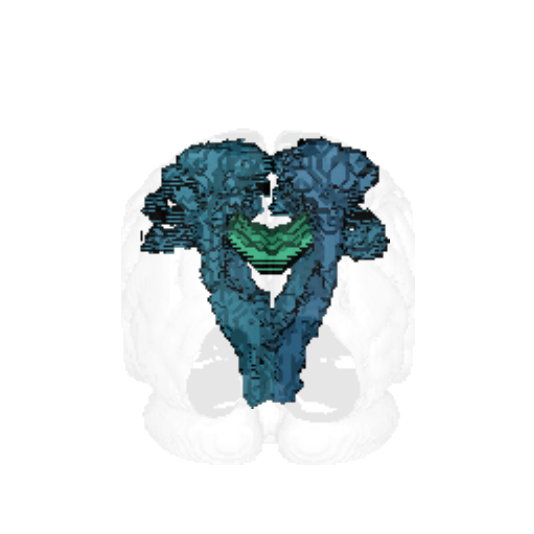}}
\subfigure[AD vs NC]{\includegraphics[width = 0.35\textwidth, trim=1cm 1cm 1cm 1cm, clip=true]{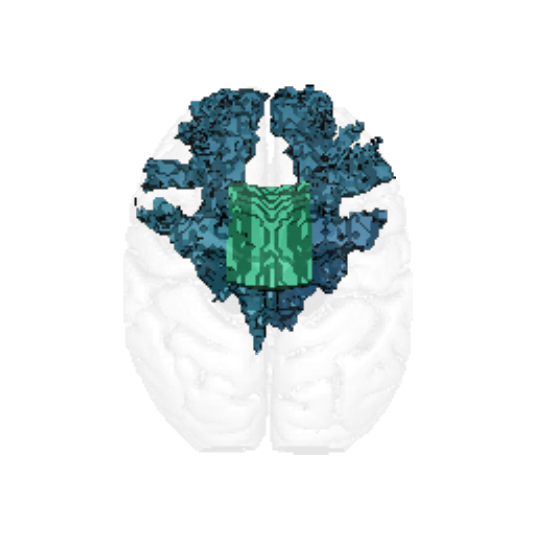}}\\
\vspace{-0pt} 
\subfigure[AD vs MCI]{\includegraphics[width = 0.4\textwidth, trim=1cm 1cm 1cm 2cm, clip=true]{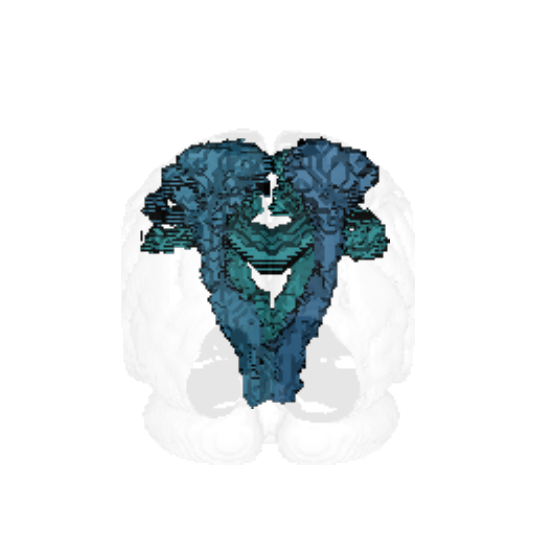}}
\subfigure[AD vs MCI]{\includegraphics[width = 0.35\textwidth, trim=1cm 1cm 1cm 1cm, clip=true]{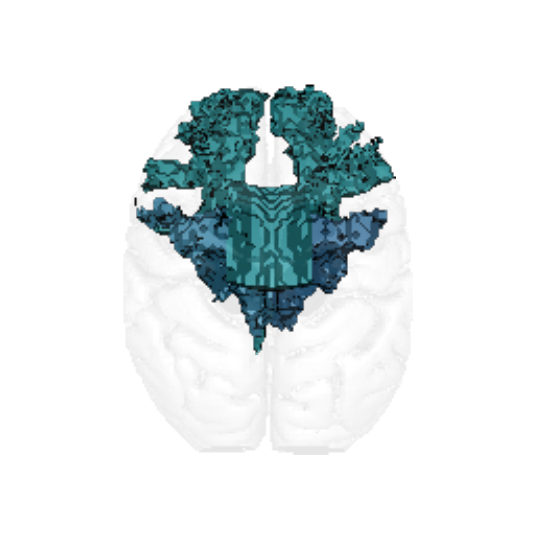}}\\
\vspace{-0pt} 
\subfigure[MCI vs NC]{\includegraphics[width = 0.4\textwidth, trim=1cm 1cm 1cm 2cm, clip=true]{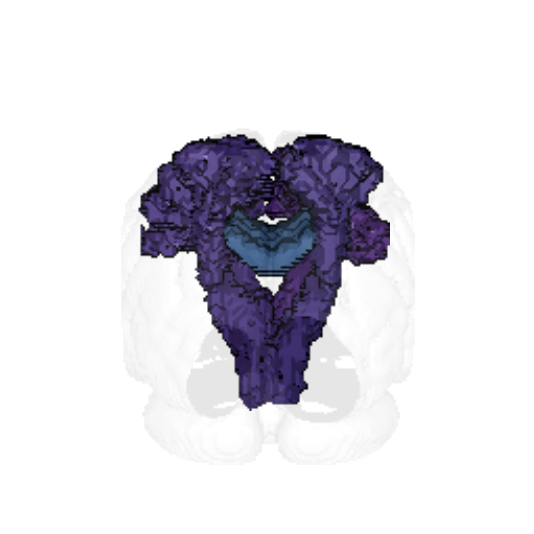}}
\subfigure[MCI vs NC]{\includegraphics[width = 0.35\textwidth, trim=1cm 1cm 1cm 1cm, clip=true]{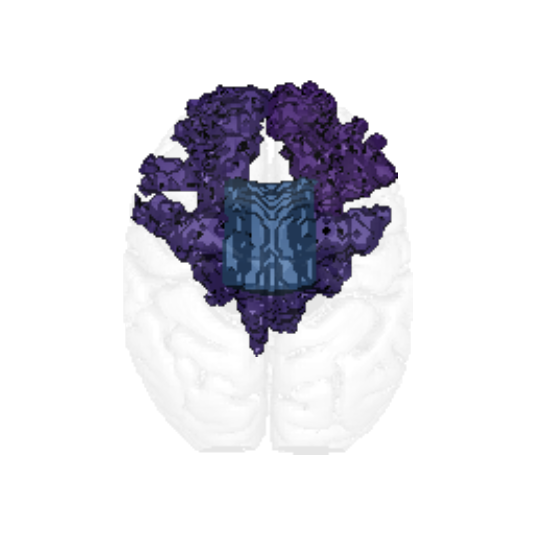}}
\end{minipage}
\begin{minipage}{0.04\textwidth}
\hspace*{-1cm}\subfigure{\includegraphics[width=3\textwidth, trim=10cm 0cm 0cm 0cm, clip=true]{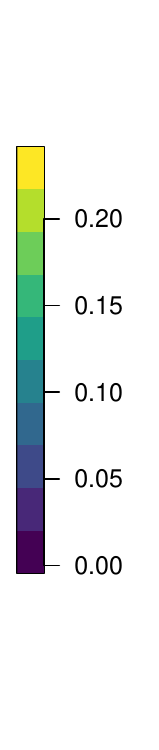}}
\end{minipage}
\begin{minipage}{0.47\textwidth}
\subfigure[AD vs NC]{\includegraphics[width = 0.4\textwidth, trim=1cm 1cm 1cm 2cm, clip=true]{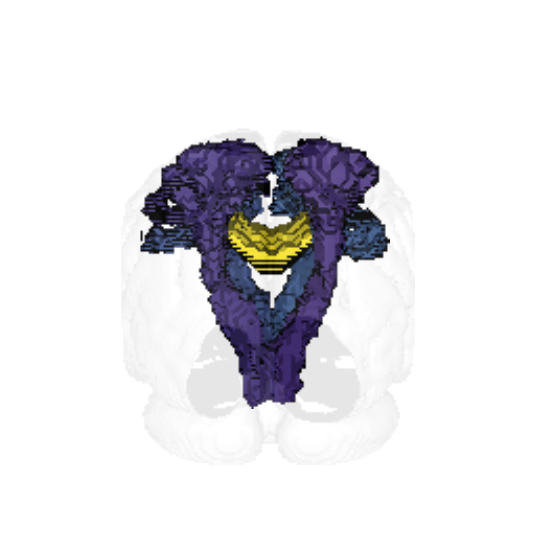}}
\subfigure[AD vs NC]{\includegraphics[width = 0.35\textwidth, trim=1cm 1cm 1cm 1cm, clip=true]{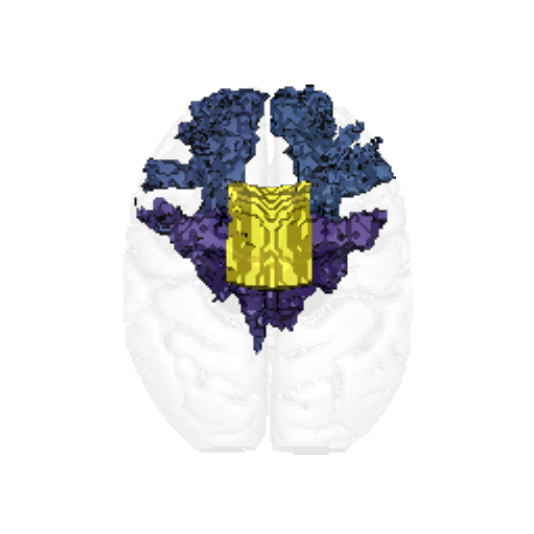}}\\
\vspace{-0pt} 
\subfigure[AD vs MCI]{\includegraphics[width = 0.4\textwidth, trim=1cm 1cm 1cm 2cm, clip=true]{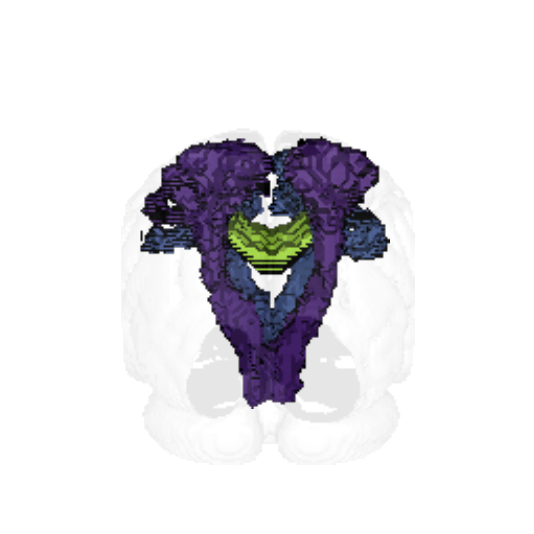}}
\subfigure[AD vs MCI]{\includegraphics[width = 0.35\textwidth, trim=1cm 1cm 1cm 1cm, clip=true]{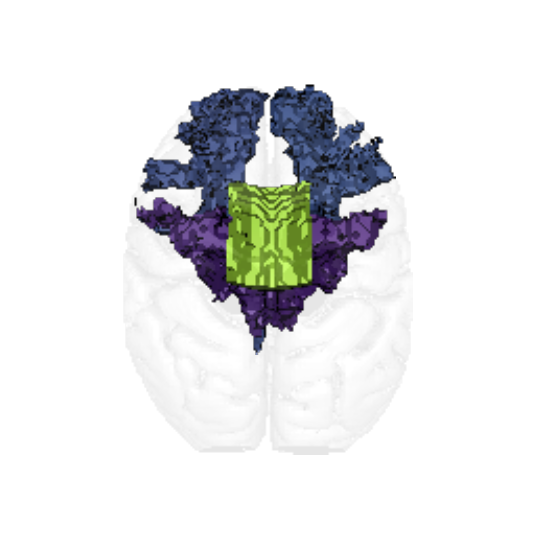}}\\
\vspace{-0pt}
\subfigure[MCI vs NC]{\includegraphics[width = 0.4\textwidth, trim=1cm 1cm 1cm 2cm, clip=true]{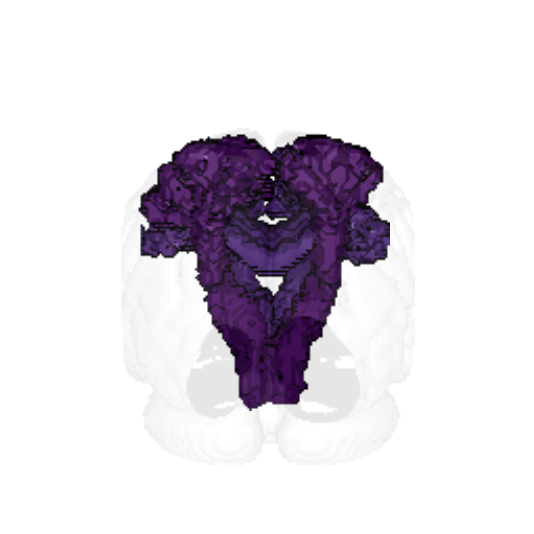}}
\subfigure[MCI vs NC]{\includegraphics[width = 0.35\textwidth, trim=1cm 1cm 1cm 1cm, clip=true]{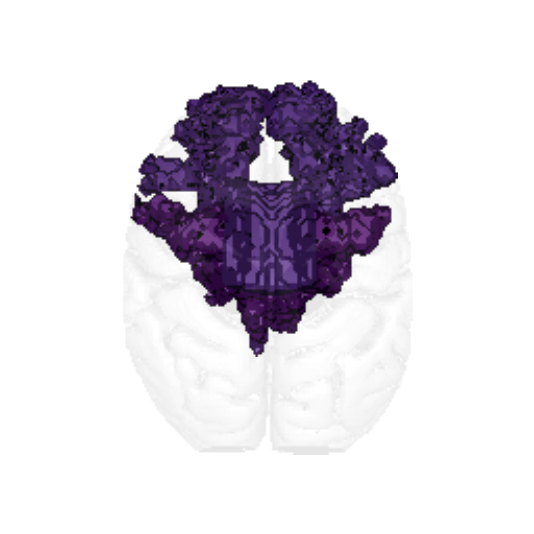}}
\end{minipage}
\caption{Inference for $\beta$'s: Summarized differences across five different tracts (left and right corticospinal tracts, left and right frontopontine tracts, and corpus callosum) in FA (sub-panels (a)-(f) to the left of the scale) and ODI (sub-panels (h)-(m) to the right of the scale) from two different angles.}
\label{fig: realFA_and_ODI}
\end{figure}

To our knowledge, a detailed longitudinal mixed-effects analysis of ADNI-3 DW-MRI data has not been conducted previously. 
Our general findings align with the existing literature \citep{bodini2014diffusion,assaf2017dMRI}, while the finer details are largely novel, offering deeper insights into disease development and progression, subject heterogeneity, and the relative efficacy of various DW-MRI-extracted markers in studying these questions.

\vspace*{-3ex} 
\section{Simulation Study} \label{sec: sim study}
\vspace*{-1ex} 

In our simulation design, 
we focus on investigating the gains achieved by considering our flexible and scalable compact HOSVD-based model with its smoothly varying functional mode matrices over smooth and non-smooth CP-based models as well as a non-smooth compact HOSVD-based construction. 
We recall here that the compact HOSVD is an equivalent representation of Tucker decomposition but with semi-orthogonal mode matrices (Lemma \ref{lem: Tucker-HOSVD}) which facilitates posterior computation (Section \ref{sec: sm post comp} in the SM). 
For a fair comparison, we consider a mixed-effect longitudinal CP factorized model as an alternative with different ranks. 
Additionally, we also fit a voxel-wise mixed model, referred to as the VMW model. 
%
%
We designed our experiments to closely replicate real-world settings. 
Details on the experimental design, including our synthetic data generating models, additional information on the CP and VMW models, etc., are provided in Section \ref{sec: sm sim add details and figs} in the SM.

\begin{table}[htbp]
\centering
\begin{tabular}{|r|rr|rr|r|} \hline 
 & \multicolumn{5}{c|}{MSE for $\bbeta$}  \\ \hline 
 & \multicolumn{2}{c|}{CP model} & \multicolumn{2}{c|}{HOSVD model} & {VWM Model} \\ \hline
& $r=5$ & $r=10$  & GLNS   & GLS & \\ 
  \hline
 Non-zero locations & 2482.9437 & 4.0649  & 0.0050 & 0.0052 & 125.49\\ 
  Zero locations & 5820.1787 & 0.1793 & 0.0040 & 0.0022 & 41.35\\ 
   \hline
\end{tabular}
\caption{MSE in estimating $\bbeta$ 
for a CP model with ranks $r_{\alpha} = r_{\beta} = r$, 
and our proposed HOSVD model with ranks $(r_{\alpha,g},r_{\alpha,1},r_{\alpha,2},r_{\alpha,3}) = (3,5,5,5)$ and $(r_{\beta,g},r_{\beta,1},r_{\beta,2},r_{\beta,3},r_{\beta,t}) = (3,5,5,5,4)$. 
Here GLNS=Non-Smooth-Laplacian, GLS=Smooth-Laplacian with $q=3$.} 
\label{tab: MSE}
\vspace*{-10pt}
\end{table}

Table \ref{tab: MSE} presents the mean squared error (MSE) of estimating $\bbeta$ by the different methods. 
The proposed compact HOSVD-based methods vastly outperform the competitors. 
Notably, the HOSVD-GLS model achieves approximately 780-fold ($4.0649/0.0052$) and 80-fold ($0.1793/0.0022$) improvements over the best-performing CP-based model for non-zero and zero locations, respectively. 
The GLS model performs comparably to the GLNS model at non-zero locations, with a slightly higher MSE likely due to the bias from the sparsity-inducing PING prior. 
However, at zero locations, GLS outperforms GLNS, leveraging the advantages of the PING prior. 
Figures \ref{fig: sim betas true} and \ref{fig: sim betas esti} in the Supplementary Materials illustrate the HOSVD-GLS method's ability to recover the true $\bbeta$ values.

\ul{Remarkably, thanks to our semi-orthogonality constraints allowing parallel updates, 
our proposed HOSVD-based models with starting $\balpha$ and $\bbeta$ ranks (3,5,5,5) and (3,5,5,5,4) achieve significantly lower computation times than the CP model with ranks $r_{\alpha}=r_{\beta}=10$ or higher, 
while delivering much greater estimation accuracy.} 
Specifically, the rank-10 CP model takes approximately 1.5 times longer than the above-specified HOSVD model. 
CP models with higher ranks were thus not considered due to their prohibitive costs.

Here we focused again on the recovery of the longitudinal changes, captured by the $\beta$ estimates. 
Similar order-of-magnitude improvements were also obtained for the initial values, captured by the $\alpha$ estimates, as well but are omitted for brevity.

\vspace*{-3ex} 
\section{Discussion} \label{sec: discussion}
\vspace*{-1ex} 
\paragraph{Summary.} 
We introduced a flexible longitudinal mixed model for ultra-high-resolution multidimensional functional data, integrating tensor factorization and basis function expansion to address key methodological challenges. 
The model captured smooth spatial and temporal variations, accounted for group differences, and accommodated subject heterogeneity with irregular observation times, leading to a novel covariance-tensor decomposition. 

We developed efficient Bayesian estimation and inference strategies, including compact HOSVD-based models with semi-orthogonal mode matrices, a Laplacian-operator-based prior with low-rank approximation, and a product Gaussian architecture for smoothness and sparsity. 
A cumulative shrinkage prior enabled semi-automated rank selection, while our computational framework ensured scalable MCMC-based posterior sampling. 

We provided theoretical results on model flexibility and posterior convergence, and the approach is computationally more efficient than CP models of similar rank due to parallel core tensor sampling. 
The method outperformed existing alternatives and provided novel insights into localized brain changes associated with disease progression in ADNI-3 data.

\paragraph{Broader Applicability and Extensions.} 
While we focused on the analysis of DW-MRI data here, 
the proposed methodology is broadly applicable to other types of longitudinal multi-dimensional data with functional modes such as, e.g., in climate science, EEG, and fMRI. 
%
Our ongoing research includes adapting the proposed ideas to the exploration of other important problems involving tensor objects in neuroimaging and other settings, 
e.g., scalar and vector-on-tensor regression models. 
Performing a joint analysis of the different markers and integrating subject-specific covariates, such as baseline age or gender, into the framework is also part of our plans. 

  \section*{Funding}
First author would like to thank the National Science Foundation for the Collaborative Research Grant DMS-2210281.

\section*{Acknowledgment}
The authors gratefully acknowledge the Editor, Associate Editor, and reviewers for their valuable feedback. This work utilized data from the Alzheimer’s Disease Neuroimaging Initiative (ADNI) database (adni.loni.usc.edu). The ADNI investigators were responsible for the design and implementation of the study and/or provided data, but they did not participate in the analysis or writing of this manuscript. A full list of ADNI investigators is available at: \url{https://adni.loni.usc.edu/wp-content/uploads/how_to_apply/ADNI_Acknowledgement_List.pdf}. 

Data collection and sharing for this project was funded by the Alzheimer's Disease Neuroimaging Initiative (ADNI) (National Institutes of Health Grant U01 AG024904) and DOD ADNI (Department of Defense award number W81XWH-12-2-0012). ADNI is funded by the National Institute on Aging, the National Institute of Biomedical Imaging and Bioengineering, and through generous contributions from the following: AbbVie, Alzheimer's Association; Alzheimer's Drug Discovery Foundation; Araclon Biotech; BioClinica, Inc.; Biogen; Bristol-Myers Squibb Company; CereSpir, Inc.; Cogstate; Eisai Inc.; Elan Pharmaceuticals, Inc.; Eli Lilly and Company; EuroImmun; F. Hoffmann-La Roche Ltd and its affiliated company Genentech, Inc.; Fujirebio; GE
Healthcare; IXICO Ltd.; Janssen Alzheimer's Immunotherapy Research \& Development, LLC.; Johnson \& Johnson Pharmaceutical Research \& Development LLC.; Lumosity; Lundbeck; Merck \& Co., Inc.; Meso
Scale Diagnostics, LLC.; NeuroRx Research; Neurotrack Technologies; Novartis Pharmaceuticals Corporation; Pfizer Inc.; Piramal Imaging; Servier; Takeda Pharmaceutical Company; and Transition
Therapeutics. The Canadian Institutes of Health Research is providing funds to support ADNI clinical sites in Canada. Private sector contributions are facilitated by the Foundation for the National Institutes of Health ({\tt www.fnih.org}). The grantee organization is the Northern California Institute for Research and Education,
and the study is coordinated by the Alzheimer's Therapeutic Research Institute at the University of Southern California. ADNI data are disseminated by the Laboratory for Neuro Imaging at the University of Southern California.

\baselineskip 16pt
\vspace*{-10pt}
\section*{Supplementary Materials}
The Supplementary Materials provide detailed information on various aspects of the study, including 
\ul{additional information on the ADNI-3 study, 
the choice of hyperparameters, 
proofs of theoretical results, 
the MCMC algorithm used for posterior sampling, 
additional figures and tables, 
and other relevant details}. 
R programs implementing our methods are provided in a separate ZIP file included with the Supplementary Materials. 

\vspace*{-15pt}



\clearpage\pagebreak\newpage
\pagestyle{fancy}
\fancyhf{}
\rhead{\bfseries\thepage}
\lhead{\bfseries Supplementary Materials}

\setcounter{equation}{0}
\setcounter{page}{1}
\setcounter{table}{1}
\setcounter{figure}{0}
\setcounter{section}{0}
\numberwithin{table}{section}
\renewcommand{\theequation}{S.\arabic{equation}}
\renewcommand{\thesubsection}{S.\arabic{section}.\arabic{subsection}}
\renewcommand{\thesection}{S.\arabic{section}}
\renewcommand{\thepage}{S.\arabic{page}}
\renewcommand{\thetable}{S.\arabic{table}}
\renewcommand{\thefigure}{S.\arabic{figure}}
\baselineskip=25pt

\baselineskip=25pt

\begin{center}
{\LARGE Supplementary Materials for\\
{\bf Bayesian Semiparametric\\ 
\vskip -10pt 
Orthogonal Tucker Factorized Mixed Models\\ 
for 
Multi-dimensional Longitudinal Functional Data
}}
\end{center}
\vskip 20pt 
\baselineskip 16pt

\vskip 2mm
\begin{center}
 Arkaprava Roy\\
 arkaprava.roy@ufl.edu\\
  Department of Biostatistics,
 University of Florida\\
 2004 Mowry Road, Gainesville, FL  32611, USA\\
 \vskip 2mm%
 Abhra Sarkar\\
 abhra.sarkar@utexas.edu \\
 Department of Statistics and Data Sciences,
 The University of Texas at Austin\\
 2317 Speedway D9800, Austin, TX 78712-1823, USA\\
 \vskip 2mm%
 and\\
 The Alzheimer Disease Neuroimaging Initiative
\end{center}

\vskip 10mm
The Supplementary Materials provide detailed information on various aspects of the study, 
including additional information on the ADNI-3 study, 
the choice of hyperparameters, 
proofs of theoretical results, 
the MCMC algorithm used for posterior sampling, 
additional figures and tables, 
and other relevant details. 
R programs implementing our methods are included in a separate ZIP file.

\clearpage\newpage 
\section{Scientific Background} \label{sec: sm sci bg}
\vspace*{-1ex}
\paragraph{Diffusion MRI.}
By detecting the movement of water molecules in tissues, DW-MRI allows us to characterize the properties of white matter bundles \citep{soares2013hitchhiker}.
The properties of the underlying diffusion process of the water are then learned by applying different diffusion models such as diffusion tensor imaging (DTI) \citep{soares2013hitchhiker}, diffusion kurtosis imaging (DKI) \citep{jensen2005diffusional}, neurite orientation dispersion and density imaging (NODDI) \citep{zhang2012noddi}, etc.
Different methods require different data acquisition protocols and provide different micro-structural properties.
Among these, DTI requires the simplest acquisition protocols and is one of the most widely used tools in clinical applications.
Specifically, the DTI model estimates $3 \times 3$ symmetric positive definite matrices, called `diffusion tensors', which represent the covariance of the local 3D Brownian motion. 
It is thus useful to measure the anisotropic diffusion properties of water molecules.
One such well-known marker is fractional anisotropy (FA) which is defined based on the eigenvalues of the diffusion tensors.
Unlike the other diffusion models, NODDI is a biophysical model that characterizes DW-MRI data in terms of biophysically meaningful parameters, 
extracting features such as orientation dispersion index (ODI),  Neurite Density Index (NDI), and Free Water Fraction (FWF), 
which have been widely used in clinical applications \citep{kamiya2020noddi}.
However, it requires collecting the DW-MRI data at different magnetic strengths or shells under a high-angular-resolution diffusion imaging (HARDI) protocol.

\begin{figure}[htbp]
\centering
\vspace*{-10pt}
\includegraphics[width = 0.5\textwidth, trim=0.5cm 4cm 1cm 4cm, clip=true]{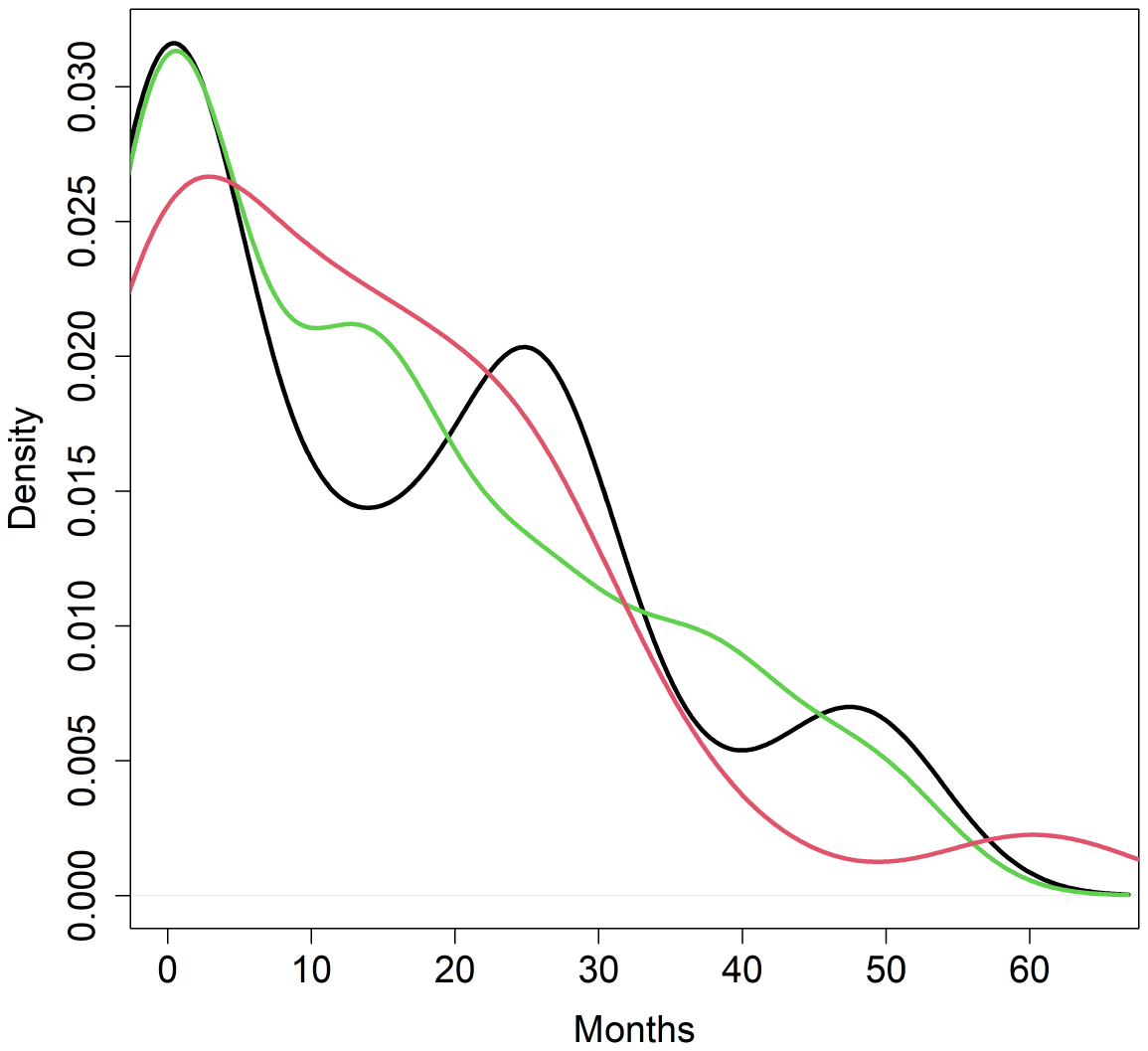}\quad
\caption{Density plots of the observation times for subjects with normal cognition (NC, blue),  mild cognitive impairment (MCI, green), and Alzheimer's (AD, red). 
\vspace*{-5pt}
}
\label{fig: obs_time}
\end{figure}

\paragraph{The ADNI-3 Data Set.}
ADNI collects data in various phases. 
The most recently concluded phase ADNI-3 lasted 2016--2022. 
It included participants from $3$ different groups: 
Healthy control subjects with normal cognition (NC), those with mild cognitive impairment (MCI), and individuals with AD. 
These participants were carefully selected and followed over time with 45 NC, 33 MCI, and 23 AD subjects. 
Figure \ref{fig: obs_time} depicts the distribution of the observation times for these groups. 
Figure \ref{fig: visit times} below shows the observation times for all subjects in the study.
In all, we have 264 high-resolution images from a total of $101$ subjects from the 3 groups.

Importantly, unlike in previous phases, ADNI-3 DW-MRI data are multishell HARDI. 
Thus, it allows us to apply both DTI and NODDI diffusion models to extract various microstructural features. 
NODDI 
has been observed to correlate with FA \citep{fukutomi2019diffusion}. 
Studies have shown that FA is a good marker for longitudinal progression among AD subjects \citep{mayo2017longitudinal, chen2020reliability, hall2021using}. 
Similar results also exist for NODDI-extracted features \citep{colgan2016application,lehmann2021longitudinal}. 
However, most of these works are based on region-wise summary statistics. 
In this work, we run whole-brain longitudinal analyses of FA, ODI, NDI, and FWF. 
We then compare the results and examine the differences in their respective longitudinal progressions.

\paragraph{Data Preprocessing.}
The raw DW-MRI images were first corrected for eddy currents using the FSL software \citep{jenkinson2012fsl}, followed by brain mask removals. 
Then, the fractional anisotropy (FA) images were obtained by applying the DTI model using the {\it dtifit} function and the TBSS pipeline in FSL. 
The three additional outcomes orientation dispersion index (ODI), Neurite Density Index (NDI), and Free Water Fraction (FWF) are obtained by applying the NODDI \citep{zhang2012noddi} model which was fitted using the AMICO package \citep{daducci2015accelerated}. 
Finally, we register the subject-specific estimates on a standard MNI152\_2mm template, producing $90 \times 108 \times 90$ dimensional images for each case. 
Since statistical analyses of neuroimaging data can be adversely influenced by empty voxels outside the brain, 
we followed \cite{feng2021brain} and 
cropped out these parts, 
resulting in images of dimension $60\times 70\times 60$. 
We also applied a linear transformation to map all the values of all outcomes between $[-5,5]$. 
The observation times are scaled to [0,1] by dividing observed times with the highest observation time of our study cohort. 

\begin{figure}[htbp]
     \centering
     \includegraphics[width=0.8\textwidth, height=15cm, trim=0cm 0cm 0cm 0.5cm, clip]{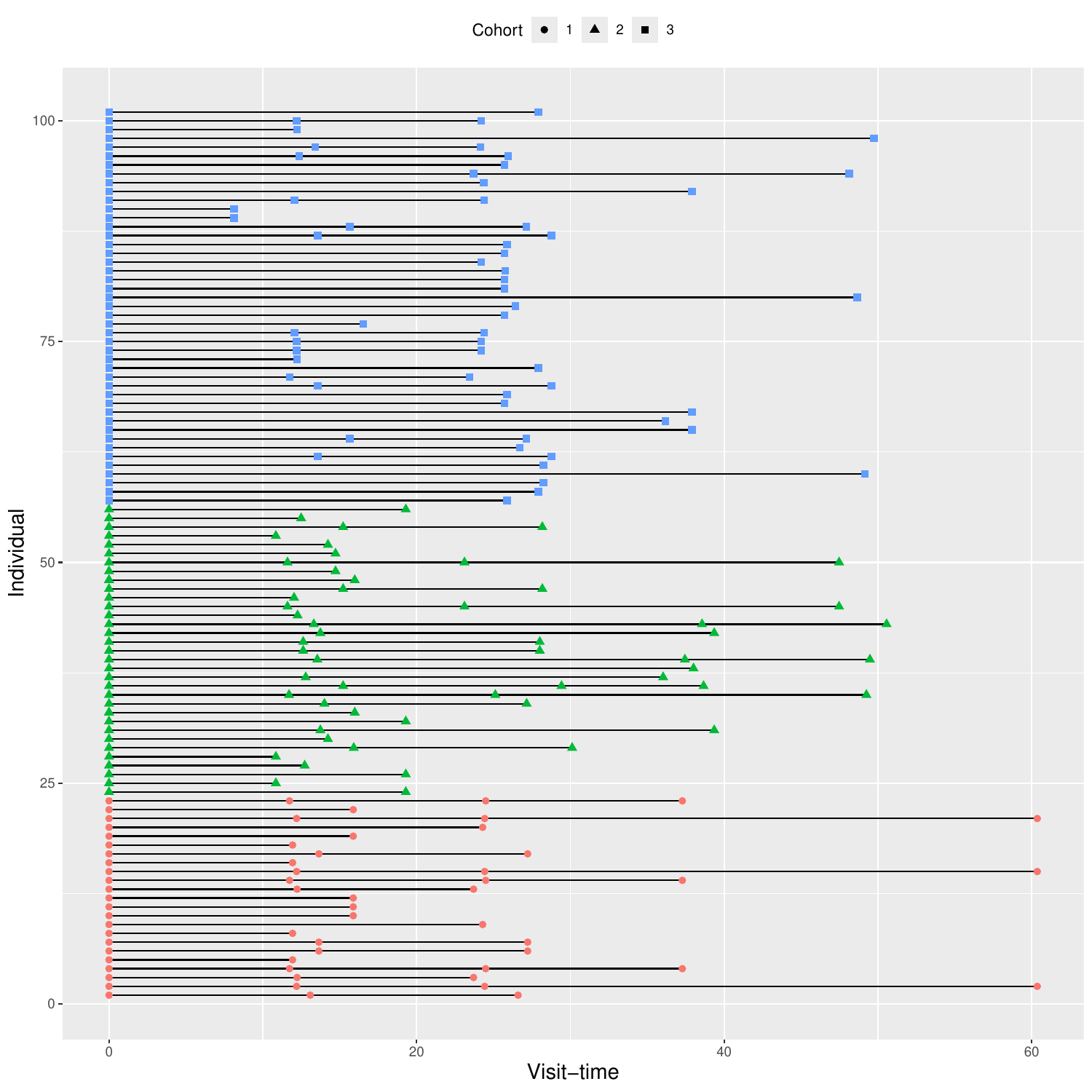}
     \caption{Observation times for the three groups: 
     Alzheimer's (AD, Cohort 1, red circles), 
     mild cognitive impairment (MCI, Cohort 2, green triangles), 
     and normal cognition (NC, Cohort 3, blue squares). 
     }
     \label{fig: visit times}
\end{figure}

\clearpage\newpage\section{Additional Discussions on Hyper-parameters} \label{sec: sm hyper-parameters1}
\vspace*{-1ex} 
With some repetition from the paper for completeness and readability, we recall that our implementation involves only a small number of prior hyperparameters: $\kappa_1, \kappa_2$ for cumulative shrinkage; $q_s$ for the PING prior; and $a_s, b_s$, $a_\tau, b_\tau$ for inverse Gamma distributions. We use $\kappa_1 = 2.1$, $\kappa_2 = 3.1$, $a_s = 10$, $b_s = 0.1$ (for all $s$), and $a_\tau = b_\tau = 0.1$. 
For the components of $\bbeta$ ($s = 1, 2, 3$), we set $q_s = 3$, while for $\balpha$, $q_s = 1$ to reflect reduced spatial sparsity in the baseline.
The $\text{Ga}(10, 0.1)$ priors encourage shrinkage, while the $\text{Ga}(0.1, 0.1)$ priors on the global precisions $\tau_\alpha^{-2}$ and $\tau_\beta^{-2}$ are diffuse. We also set the number of B-spline bases $d_t$ to approximate the maximum number of patient visits, and the number of Gaussian bases $m_s = \lfloor d_s / 2 \rfloor$ with spacing $v_s = d_s / m_s$.

The choices for $a_{\tau}, b_{\tau}$ are non-informative.
The choices of $\kappa_{1}$ and $\kappa_{2}$ are as recommended by \cite{durante2017note}. 
The choice for $a_{s}$ and $b_{s}$ induces shrinkage that offers numerical stability and is partly motivated by \cite{roy2021perturbed}. 
The choice $q = 3$ for the PING was found to perform best across various settings in \cite{roy2021spatial}. 
Although \cite{roy2021spatial} also explored cross-validation methods for setting $q$, this approach was computationally prohibitive for our purposes. 
Instead, we conducted a preliminary exploration and again found $q_{s} = 3$ to yield very robust performances. 
Our choice of the number of B-splines is supported by \citet{ruppert2002selecting}, who showed that beyond a minimum threshold, the fit quality is largely insensitive to the exact number of splines -- a finding consistent with our own numerical experiments. 
The choice for the number of Gaussian bases 
is also motivated by the numerical experiments from \cite{roy2021spatial} using low-rank approximations, and also found to be working well in our current numerical experiments. 

\vspace*{-3ex} 
\section{Posterior Computation} \label{sec: sm post comp}
\vspace*{-1ex} 


Here we present the full conditional posterior distribution for all the parameters. 
Derivations are omitted for brevity. 

To sample from the subspace constraint multivariate Gaussian, we use the following result that if $\bx\sim\Normal(\bmu,\bSigma)$ conditioned on $\bA\bx=\bc$, the resulting conditional distribution $\bx\bigm|\bA\bx=\bc \sim \Normal(\bmu_{c},\bSigma_{c})$, where $\bmu_{c}=\bmu+\bSigma\bA\trans\bB(\bc-\bA\bmu)$ and $\bSigma_{c}=\bSigma-\bSigma\bA\trans\bB\bA\bSigma$ with $\bB=(\bA\bSigma\bA\trans)^{-1}$. 
The calculation is straightforward, using the results from the conditional multivariate normal distribution. 
Let $\by = \bA \bx$ and suppose $\bx \sim \mathcal{N}(\bmu, \bSigma)$.  
Then the stacked random vector
is jointly Gaussian with distribution
\[
\begin{pmatrix}
\bx \\
\by
\end{pmatrix}
\sim
\mathcal{N}\!\left(
\bmu_f,\,
\bSigma_f
\right),
\]
where the joint mean and covariance matrix are
\[
\bmu_f
=
\begin{pmatrix}
\bmu \\
\bA\bmu
\end{pmatrix},
\qquad
\bSigma_f
=
\begin{pmatrix}
\bSigma & \bSigma\bA^{\top} \\
\bA\bSigma & \bA\bSigma\bA^{\top}
\end{pmatrix}.
\]

Now, we can calculate the distribution of $\bx|\by=\bc$ based on this joint distribution and obtain the above result. 
Furthermore, we exploit the banded precision structure in our prior to develop a more efficient sampler as discussed below. We use Algorithm 2 from \cite{cong2017fast} for sampling from the hyperplane constraint Gaussian with some additional improvements due to banded structures on the precision.

\vspace*{-1ex}
\begin{itemize}[leftmargin=*,itemsep=0em]
    \item {\bf Updating the mode matrix parameters $\bgamma^{(s)}_{\alpha,\ell,k}$:} 
    Let 
    $\bG^{(s)} = [\bg^{(s)}(1),\dots,\bg^{(s)}(d_{s})]$, 
    $\ba^{(s)}_{\alpha,\ell,k}=[a^{(s)}_{\alpha,\ell,k}(1),\dots,a^{(s)}_{\alpha,\ell,k}(d_{s})]\trans$, 
    $\diag(\ba^{(s)}_{\alpha,\ell,-k}) = \diag\{\prod_{j\neq k}a^{(s)}_{\alpha,\ell,j}(1),\dots,\prod_{j\neq k}a^{(s)}_{\alpha,\ell,j}(d_{s})\}$, 
    $\diag\{(\ba^{(s)}_{\alpha,\ell,-k})^{2}\}$ be similarly defined, 
    $\bA^{(s)}_{\alpha} = [\ba^{(s)}_{\alpha,1},\dots,\ba^{(s)}_{\alpha,r_{\alpha,s}}]$, 
    and 
    $\bA^{(s)}_{\alpha,-\ell}$ be $\bA^{(s)}_{\alpha}$ with the $\ell\th$ column removed. 
    Also, 
    \vspace*{-5ex}\\
    \bse
    & \zeta_{\alpha,\ell}^{(i,s)}(h_{g},h_{1},h_{2},h_{3})=\sum_{m \in \S_{g}}\sum_{\{z_{m}: z_{s} = \ell\}} \eta_{\alpha,z_{g},z_{1},z_{2},z_{3}}^{(i)} \prod_{j\in\S_{g}\setminus {s}} a^{(j)}_{\alpha,z_{j}}(h_{j}),\\ 
    %
    %
    & \zeta_{\alpha,-\ell}^{(i,s)}(h_{g},h_{1},h_{2},h_{3})=\sum_{m \in \S_{g}}\sum_{\{z_{m}: z_{s}\neq \ell\}} \eta_{\alpha,z_{g},z_{1},z_{2},z_{3}}^{(i)} \prod_{j\in\S_{g}} a^{(j)}_{\alpha,z_{j}}(h_{j}),\\ 
    & \zeta_{\alpha,\ell,SS}^{(i,s)} = \sum_{m \in \S_{g}\setminus s}\sum_{h_{m}} \{\zeta_{\alpha,\ell}^{(i,s)}(h_{g},h_{1},h_{2},h_{3})\}^{2},\\ 
    & R_{\alpha}^{(i)}(h_{g},h_{1},h_{2},h_{3})=\sum_{t=t_{i,1}}^{t_{i,n_{i}}}\{Y^{(i)}(h_{g},h_{1},h_{2},h_{3},t)-\beta^{(i)}(h_{g},h_{1},h_{2},h_{3},t)\},\\
    & R_{\alpha,\ell}^{(i,s)}(h_{g},h_{1},h_{2},h_{3})=R_{\alpha}^{(i)}(h_{g},h_{1},h_{2},h_{3})-n_{i}\zeta_{\alpha,-\ell}^{(i,s)}(h_{g},h_{1},h_{2},h_{3}),\\
    %
    %
    & \wt{r}_{\alpha,\ell}^{(s)}(h_{s}) = \sum_{i}\sum_{m \in \S_{g}}\sum_{\{h_{m}: m\neq s\}}R_{\alpha,\ell}^{(i,s)}(h_{g},h_{1},h_{2},h_{3})\zeta_{\alpha,\ell}^{(i,s)}(h_{g},h_{1},h_{2},h_{3}),\\
    & \bGamma^{(s)}_{\alpha,\ell,k} = (\bA^{(s)}_{\alpha,-\ell})\trans\diag(\ba^{(s)}_{\alpha,\ell,-k})(\bG^{(s)})\trans.
    \ese
    \vspace*{-5ex}\\
    While the $\zeta_{\alpha,\ell}^{(i,s)}(h_{g},h_{1},h_{2},h_{3})$'s are constant along the $s\th$ direction, i.e., do not vary with $h_{s}$, 
    they are defined as above for notational convenience. 
    Additionally, minor modifications are needed for the case $s=g$ 
    where the sum over $i$ in the definition of $\wt{r}_{\alpha,\ell}^{(g)}(h_{g})$ is taken over the individuals belonging specifically to the $h_{g}\th$ group 
    and not on the entire sample. 
    Similar comments also apply to $\zeta_{\beta,\ell}^{(i,s)}(h_{g},h_{1},h_{2},h_{3},t)$ defined below.

    Next, recalling that the PING process variance $\sigma^{(s)}_{\alpha,\ell}$ is absorbed by the first component, 
    let $\sigma^{(s)}_{\alpha,\ell,k} = \sigma^{(s)}_{\alpha,\ell}$ for $k=1$ and $\sigma^{(s)}_{\alpha,\ell,k}=1$ for $k>1$. 
    Finally, let $\wt{\br}_{\alpha,\ell}^{(s)} = [\wt{r}_{\alpha,\ell}^{(s)}(1),\dots,\wt{r}_{\alpha,\ell}^{(s)}(d_{s})]\trans$. 
    The full conditionals of the factor matrix parameters are then given by 
    \vspace*{-5ex}\\
    \bse
    & \bgamma^{(s)}_{\alpha,\ell,k} \sim \MVN(\bmu^{(s)}_{\alpha,\ell,k},\bSigma^{(s)}_{\alpha,\ell,k} \mid \bGamma^{(s)}_{\alpha,\ell,k}\bgamma^{(s)}_{\alpha,\ell,k}=\bzero), ~\text{where}\\
    & \bSigma^{(s)}_{\alpha,\ell,k} = \bigg[\sum_{i} \frac{1}{\sigma^{2}} n_{i} \zeta_{\alpha,\ell,SS}^{(i,s)} \bG^{(s)}\diag\{(\ba^{(s)}_{\alpha,\ell,-k})^{2}\}(\bG^{(s)})\trans + \frac{1}{(\sigma^{(s)}_{\alpha,\ell,k})^{2}}(\bQ^{(s)})^{-1}\bigg]^{-1}, \\
    & \text{and}~~\bmu^{(s)}_{\alpha,\ell,k}=\sigma^{-2}\bSigma^{(s)}_{\alpha,\ell,k}\bG^{(s)}\diag(\ba^{(s)}_{\alpha,\ell,-k}) \wt{\br}_{\alpha,\ell}^{(s)}.
    \ese
    \vspace*{-5ex}\\    
    Note that, by our construction, $\bG^{(s)}$ and $(\bQ^{(s)})^{-1}$ are both banded matrices. 
    Sampling the $\bgamma^{(s)}_{\alpha,\ell,k}$'s, therefore, requires sampling from multivariate normal distributions with banded precision matrices subject to some linear constraints. 
    This is done using the efficient banded-Cholesky algorithm of \citet[Section 3.1.3]{rue2001fast}. 


    \item {\bf Updating the core tensor parameters $\etam^{(i)}_{\alpha}$:} 
    The posterior full conditional of the core tensor parameters $\eta_{\alpha,z_{g},z_{1},z_{2},z_{3}}^{(i)}$ is given by 
    \vspace*{-5ex}\\
    \bse
    & \eta_{\alpha,z_{g},z_{1},z_{2},z_{3}}^{(i)} \sim \Normal(\mu_{\alpha,\eta,z_{g},z_{1},z_{2},z_{3}}^{(i)},\sigma_{\alpha,\eta,z_{g},z_{1},z_{2},z_{3}}^{2}),~\text{where}\\ 
    & \sigma_{\alpha,\eta,z_{g},z_{1},z_{2},z_{3}}^{2} = \bigg[\frac{1}{\sigma^{2}}\{\prod_{s \in \S_{g}} a^{(s)}_{\alpha,z_{s}}(h_{s})\}^{2}+\frac{1}{\tau^{2}_{\alpha}\sigma_{\alpha,z_{g},\ell,z_{2},z_{3}}^{2}}\bigg]^{-1},~ \text{and}\\
    & \frac{\mu_{\alpha,\eta,z_{g},z_{1},z_{2},z_{3}}^{(i)}}{\sigma_{\alpha,\eta,z_{g},z_{1},z_{2},z_{3}}^{2}} = \left\{\frac{1}{\sigma^{2}} \sum_{m \in \S_{g}} \sum_{h_{m}} R_{\alpha}^{(i)}(h_{g},h_{1},h_{2},h_{3})\prod_{s \in \S_{g}} a^{(s)}_{\alpha,z_{s}}(h_{s}) +\frac{c_{\alpha,z_{g},z_{1},z_{2},z_{3}}}{\tau^{2}_{\alpha}\sigma_{\alpha,z_{g},z_{1},z_{2},z_{3}}^{2}}\right\}.
    \ese   
    \vspace*{-5ex}\\
    Importantly, due to the imposed orthogonality constraints on the factor matrices, these full conditionals are mutually independent across all combinations of $(i,z_{g},z_{1},z_{2},z_{3})$, allowing them to be sampled in an embarrassingly parallel manner. 
\end{itemize}


\vspace*{-3ex}
\begin{itemize}[leftmargin=*,itemsep=0em]
    \item {\bf Updating the mode matrix parameters $\bgamma^{(s)}_{\beta,\ell,k}$:} 
    Let 
    $\ba^{(s)}_{\beta,\ell,k}=[a^{(s)}_{\beta,\ell,k}(1),\dots,a^{(s)}_{\beta,\ell,k}(d_{s})]\trans$, 
    $\diag(\ba^{(s)}_{\beta,\ell,-k}) = \diag\{\prod_{j\neq k}a^{(s)}_{\beta,\ell,j}(1),\dots,\prod_{j\neq k}a^{(s)}_{\beta,\ell,j}(d_{s})\}$, 
    $\diag\{(\ba^{(s)}_{\beta,\ell,-k})^{2}\}$ be similarly defined, 
    $\bA^{(s)}_{\beta} = [\ba^{(s)}_{\beta,1},\dots,\ba^{(s)}_{\beta,r_{\beta,s}}]$, 
    and 
    $\bA^{(s)}_{\beta,-\ell}$ be $\bA^{(s)}_{\beta}$ with the $\ell\th$ column removed. 
    %
    Also, let $\bb^{(i)}_{t,\cdot}$ be the $t\th$ row of $\bB^{(i)}$ defined in Section \ref{sec: semi-orthogonality}. 
    We handle the cases $s\neq t$ and $s=t$ separately.
    
    First, for the cases $s \in \S_{g}$, let 
    \vspace*{-5ex}\\
    \bse
    & \zeta_{\beta,\ell}^{(i,s)}(h_{g},h_{1},h_{2},h_{3},t)=\sum_{m \in \S_{g,t}}\sum_{\{z_{m}: z_{s} = \ell\}} \eta_{\beta,z_{g},z_{1},z_{2},z_{3},z_{t}}^{(i)} \left\{ \prod_{j\in\S_{g} \setminus s} a^{(j)}_{\beta,z_{j}}(h_{j}) \right\} (\bb^{(i)}_{t,\cdot}\ba^{(t)}_{\beta,z_{t}}),\\ 
    %
    %
    & \zeta_{\beta,-\ell}^{(i,s)}(h_{g},h_{1},h_{2},h_{3},t)=\sum_{m \in \S_{g,t}}\sum_{\{z_{m}: z_{s}\neq \ell\}} \eta_{\beta,z_{g},z_{1},z_{2},z_{3},z_{t}}^{(i)} \left\{\prod_{j\in\S_{g}} a^{(j)}_{\beta,z_{j}}(h_{j})\right\}(\bb^{(i)}_{t,\cdot}\ba^{(t)}_{\beta,z_{t}}),\\ 
    & \zeta_{\beta,\ell,SS}^{(i,s)} = \sum_{t=t_{i,1}}^{t_{i,n_{i}}}\sum_{m \in \S_{g}\setminus s}\sum_{h_{m}} \{\zeta_{\beta,\ell}^{(i,s)}(h_{g},h_{1},h_{2},h_{3},t)\}^{2},\\ 
    %
    %
    & R_{\beta,\ell}^{(i,s)}(h_{g},h_{1},h_{2},h_{3},t) = Y^{(i)}(h_{g},h_{1},h_{2},h_{3},t)-\alpha^{(i)}(h_{g},h_{1},h_{2},h_{3})-\zeta_{\beta,-\ell}^{(i,s)}(h_{g},h_{1},h_{2},h_{3},t),\\
    & \wt{r}_{\beta,\ell}^{(s)}(h_{s}) = \{\sum_{i}\sum_{t=t_{i,1}}^{t_{i,n_{i}}}\sum_{m \in \S_{g}}\sum_{\{h_{m}: m\neq s\}}R_{\beta,\ell}^{(i,s)}(h_{g},h_{1},h_{2},h_{3},t)\zeta_{\beta,\ell}^{(i,s)}(h_{g},h_{1},h_{2},h_{3},t)\},\\
    & \bGamma^{(s)}_{\beta,\ell,k} = (\bA^{(s)}_{\beta,-\ell})\trans \diag(\ba^{(s)}_{\beta,\ell,-k})(\bG^{(s)})\trans.
    \ese
    \vspace*{-5ex}\\
    Additionally, as before, let $\sigma^{(s)}_{\beta,\ell,k} = \sigma^{(s)}_{\beta,\ell}$ for $k=1$ and $\sigma^{(s)}_{\beta,\ell,k}=1$ for $k>1$; 
    and $\wt{\br}_{\beta,\beta,\ell}^{(s)} = [\wt{r}_{\beta,\ell}^{(s)}(1),\dots,\wt{r}_{\beta,\ell}^{(s)}(d_{s})]\trans$. 
    The mode matrix parameters $\bgamma^{(s)}_{\beta,\ell,k}$ are then updated as 
    \vspace*{-5ex}\\
    \bse
    & \bgamma^{(s)}_{\beta,\ell,k} \sim \MVN(\bmu^{(s)}_{\beta,\ell,k},\bSigma^{(s)}_{\beta,\ell,k} \mid \bGamma^{(s)}_{\beta,\ell,k}\bgamma^{(s)}_{\beta,\ell,k}=\bzero), ~\text{where}\\
    & \bSigma^{(s)}_{\beta,\ell,k} = \bigg[\sum_{i} \frac{1}{\sigma^{2}} n_{i} \zeta_{\beta,\ell,SS}^{(i,s)} \bG^{(s)}\diag\{(\ba^{(s)}_{\beta,\ell,-k})^{2}\}(\bG^{(s)})\trans + \frac{1}{(\sigma^{(s)}_{\beta,\ell,k})^{2}}(\bQ^{(s)})^{-1}\bigg]^{-1}, \\
    & \text{and}~~\bmu^{(s)}_{\beta,\ell,k}=\sigma^{-2}\bSigma^{(s)}_{\beta,\ell,k}\bG^{(s)}\diag(\ba^{(s)}_{\beta,\ell,-k}) \wt{\br}_{\beta,\ell}^{(s)}.
    \ese
    \vspace*{-5ex}\\
    Next, for the case $s=t$, let 
    \vspace*{-5ex}\\
    \bse
    & \zeta_{\beta,\ell}^{(i,t)}(h_{g},h_{1},h_{2},h_{3},t) = \sum_{m \in \S_{g}}\sum_{z_{m}} \eta_{\beta,z_{g},z_{1},z_{2},z_{3},\ell}^{(i)} \left\{\prod_{j\in\S_{g}} a^{(j)}_{\beta,z_{j}} (h_{j})\right\}, \\
    & \zeta_{\beta,-\ell}^{(i,t)}(h_{g},h_{1},h_{2},h_{3},t) = \sum_{m \in \S_{g,t}}\sum_{\{z_{m}: z_{t} \neq \ell\}} \eta_{\beta,z_{g},z_{1},z_{2},z_{3},z_{t}}^{(i)} \left\{\prod_{j\in\S_{g}} a^{(j)}_{\beta,z_{j}} (h_{j})\right\}(\bb^{(i)}_{t,\cdot}\ba^{(t)}_{\beta,z_{t}}), \\
    & \zeta_{\beta,\ell,SS}^{(i,t)} = \sum_{m \in \S_{g}}\sum_{h_{m}} \{\zeta_{\beta,\ell}^{(i,t)}(h_{g},h_{1},h_{2},h_{3},t)\}^{2},\\
    & R_{\beta,\ell}^{(i,t)}(h_{g},h_{1},h_{2},h_{3},t)=Y^{(i)}(h_{g},h_{1},h_{2},h_{3},t)-\alpha^{(i)}(h_{g},h_{1},h_{2},h_{3})-\zeta_{\beta,-\ell}^{(i,t)}(h_{g},h_{1},h_{2},h_{3},t),\\
    &\wt{r}_{\beta,\ell}^{(i,t)}(t) = \sum_{m \in \S_{g}}\sum_{h_{m}}R_{\beta,\ell}^{(i,t)}(h_{g},h_{1},h_{2},h_{3},t)\zeta_{\beta,\ell}^{(i,t)}(h_{g},h_{1},h_{2},h_{3}).
    \ese
    \vspace*{-5ex}\\
    %
    %
    When $s=t$, there is only one component $\btheta_{\beta,\ell,1}^{(t)} = \ba_{\beta,\ell}^{(t)}$ (Section \ref{sec: priors}) which are updated as 
    \vspace*{-5ex}\\
    \bse
    & \btheta_{\beta,\ell,1}^{(t)} \sim \MVN(\bmu^{(t)}_{\beta,\ell,1},\bSigma^{(t)}_{\beta,\ell,1} \mid \bGamma^{(t)}_{\beta,\ell,1}\ba_{\beta,\ell}^{(t)}=\bzero), ~\text{where}\\
    & \bSigma^{(t)}_{\beta,\ell,1} = \bigg[\sum_{i} \frac{1}{\sigma^{2}} \zeta_{\beta,\ell,SS}^{(i,t)} (\bB^{(i)})\trans\bB^{(i)} + \frac{1}{(\sigma^{(s)}_{\beta,\ell,1})^{2}}(\bQ^{(s)})^{-1}\bigg]^{-1}, \\
    & \text{and}~~\bmu^{(t)}_{\beta,\ell,1}=\sigma^{-2}\bSigma^{(t)}_{\beta,\ell,1}\{\sum_{i}(\bB^{(i)})\trans \wt{\br}_{\beta,\ell}^{(i,t)}\},
    \ese
    \vspace*{-5ex}\\
    where $\bGamma^{(t)}_{\beta,\ell,1}=\bA^{(t)}_{\beta,-\ell}$ and $\wt{\br}_{\beta,\ell}^{(i,t)}=[\wt{r}_{\beta,\ell}^{(i,t)}(t_{i,1}),\dots,\wt{r}_{\beta,\ell}^{(i,t)}(t_{i,n_i})]\trans$.

Additionally, we incorporate an adaptive sampling step from \cite{bhattacharya2011sparse}, enabling the dynamic removal of redundant columns or the addition of new ones, automating rank selection.
    
    \item {\bf Updating the core tensor parameters $\etam^{(i)}_{\beta}$:} Unlike $\etam^{(i)}_{\alpha}$, the posteriors of core tensor elements are not conditionally independent across the different values of $z_{t}$. 
    However, given the factor matrices, for a fixed $z_{t}=\ell$, the posteriors of $\eta_{\beta,z_{g},z_{1},z_{2},z_{3},\ell}^{(i)}$ are conditionally independent across the indices $(z_{g},z_{1},z_{2},z_{3})$. 
    The posterior full conditional of the core tensor parameters $\eta_{\beta,z_{g},z_{1},z_{2},z_{3},\ell}^{(i)}$ is given by 
    \vspace*{-5ex}\\
    \bse
    & \eta_{\beta,z_{g},z_{1},z_{2},z_{3},\ell}^{(i)} \sim \Normal(\mu_{\beta,\eta,z_{g},z_{1},z_{2},z_{3},\ell}^{(i)},\sigma_{\beta,\eta,z_{g},z_{1},z_{2},z_{3},\ell}^{2}),~\text{where}\\ 
    & \sigma_{\beta,\eta,z_{g},z_{1},z_{2},z_{3},\ell}^{2} = \bigg[\frac{1}{\sigma^{2}}\{\prod_{j\in\S_{g}} a^{(j)}_{\alpha,z_{j}}(h_{j}) (\bb^{(i)}_{t,\cdot}\ba^{(t)}_{\beta,\ell})\}^{2}+\frac{1}{\tau^{2}_{\beta}\sigma_{\beta,z_{g},\ell,z_{2},z_{3},\ell}^{2}}\bigg]^{-1},~ \text{and}\\
    %
    %
    %
    &\displaystyle\frac{\mu_{\beta,\eta,z_{g},z_{1},z_{2},z_{3},\ell}^{(i)}} {\sigma_{\beta,\eta,z_{g},z_{1},z_{2},z_{3},\ell}^{2}} = \left\{\frac{1}{\sigma^{2}} \sum_{m \in \S_{g,t}} \sum_{h_{m}} R_{\beta, \ell}^{(i,t)}(h_{g},h_{1},h_{2},h_{3},t) \prod_{s \in \S_{g}} a^{(s)}_{\alpha,z_{s}}(h_{s}) (\bb^{(i)}_{t,\cdot}\ba^{(t)}_{\beta,\ell})  +\frac{c_{\beta,z_{g},z_{1},z_{2},z_{3},\ell}}{\tau^{2}_{\beta}\sigma_{\beta,z_{g},z_{1},z_{2},z_{3},\ell}^{2}}\right\}.
    \ese   
    \vspace*{-4ex}\\
    For each fixed $\ell$, the $\eta_{\beta,z_{g},z_{1},z_{2},z_{3},\ell}^{(i)}$'s  are again sampled in an embarrassingly parallel manner.


\end{itemize}


\clearpage\newpage
\section{Proofs of Theoretical Results} \label{sec: sm proofs}

\subsection{Proof of Lemma \ref{lem: Tucker-HOSVD}} 
Consider a Tucker decomposition of a $d_{1} \times \dots \times d_{p}$ tensor $\btheta=\etam \times_{1} \bA^{(1)} \times_{2} \bA^{(2)} \cdots \times_{p}\bA^{(p)}$ with an $r_{1} \times \dots \times r_{p}$ dimensional core tensor $\etam$. 
Let $\bA^{(j)}=\bQ^{(j)}\bR^{(j)}$ be the QR decomposition of the $j\th$ mode matrix, 
where $\bQ^{(j)}$ is a $d_{j} \times r_{j}$ semi-orthogonal matrix and $\bR^{(j)}$ is an $r_{j}\times r_{j}$ upper triangular matrix. 
Then, 
with a new core-tensor $\etam'=\etam \times_{1}\bR^{(1)}\times_{2}\bR^{(2)}\cdots\times_{p}\bR^{(p)}$, which is of the same dimension as the original $\etam$, 
$\btheta$ can be alternatively represented as $\btheta=\etam'\times_{1}\bQ^{(1)}\times_{2} \bQ^{(2)}\cdots\times_{p}\bQ^{(p)}$, which is a compact HOSVD.

\subsection{Proofs of Lemmas \ref{thm::largesupping} and \ref{thm::largesupcsc}}
\begin{proof}[Proof of Lemma \ref{thm::largesupping}]
First, let $q$ be an odd integer and $\bbeta=\prod_{k=1}^{q}\bbeta_{k}$. 
We show that, for a given $\epsilon > 0$, with a properly chosen $\kappa = \kappa (\epsilon)$, 
$\{\big\|\bbeta-\bbeta_{\star}\big\|_{\infty}<\epsilon\}\supseteq \{ \|\bbeta_{k}-\bbeta_{\star}^{1/q}\|_{\infty} < \kappa: k=1,\ldots,q\}$. 
We have 
$\prod_{k=1}^{q}\bbeta_{k}=\prod_{k=1}^{q}(\bbeta_{k}-\bbeta_{\star}^{1/q}+\bbeta_{\star}^{1/q})=\bbeta_{\star} + f(\bbeta_{0}^{1/q}-\bbeta_{1},\ldots,\bbeta_{\star}^{1/q}-\bbeta_{q}, \bbeta_{\star}^{1/q})$.
Hence, $\|\bbeta-\bbeta_{\star}\|_{\infty}\le \|f(\bbeta_{\star}^{1/q}-\bbeta_{1},\ldots,\bbeta_{0}^{1/q}-\bbeta_{q}, \bbeta_{\star}^{1/q})\|_{\infty} $.
Now, if $\|\bbeta_{k} - \bbeta_{\star}^{1/q}\|_{\infty}<\kappa$ for all $k$, then $\|f(\bbeta_{\star}^{1/q}-\bbeta_{1},\ldots,\bbeta_{0}^{1/q}-\bbeta_{q}, \bbeta_{\star}^{1/q})\|_{\infty}\leq \kappa \|\bbeta_{\star}\|_{\infty}$ when $\kappa < 1$. 
For the Gaussian prior distribution, $\Pi(\|\bbeta_{0}^{1/q}-\bbeta_{k}\|_{\infty}<\kappa)>0$. 
Setting $\kappa \|\bbeta_{\star}\|_{\infty} \leq \epsilon$ completes the proof.

Next, if $q$ is even, we can write $\bbeta_{\star}=\prod_{k=1}^{q-1}\bbeta_{\star}^{1/(q-1)}$. Then, $\{\big\|\bbeta-\bbeta_{\star}\big\|_{\infty}<\epsilon\}\supseteq \{\|\bbeta_{q}-\bone\|_{\infty} < \kappa, \|\bbeta_{k}-\bbeta_{\star}^{1/(q-1)}\|_{\infty} < \kappa: k=1,\ldots,q-1\}$. Now, subsequent steps will be similar to the previous case.
\end{proof}

\begin{proof}[Proof of Lemma \ref{thm::largesupcsc}]
We prove this by induction. 
To do that, we first prove the assertion for $\bA$ and $\bA_{\star}$ having 2 columns as $\bA=[\ba_{1};\ba_{2}]$ and $\bA_{\star}=[\ba_{*,1};\ba_{*,2}]$. 

Let $\bP_{*,1}=\ba_{*,1} (\ba_{*,1}\trans\ba_{*,1})^{-1}\ba_{*,1}\trans$ and $\bP_{1}$ be defined similarly. 
Then by construction $\bP_{*,1}\ba_{*,2}=0$ and $\bP_{1}\ba_{2}=0$. We have $\|\bA-\bA_{\star}\|_{\infty}\leq \|\ba_{1}-\ba_{*,1}\|_{\infty} + \|(\bI-\bP_{1})\ba_{*,2}-\ba_{2}\|_{\infty} + \|(\bP_{*,1}-\bP_{1})\ba_{*,2}\|_{\infty}$. 
We further have $\|\bP_{*,1}-\bP_{1}\|_{\infty} \leq \|\ba_{1}-\ba_{*,1}\|_{\infty}(\ba_{*,1}\trans\ba_{*,1})^{-1} +  |(\ba_{1}\trans\ba_{1})^{-1}-(\ba_{*,1}\trans\ba_{*,1})^{-1}|\|\ba_{1}\|_{\infty}$. 
Hence, we need 1) concentration of $\ba_{1}$ around $\ba_{*,1}$ and 2) concentration of $\ba_{2}$ around $(\bI-\bP_{1})\ba_{*,2}$. 
As $\ba_{1}$ concentrates more and more around $\ba_{*,1}$, the term $\|(\bP_{*,1}-\bP_{1})\ba_{*,2}\|_{\infty}$ also decreases to zero. 
Finally, by construction of the CSC prior, $\{(\bI-\bP_{1})\bx:\bx\in\mathbb{R}^{p}\}$ is in the support of $\ba_{2}$. 
Thus, Lemma~\ref{lem: sup-norm PING} applies to ensure $\Pi(\|(\bI-\bP_{1})\ba_{*,2}-\ba_{2}\|_{\infty}<\epsilon/3)>0$ for a small enough $\epsilon$ (with some minor modifications in the argument under the CSC constraint).
Combining these arguments, we can conclude $\Pi(\|\bA-\bA_{0}\|_{\infty}<\epsilon) > 0$ when $\bA$ and $\bA_{0}$ have 2 columns.

Next, let the assertion hold for $\bA$ and $\bA_{0}$ having $H$ columns. 
Then for matrices with $(H+1)$ columns, we can rewrite the matrices as $\bA=[\bA_{1:H};\ba_{H+1}]$ and $\bA_{0}=[\bA_{*,1:H};\ba_{*,H+1}]$ with a block of the first $H$ columns and the last column. 
Then, we can repeat the above steps for the projection matrix part as $\|\bP_{*,1:H}-\bP_{1:H}\|_{\infty} \leq \|\bA_{1:H}-\bA_{*,1:H}\|_{\infty}\|(\bA_{*,1:H}\trans\bA_{*,1:H})^{-1}\|_{\infty} +  \|(\bA_{1:H}\trans\bA_{1:H})^{-1}-(\bA_{*,1:H}\trans\bA_{*,1:H})^{-1}\|_{\infty}\|\bA_{1:H}\|_{\infty}$. 
Therefore, using Lemma~\ref{lem: sup-norm PING} and by induction, we can conclude $\Pi(\|\bA-\bA_{\star}\|_{\infty}<\epsilon) > 0$ for semi-orthogonal matrices $\bA$ and $\bA_{\star}$ of dimension $p\times (H+1)$, which completes the proof.
\end{proof}

\subsection{Proof Outline of Theorem \ref{thm: consistency}} 
We present first a proof outline here. 
Details are presented separately below. 

With some repetition from the main paper to keep this section relatively self-contained, we recall that we will use $\|\cdot\|_{2}$ to denote the Euclidean norm of a vector. 
For a matrix $\bA$, the Frobenius norm $\|\bA\|_{F}=\|\vect(\bA)\|_{2}$ is the Euclidean norm of its vectorization. 
For both vectors and matrices, $\|\cdot\|_{\infty}$ will denote the maximum absolute value of the entries. 

Additionally, the operator norm of a matrix $\bA$ is given by $\bA_{op}=\sup\{ \|\bA \bx\|_{2}: \|\bx\|_{2}\le 1\}$; and 
the symbol $\asymp$ will stand for equality of the order of growth or decay for two sequences of numbers. 

We denote the joint density of the data by $Q_{\bTheta}$ with  
$\bTheta=\{\etam_{\alpha},\bA^{(s)}_{\alpha}, s \in \S_{g}\} \cup \{\etam_{\beta},\bA^{(s)}_{\beta}, s \in \S_{g,t}\}\cup\{\bC_{\alpha},\bC_{\beta},\bSigma_{\alpha},\bSigma_{\beta},\sigma_{\epsilon}^{2}\}$ denoting the complete set of parameters. 
Also, we denote the true parameters by $\bTheta_{0}$. 

Using the latent factor representation of the proposed model from Section~\ref{sec: factor model}, 
we have $\bY^{(i)} \sim \MVN(\bmu^{(i)},\bSigma^{(i)})$,  
where $\bmu^{(i)}=\bmu_{h_{g}}$ and $\bSigma^{(i)} = \bSigma_{h_{g}}$ if the $i\th$ subject is in the $h_{g}\th$ group, 
with $\bmu_{h_{g}}=\bLambda_{\alpha,h_{g}}\vec(\bC_{\alpha})+\bLambda_{\beta,h_{g}}\vec(\bC_{\beta})$ and $\bSigma_{h_{g}}=\bLambda_{\alpha,h_{g}}\bSigma_{\alpha,\eta}\bLambda_{\alpha,h_{g}}\trans+\bLambda_{\beta,h_{g}}\bSigma_{\beta,\eta}\bLambda_{\beta,h_{g}}\trans+\sigma^{2}_{\epsilon}\bI_{nd_{1}d_{2}d_{3}}$.
Also, $\sum_{i}\|\bmu^{(i)}-\bmu_{0}^{(i)}\|_{2}^{2}=n\sum_{h_{g}=1}^{d_{g}}N_{g}\|\balpha_{h_{g}}-\balpha_{0,h_{g}}\|^{2}_{F} + \sum_{h_{g}=1}^{d_{g}}N_{h_{g}}\sum_{k=1}^{n}\|\bbeta_{t_{k},h_{g}}-\bbeta_{0,h_{g},t_{k}}\|^{2}_{F}$. 

Expanding the estimates in Parts (ii) and (iii) of Lemma~\ref{lem: normal divergences} in a Taylor series [i.e., $\log(1+x)=x-x^{2}/2+o(1)$ as $x\rightarrow 0$], the Kullback-Leibler (KL) divergence and variation between $Q_{\bTheta}$ and $Q_{\bTheta_{0}}$ is bounded by a multiple of $\sum_{i=1}^{n}\|\bR^{(i)}\|^{2}_{F}+\sum_{i=1}^{n}(\|\bSigma^{(i)}\|^{-1/2}_{op}+\|\bSigma^{(i)}\|^{-1}_{op}\|\bSigma_{0}^{(i)}\|^{1/2}_{op})\|\bmu^{(i)}-\bmu_{0}^{(i)}\|_{2}^{2}$ with $\bR^{(i)}=(\bSigma_{0}^{(i)})^{-1/2}(\bSigma^{(i)}- \bSigma_{0}^{(i)})(\bSigma_{0}^{(i)})^{-1/2}$
with $\bmu^{(i)}=\bmu_{h_{g}}$ and $\bSigma^{(i)} = \bSigma_{h_{g}}$, for $i\th$ subject in $h_{g}\th$ group. 
Using the second assertions of Lemmas  \ref{lem: S.1.7} and \ref{lem: S.1.8} with \ref{lem: S.1.9}, we have KL divergence between $Q_{\bTheta}$ and $Q_{\bTheta_{0}}$ bounded by the Frobeneous distances between the core-tensors and supremum distances between the mode-matrices. 
By Lemma E.7 of \cite{GhosalBook}, the approximation error by B-spline bases for $\iota$-H\"older smooth function is given by $M^{-\iota}$ when approximated using $M$-bases. 
For the spatial and temporal mode matrices, we have considered Gaussian bases, for which \cite{hangelbroek2010nonlinear} provides a similar approximation error bound.
%
%
Let $\bA^{(s)}_{0,M}$ be a generic mode matrix with columns specified using $M$ bases, and $\bA^{(s)}_{0}$ be a mode matrix with columns being $\iota$-H\"older smooth. 
Then, using those above-mentioned results, $\|\bA^{(s)}_{0,M}-\bA^{(s)}_{0}\|_{\infty} \lesssim M^{-\iota}$. 
Also, for both B-spline and Gaussian bases, we have $\|\bpsi_{M}\trans\bgamma_{1}-\bpsi_{M}\trans\bgamma_{2}\|_{\infty}\lesssim \|\bgamma_{1}-\bgamma_{2}\|_{\infty}$, where $\bpsi_{M}$ stands for the set of spline or Gaussian bases of size $M$. Then invoking Lemma~\ref{lem: sup-norm CSC-PING} for mode matrices and standard concentration bounds of Gaussian for core tensor elements, the required prior positivity condition involving the KL divergence and variation is satisfied under the regularity conditions in Assumptions 1-4.

Let $m_{1,N},m_{2,N},m_{3,N}$ be the number of bases along the three image directions respectively, 
and $d_{t,N}$ be the number of B-spline bases in the mode-matrix for time direction. 
Combining the above approximation results along with Lemmas \ref{lem: sup-norm PING}, \ref{lem: sup-norm CSC-PING}, and  \ref{lem: S.1.9} ensures that 
$\bTheta_{0}$ belongs to the KL support of the proposed prior.

Finally, we construct the test based on the Reyni divergence.  
Let $\phi_{N}=\mathbf{1}\{Q_{\bTheta_{1}}/Q_{\bTheta_{0}}>1\}$ stand for the likelihood ratio test. 
Applying  Markov inequality to the square root of the likelihood ratio, both Type I and Type II error probabilities are bounded by $\exp\{-Nd_{R}(Q_{\bTheta_{1}},Q_{\bTheta_{0}})\}$, where $d_{R}=-\frac{1}{N}\log \int \sqrt{Q_{\bTheta_{1}} Q_{\bTheta_{0}}}$ is the average Reyni divergence. 
Then, using the Cauchy-Schwarz inequality, 
the probability of type II error of $\phi_T$ at some $\bTheta_{2}$ is bounded by 
\begin{align}
    \eE_{\bTheta_{2}}(1-\phi_N) &\le [\eE_{\bTheta_{1}}(1-\phi_N)]^{1/2} [\eE_{\bTheta_{1}} (Q_{\bTheta_{2}}/Q_{\bTheta_{1}})^{2} ]^{1/2} \nonumber.
    \label{type II error}
\end{align}
Now, by Part (4) of Lemma \ref{lem: normal divergences} we have $\eE_{\bTheta_{1}} (Q_{\bTheta_{2}}/Q_{\bTheta_{1}})^{2}$ bounded by a constant within each of $\C_{N}$ pieces of $\W_{N}$ due to Lemma ~\ref{lem: covering}. 
Hence, the above test is exponentially consistent for the small pieces as long as $\eE_{\bTheta_{1}} (Q_{\bTheta_{2}}/Q_{\bTheta_{1}})^{2}$ is bounded for all the $\bTheta_{2}$ in that piece around $\bTheta_{1}$. 
Let $\phi_{N,j}$ be the test for each such piece. 
The final test function $\chi_{N}$ satisfying exponentially decaying Type I and Type II probabilities is then obtained by taking the maximum over the tests $\phi_{N,j}$'s. 
Thus $\chi_{N}=\max_j\phi_{N,j}$. Then Type II error is bounded by the maximum of the Type II errors of the individual pieces and thus remains exponentially small. 
However, Type I error is given by $\eE_{\bTheta_{0}}(\chi_{N}) \leq \sum_{j} \eE_{\bTheta_{0}}\phi_{N,j} \leq \C_{N} \eE_{\bTheta_{0}} \phi_{N,j}$. 
Hence, we need to show that $\log \C_N\lesssim N$, where $\C_{N}$ is the required number of balls needed to cover our sieve $\W_{N}$ with a radius such that $\eE_{\bTheta_{1}} (Q_{\bTheta_{2}}/Q_{\bTheta_{1}})^{2}$ is bounded for all $\bTheta_{2}$ within each piece setting $\bTheta_{1}$ as the center.
Lemma \ref{lem: covering} defines the required sieve $\W_{N}$ satisfying the above requirements.
It also shows that the prior probability of its complement is upper-bounded by $\exp(-CN)$ for some constant $C>0$.

Following Part II of the proof of Theorem 3.1 in \cite{ning2020bayesian}, we can then conclude that having $d_{R}(Q_{\bTheta_{1}},Q_{\bTheta_{0}})\lesssim \epsilon^{2}$ implies $\frac{1}{N}\sum_{i=1}^{N}\|\bmu_{1}^{(i)}-\bmu_{0}^{(i)}\|^{2}_{2}\lesssim \epsilon^{2}$, which completes the proof of the main Theorem \ref{thm: consistency}.

\subsection{Proof Details of Theorem \ref{thm: consistency}} 

\begin{Lem}[Auxiliary Result]
\label{lem: norm}
If $\btheta=\etam_{1}\times \bA^{(1)}\times_{2} \bA^{(2)}\times_{3} \bA^{(3)}$ with $\bA^{(s)}$ is $d_{s}\times r_{s}$, then 
(1) $\|\btheta\|^{2}_{F} \le \|\etam\|^{2}_{F}\prod_{s}(d_{s}r_{s}\|\bA^{(s)}\|^{2}_{\infty})$, 
(2)  $\|\bM_{1}\otimes \bM_{2}\|_{op}=\|\bM_{1}\|_{op}\|\bM_{2}\|_{op}$, 
(3) $\|\bM_{1}\otimes \bM_{2}\|^{2}_{F}=\|\bM_{1}\|^{2}_{F}\|\bM_{2}\|^{2}_{F}$.
\end{Lem}

\begin{proof}
We have $\vec(\btheta) = \bLambda\vec(\etam)$, where $\bLambda=\bA^{(1)}\otimes \bA^{(2)}\otimes \bA^{(3)}$.
Then, 
\vspace*{-5ex}\\
\bse
\textstyle \|\btheta\|^{2}_{F}=\vec(\etam)\trans\bLambda\trans\bLambda \vec(\etam)\le \|\etam\|^{2}_{F} \|\bLambda\|^{2}_{op}=\|\etam\|^{2}_{F} \prod_{s}\|\bA^{(s)}\|^{2}_{op}\leq \|\etam\|^{2}_{F}\prod_{s}(d_{s}r_{s}\|\bA^{(s)}\|^{2}_{\infty}).
\ese
\vspace*{-4ex}\\
Parts (2) and (3) follow directly from the definitions of operator and Frobenius norms.
\end{proof}

\begin{Lem}[Auxiliary Result]
Let $\btheta_{1}$ and $\btheta_{2}$ admit Tucker decompositions with the same ranks as $\btheta_{1}=\etam_{1}\times_{1} \bA^{(1)}_{1}\times_{2} \bA^{(2)}_{1}\times_{3} \bA^{(3)}_{1}$ and $\btheta_{2}=\etam_{2}\times_{1} \bA^{(1)}_{2}\times_{2} \bA^{(2)}_{2}\times_{3} \bA^{(3)}_{2}$, then $\{\|\btheta_{1}-\btheta_{2}\|^{2}_{F}\leq \epsilon^{2}\} \supseteq \{ \|\etam_{1}-\etam_{2}\|^{2}_{F} \leq \kappa^{2}, \|\bA^{(s)}_{1}-\bA^{(s)}_{2}\|^{2}_{\infty} \leq \kappa^{2}: s=1,2,3\}$ if $\|\etam_{2}\|_{F}^{2}\leq B_{1}$ and $\max_{s}\|\bA_{2}^{(s)}\|_{\infty}^{2}\leq B_{2}$ with $\epsilon^{2}\gtrsim d_{1}d_{2}d_{3}r_{1}r_{2}r_{3}\kappa^{2}B^{3}_{2}B_{1}$.
\label{lem:difftensor}
\end{Lem}

\begin{proof}
$\etam_{1}\times_{1} \bA^{(1)}_{1}\times_{2} \bA^{(2)}_{1}\times_{3} \bA^{(3)}_{1}=(\etam_{1}-\etam_{2}+\etam_{2})\times_{1} (\bA^{(1)}_{1}-\bA^{(1)}_{2}+\bA^{(1)}_{2})\times_{2} (\bA^{(2)}_{1}-\bA^{(2)}_{2}+\bA^{(2)}_{2})\times_{3} (\bA^{(3)}_{1}-\bA^{(3)}_{2}+\bA^{(3)}_{2})=\etam_{2}\times_{1} \bA^{(1)}_{2}\times_{2} \bA^{(2)}_{2}\times_{3} \bA^{(3)}_{2} + \sum_{I}(\etam_{1}-\etam_{2})\times_{1}\b1e_{1}\times_{2}\b1e_{2}\times_{3}\b1e_{3}+\sum_{I}\etam_{2}\times_{1}\b1e_{1}\times_{2}\b1e_{2}\times_{3}\b1e_{3},$ where $\b1e_{j}$ is $\bA^{(j)}_{1}-\bA^{(j)}_{2}$ if $j\in I$ otherwise it is $\bA^{(j)}_{2}$. Here $I=\{\{1\},\{2\},\{3\}, \{1,2\},\{1,3\},\{2,3\}\}$. 
Thus, $(\btheta_{1}-\btheta_{2}) = \sum_{I}(\etam_{1}-\etam_{2})\times_{1}\b1e_{1}\times_{2}\b1e_{2}\times_{3}\b1e_{3}+\sum_{I}\etam_{2}\times_{1}\b1e_{1}\times_{2}\b1e_{2}\times_{3}\b1e_{3}$.
The final assertion can now be obtained by applying the triangle inequality to this expansion to bound $\|\btheta_{1}-\btheta_{2}\|^{2}_{F}$. 
\end{proof}

The following lemma presents some standard divergence results for multivariate normal distributions, which can be verified easily. 

\begin{Lem}[Auxiliary Result]
  Let $f_{1}$ and $f_{2}$ be probability densities of $k$-dimensional normal distributions with means  
$\bzeta_{1}$ and $\bzeta_{2}$ with dispersion matrices $\bDelta_{1}$ and  $\bDelta_{2}$ respectively. 
Let $\bR=\bDelta_{1}^{-1/2}(\bDelta_{2}- \bDelta_{1})\bDelta_{1}^{-1/2}$ with eigenvalues $-1<\rho_{1},\ldots,\rho_{k}<\infty$. Then the following assertions hold. 
\begin{enumerate}
\item Reyni Divergence: $d_{R}(f_{1},f_{2}) =-\log \int \sqrt{f_{1} f_{2}}$ is
\bse
& \hspace*{-1.5cm} \frac12 \log \det \{(\bDelta_{2}+\bDelta_{2})/2\} -\frac14 \log \det (\bDelta_{2})-\frac14\log\det (\bDelta_{2}) + \frac{1}{8}(\bzeta_{1}-\bzeta_{2})\trans\left(\frac{\bDelta_{2}+\bDelta_{2}}{2}\right)(\bzeta_{1}-\bzeta_{2})\\
       &=\frac12 \log \det (\bI+\bR/2)-\frac14 \log\det (\bI+\bR)+ \frac{1}{8}(\bzeta_{1}-\bzeta_{2})\trans\left(\frac{\bDelta_{2}+\bDelta_{2}}{2}\right)(\bzeta_{1}-\bzeta_{2})\\
       &=\frac14 \sum_{j=1}^k [2\log (1+\rho_{j}/2)-\log (1+\rho_{j})]+ \frac{1}{8}(\bzeta_{1}-\bzeta_{2})\trans\left(\frac{\bDelta_{2}+\bDelta_{2}}{2}\right)(\bzeta_{1}-\bzeta_{2}).
\ese
\item Kullback-Leibler Divergence: $K(f_{1},f_{2}) =\int f_{1} \log (f_{2}/f_{1})$ is
\bse
&&\hspace*{-2cm}\frac12 \log \det (\bDelta_{2})-\frac12 \log\det (\bDelta_{2}) +\frac12 \trace\{\bDelta_{2}^{-1/2}(\bDelta_{2}- \bDelta_{2}) \bDelta_{2}^{-1/2}\} + \frac12\|\bDelta_{2}^{-1/2}(\bzeta_{1}-\bzeta_{2})\|^{2}_{F}\\
       &=&\frac12 \log \det (\bI+\bR)+\frac12 \trace\{(\bI+\bR)^{-1}-\bI\}+\frac12\|\bDelta_{2}^{-1/2}(\bzeta_{1}-\bzeta_{2})\|^{2}_{F}\\
       &=&\frac12 \sum_{j=1}^k [\log (1+\rho_{j})-\rho_{j}/(1+\rho_{j})]+\frac12\|\bDelta_{2}^{-1/2}(\bzeta_{1}-\bzeta_{2})\|^{2}_{F}.
\ese
\item The Kullback-Leibler variation: $V(f_{1},f_{2}) =\int f_{1} \{\log (f_{2}/f_{1})-K(f_{1},f_{2})\}^{2}$ is 
\bse
       &\frac12 \trace\{(\bDelta_{2}^{-1/2}(\bDelta_{2}- \bDelta_{1}) \bDelta_{2}^{-1/2})^{2}\} + \|\bDelta_{2}^{1/2}\bDelta_{2}^{-1}(\bzeta_{1}-\bzeta_{2})\|^{2}_{F}
       \\&\quad=\frac12 \trace[\{(\bI+\bR)^{-1}-\bI\}^{2}]
       =\frac12 \sum_{j=1}^k \rho_{j}^{2}/(1+\rho_{j})^{2}+\|\bDelta_{2}^{1/2}\bDelta_{2}^{-1}(\bzeta_{1}-\bzeta_{2})\|^{2}_{F}.
\ese
\item Expected Squared Likelihood: If $(2\bDelta_{2}-\bDelta_{1})$ is positive definite, $\int (f_{2}/f_{1})^{2} f_{1}$ is 
\begin{align*} 
G\frac{\det (\bDelta_{2})}{\sqrt{\det(\bDelta_{2}) \det (2\bDelta_{2}-\bDelta_{1})}}=G\frac{1}{\sqrt{\det(\bI+\bR)\det(\bI-\bR)}}=G\exp\left[ \sum_{j=1}^k \log (\rho_{j}^{2}-1)/2 \right],
\end{align*}
where $G=\exp\{(\bzeta_{1}-\bzeta_{2})\trans(2\bDelta_{2}-\bDelta_{2})^{-1}(\bzeta_{1}-\bzeta_{2})\}$.
\end{enumerate}
\label{lem: normal divergences}
\end{Lem}

\begin{Lem}
\label{lem:eigenvalues}
The eigenvalues of $\bSigma_{h_{g},0}$ lie uniformly in some compact interval $[b_1,b_2]\subset (0,\infty)$. 
\end{Lem}

\begin{proof}
Under a factor model setting with $\bSigma_{h_{g},0}=\bLambda_{\alpha,h_{g},0}\bSigma_{\alpha,\eta,0}\bLambda_{\alpha,h_{g},0}\trans+\bLambda_{\beta,h_{g},0}\bSigma_{\beta,\eta,0}\bLambda_{\beta,h_{g},0}\trans+\sigma^{2}_{0,\epsilon}\bI$, then $\|\bSigma_{h_{g},0}\|_{op}\le \|\bSigma_{\alpha,\eta,0}\|_{op}\|\bLambda_{\alpha,h_{g},0}\|^{2}_{op}+\|\bSigma_{\beta,\eta,0}\|_{op}\|\bLambda_{\beta,h_{g},0}\|^{2}_{op} + \sigma^{2}_{\epsilon,0} $. 
On the other hand, $\|\bSigma^{-1}_{h_{g},0}\|_{op}\leq \frac{1}{\sigma^{2}_{\epsilon,0}}$.
We have $\|\bLambda_{\alpha,h_{g},0}\|_{op}\leq \|\ba^{(g)}_{\alpha,0}(h_{g})\|_{op}\prod_{k=1}^{3}\|\bA^{(s)}_{\alpha,0}\|_{op}$ and $\|\bLambda_{\beta,h_{g},0}\|_{op}\leq \|\ba^{(g)}_{\beta,0}(h_{g})\|_{op}\|\bA^{(t)}_{\beta,0}\|_{op}\prod_{k=1}^{3}\|\bA^{(s)}_{\beta,0}\|_{op}$, where $\bA^{(t)}_{\beta,0}$ is the mode-matrix evaluated on the common observation-times.
By definition, $\|\bA^{(s)}_{\alpha,0}\|^2_{op}=\|(\bA^{(s)}_{\alpha,0})^T\bA^{(s)}_{\alpha,0}\|^2_{\infty}$ as $(\bA^{(s)}_{\alpha,0})^T\bA^{(s)}_{\alpha,0}$ is diagonal due to semi-orthogonality and due to Assumption (1), it is bounded for $s\in \S$. A similar argument holds for $\|\bA^{(s)}_{\beta,0}\|^2_{op}$ as well for $s\in \S$.
\end{proof}

\begin{Lem}
(i) Let $\bL_{1},\bL_{2}$ be $p \times r$ matrices and $\bD_{1},\bD_{2}$ be $r \times r$ diagonal matrices with entries in $(0,B)$ in absolute values for some $B>0$. 
Then, for $\bOmega_1=\bL_{1}\bD_{1} \bL_{1}\trans$, $\bOmega_2=\bL_{2}\bD_{2}\bL_{2}\trans$, we have that $\|\bOmega_1-\bOmega_2\|_{F}\le 2B^{2}\sqrt{pr} \|\bL_{1}-\bL_{2}\|_{F}+ pr B^{2}\|\bD_{1}-\bD_{2}\|_{F}$. 
(ii) Let the entries in $\bL_{1}$ and $\bD_{1}$ be bounded in absolute value by $B$, then for any $\epsilon>0$, we have $\{\|\bOmega_1-\bOmega_2\|_{F}\leq \epsilon\}\supseteq \{\|\bL_{1}-\bL_{2}\|_{F}\le \delta_{1}, \|\bD_{1}-\bD_{2}\|_{F}\le \delta_{2}\}$ for some small $\delta_{1},\delta_{2}$, depending on $\epsilon$.  
\label{lem: S.1.7}
\end{Lem}

\begin{proof}
(i) By matrix norm inequalities $\|\bA \bB\|_{F}\le \min\{\|\bA\|_{F}\|\bB\|_{op}, \|\bA\|_{op}\|\bB\|_{F}\}$, we have 
\begin{align*}
\|\bOmega_1-\bOmega_2\|_{F}
\le \|\bL_{1}-\bL_{2}\|_{F}(\|\bD_{1}\|_{op}\|\bL_{1}\|_{op}+\|\bD_{2}\|_{op}\|\bL_{2}\|_{op})+\|\bL_{1}\|_{op}\|\bL_{2}\|_{op}\|\bD_{1}-\bD_{2}\|_{F}.
\end{align*}    
Since $\bD_{1},\bD_{2}$ are diagonal matrices with entries in $(0,B)$ and $\|\bL_{1}\|_{op},\|\bL_{2}\|_{op} \leq \sqrt{pr}B$, the result is immediate. 

(ii) To prove the second assertion, note that $\|\bD_{2}\|_{op}\leq \|\bD_{1}\|_{op}+\|\bD_{2}-\bD_{1}\|_{op}$ and $\|\bL_{2}\|_{op}\leq \|\bL_{1}\|_{op}+\|\bL_{2}-\bL_{1}\|_{op}$. 
Hence, the assertion holds for small $\delta_{1}\lesssim \epsilon/(\sqrt{pr}B^2),\delta_{2}\le \epsilon/(prB^2)$.
\end{proof}

\begin{Lem}
(i) Let $\bL_{1},\bL_{2}$ be $(p\times r)$ matrices, 
$\bD_{1},\bD_{2}$ be $r \times r$ diagonal matrices, 
and $\bc_1,\bc_2$ be $r \times 1$ vectors with entries in $(0,B)$ in absolute values for some $B>0$. 
Then, we have that $\|\bL_{1}\bc_1-\bL_{2}\bc_2\|_{F}\le B\sqrt{pr} \|\bc_1-\bc_2\|_{2}+ B\|\bL_{1}-\bL_{2}\|_{F}+\|\bL_{1}-\bL_{2}\|_{op}\|\bc_1-\bc_2\|_{2}$.
(ii) Let the entries in $\bL_{1}$ be $\bD_{1}$ be bounded in absolute value by $B$, then for any $\epsilon > 0$, we have $\{\|\bL_{1}\bc_1-\bL_{2}\bc_2\|_{F}\leq \epsilon\}\supseteq \{\|\bL_{1}-\bL_{2}\|_{F}\le \delta_{1}, \|\bc_{1}-\bc_{2}\|_{F}\le \delta_{2}\}$ 
for some small $\delta_{1},\delta_{2}$, depending on $\epsilon$.  
\label{lem: S.1.8}
\end{Lem}

\begin{proof}
(i) We have $\bL_{1}\bc_1-\bL_{2}\bc_2=\bL_{1}\bc_1-(\bL_{2}-\bL_{1}+\bL_{1})(\bc_2-\bc_1+\bc_1)=-(\bL_{2}-\bL_{1})(\bc_2-\bc_1)-(\bL_{2}-\bL_{1})\bc_1-\bL_{1}(\bc_2-\bc_1)$. 
Then, applying the matrix norm inequalities, we get the final result.

(ii) Using the above decomposition, the second assertion holds immediately for small $\delta_{1},\delta_{2}\lesssim \epsilon$ for small $\delta_{1}\lesssim \epsilon/B,\delta_{2}\lesssim \epsilon/(\sqrt{pr}B)$.
\end{proof}

\begin{Lem}
$\{\|\bLambda_{\alpha,h_{g},1}-\bLambda_{\alpha,h_{g},2}\|^{2}_{2}\leq \epsilon^{2}\} \supseteq \{ \|\bA^{(s)}_{\alpha,1}-\bA^{(s)}_{\alpha,2}\|^{2}_{\infty} \leq \kappa^{2}: s \in \S_{g}\}$  with $\epsilon^{2}\gtrsim d_{1}d_{2}d_{3}\kappa^{2}B^{3}_{2} r_{\alpha,1} r_{\alpha,2}r_{\alpha,3}$ and $\{\|\bLambda_{\beta,h_{g},1}-\bLambda_{\beta,h_{g},2}\|^{2}_{2}\leq \epsilon_{1}^{2}\} \supseteq \{ \|\bA^{(s)}_{\beta,1}-\bA^{(s)}_{\beta,2}\|^{2}_{\infty} \leq \kappa^{2}: s \in \S_{g,t}\}$  with $\epsilon_{1}^{2}\gtrsim d_{1}d_{2}d_{3}n\kappa^{2}B^{3}_{2}r_{\beta,1}r_{\beta,2}r_{\beta,3}r_{\beta,t}$.
\label{lem: S.1.9}
\end{Lem}

\begin{proof} We have
$\{\|\btheta_{1}-\btheta_{2}\|^{2}_{F}\leq \epsilon^{2}\} \supseteq \{ \|\bA^{(s)}_{\alpha,1}-\bA^{(s)}_{\alpha,2}\|^{2}_{\infty} \leq \kappa^{2}: s \in \S_{g} \}$ if $\max_{s}\|\bA^{(s)}_{\alpha,2}\|_{\infty}^{2}\leq B_{2}$ with $\epsilon^{2}\gtrsim d_{1}d_{2}d_{3}\kappa^{2}B^{3}_{2}R_{1}R_{2}R_{3}$ on applying Lemma~\ref{lem:difftensor} and setting the core-tensor values at 1.
\end{proof}

Lemmas \ref{lem: S.1.8} and \ref{lem: S.1.9} allow quantifying the separation between $\bmu_{h_{g},1}$ and $\bmu_{h_{g},2}$ in terms of the differences between mode-matrices and core-tensors. Similarly, Lemmas \ref{lem: S.1.7} and \ref{lem: S.1.9} allow quantifying the separation between $\bSigma_{h_{g},1}$ and $\bSigma_{h_{g},2}$ again in terms of the differences between mode-matrices and core-tensors. When one of them is the truth, we will consider the second assertion of the Lemmas \ref{lem: S.1.7} and \ref{lem: S.1.8}.

\begin{Lem}[Sieve]
\label{lem: covering}
Let $\mathcal{W}_{N}$ be the collection of $\bA^{(g)}_{\alpha}$, 
$\{\bgamma^{(s)}_{\alpha,\ell,k}: 1\le \ell\le r_{\alpha,s}, 1\le k \le q_{\alpha, s},  s \in \S_{g}\}$, $\bA^{(g)}_{\beta}$, 
$\{\bgamma^{(s)}_{\beta,\ell,k}: 1\le \ell\le r_{\beta,s},1 \le k \le q_{\beta,s},  s \in \S_{g,t}\}$,  
$\{\etam_{\alpha}\}$, $\{\etam_{\beta}\}$, 
and the variance parameters $\sigma,\bsigma_{\alpha},\bsigma_{\beta}$ such that\\ 
(1) $\|\bA^{(g)}_{\alpha}\|_{\infty},\|\bgamma^{(s)}_{\alpha,\ell,k}\|_{\infty}\le B_{N,\alpha}$ for all $k$ and $s \in \S_{g}$, ~~~\\
(2) $\|\bA^{(s)}_{\beta}\|_{\infty}, \|\bgamma^{(s)}_{\beta,\ell,k}\|_{\infty} \le B_{N,\beta}$ for all $k$ and $s \in \S_{g,t}$, ~~~\\
(3) $\|\etam_{\alpha}\|_{\infty}\leq E_{N,\alpha}$, ~~~\\
(4) $\|\etam_{\beta}\|_{\infty}\leq E_{N,\beta}$, ~~~\\
(5) $\sigma,\|\bsigma_{\alpha}\|_{\infty},\|\bsigma_{\beta}\|_{\infty} \in (1/n, e^{c_{2}n})$ 
with $B_{N,\alpha}, B_{N,\beta},E_{N,\alpha},E_{N,\beta}, s_{N}\asymp N^{c}$, $m_{1,N},m_{2,N},m_{3,N}, K_{N} \asymp N^{c_{1}}$ for some constant $c,c_{1},c_{2}>0$ \citep[following][]{shen2015adaptive}, and ~~~\\
(6) $m_{1,N},m_{2,N},m_{3,N} \lesssim m_{N} \asymp N/\log(N)$.\\ 
Then, (i) $P(\bTheta \notin \W_{N})\leq e^{-CN}$ and (ii) $\mathcal{W}_N$ can be split into $\C_{N}$ pieces $\mathcal{W}_{N,l}$, $l=1,\ldots,\C_{N}$, $\log \C_{N}\lesssim N$, such that there exist exponentially consistent test functions $\phi_{T,l}$, $k=1,\ldots,N$ whenever $\frac{1}{N}\sum_{i=1}^{N}\|\bmu^{(i)}-\bmu_{0}^{(i)}\|^{2}_{2} + \|\bSigma-\bSigma_{0}\|^{2}_{F} > \epsilon^{2}$.
\end{Lem}

\begin{proof}
Since priors for the core-tensor elements and marginal distribution of the mode-matrix elements are Gaussian, we apply the Markov inequality to get the upper bound for $P(\bTheta \notin \W_{N})\leq M^{6}_{N}r_{\alpha,1}r_{\alpha,2}r_{\alpha,3}r_{\beta,1}r_{\beta,2}r_{\beta,3}Kr_{\beta,t}G^{2}r_{\alpha,g}r_{\beta,g}\exp(-C'N^{c}) + 4\exp(-c_{1}N^{c_{1}}\log N)+(r_{\alpha,g}r_{\alpha,1}r_{\alpha,2}r_{\alpha,3}+r_{\beta,g}r_{\beta,1}r_{\beta,2}r_{\beta,3}r_{\beta,t}+1)\exp(-C''N^{c})$ for some constant $C',C''>0$. 
We can choose $c_{1}$ such that $\log P(\bTheta \notin \W_{N})\lesssim -N$.

We divide an interval $[-N^{c},N^{c}]$ into $N_{3}\asymp N^{B}$ equal sub-intervals for a sufficiently large $B$ to be chosen later. 
The total number of units will be $\C_{N}=N_{3}^{M^{4}_{N}P_{N}}$, where $P_{N}$ is the total number of parameters $m_{N}r_{\alpha,1}+m_{N}r_{\alpha,2}+m_{N}r_{\alpha,3}+m_{N}r_{\beta,1}+m_{N}r_{\beta,2}+m_{N}r_{\beta,3}+Kr_{\beta,t}+Gr_{\alpha,g}+Gr_{\beta,g}+r_{\alpha,g}r_{\alpha,1}r_{\alpha,2}r_{\alpha,3}+r_{\beta,g}r_{\beta,1}r_{\beta,2}r_{\beta,3}r_{\beta,t}$. 
Under the assumptions on the individual components, it is clear that $\log(\C_N)\lesssim N$.

In our test construction, within each piece of $\W_{N}$, we need $\eE_{\bTheta_{1}} (Q_{\bTheta_{2}}/Q_{\bTheta_{1}})^{2}$ bounded. Now, by Part 3 of Lemma \ref{lem: normal divergences} we have $\eE_{\bTheta_{1}} (Q_{\bTheta_{2}}/Q_{\bTheta_{1}})^{2}$ bounded by a constant if 1) $\|\bSigma_{1}-\bSigma_{2}\|_{op}\leq \frac{1}{N^{2}nd_1d_2d_3}$ and 2) $\sum_{i=1}^{N}\|\bmu^{(i)}_{1}-\bmu^{(i)}_{2}\|^{2}_{2}$ is bounded by a constant. Combining Lemmas \ref{lem: S.1.7}, \ref{lem: S.1.8}, \ref{lem: S.1.9}, we can control 1) and 2)  within each piece of $\W_{N}$ if $\|\bTheta_{1}-\bTheta_{2}\|_{\infty}\le N^{-B'}$ following Part I of the proof of Theorem 3.1 in \cite{ning2020bayesian}. Thus, we need to set the above $B$ as $c+B'$. This completes the proof.
\end{proof}

\clearpage\newpage
\section{Additional Results for Real ADNI-3 Data Analysis}
\label{sec: sm add figs real}

We start with a brief discussion about inference on the individual starting values, captured by the $\alpha$'s. 
Figure~\ref{fig: realbaseFA2 and realbaseODI2} illustrates the distributions of the estimated $\alpha$'s in the corpus callosum (CC) region for the different groups, namely normal cognition (NC), mild cognitive impairment (MCI), and Alzheimer's (AD), indicating, as expected, that the AD group was well-separated from the other two at the study's onset. 
The ordering of the groups based on baseline estimates in the CC region also broadly aligns with their respective degrees of cognitive decline that define these groups, specifically for ODI (Figure~\ref{fig: realbaseFA2 and realbaseODI2}).
Results associated with NDI and FWF are in Figure~\ref{fig: realbaseFWF2 and realbaseNDI2}.

\begin{figure}[htbp]
\centering
\includegraphics[width = 0.45\textwidth, height = 6cm, trim=0 0 3cm 0, clip=true]{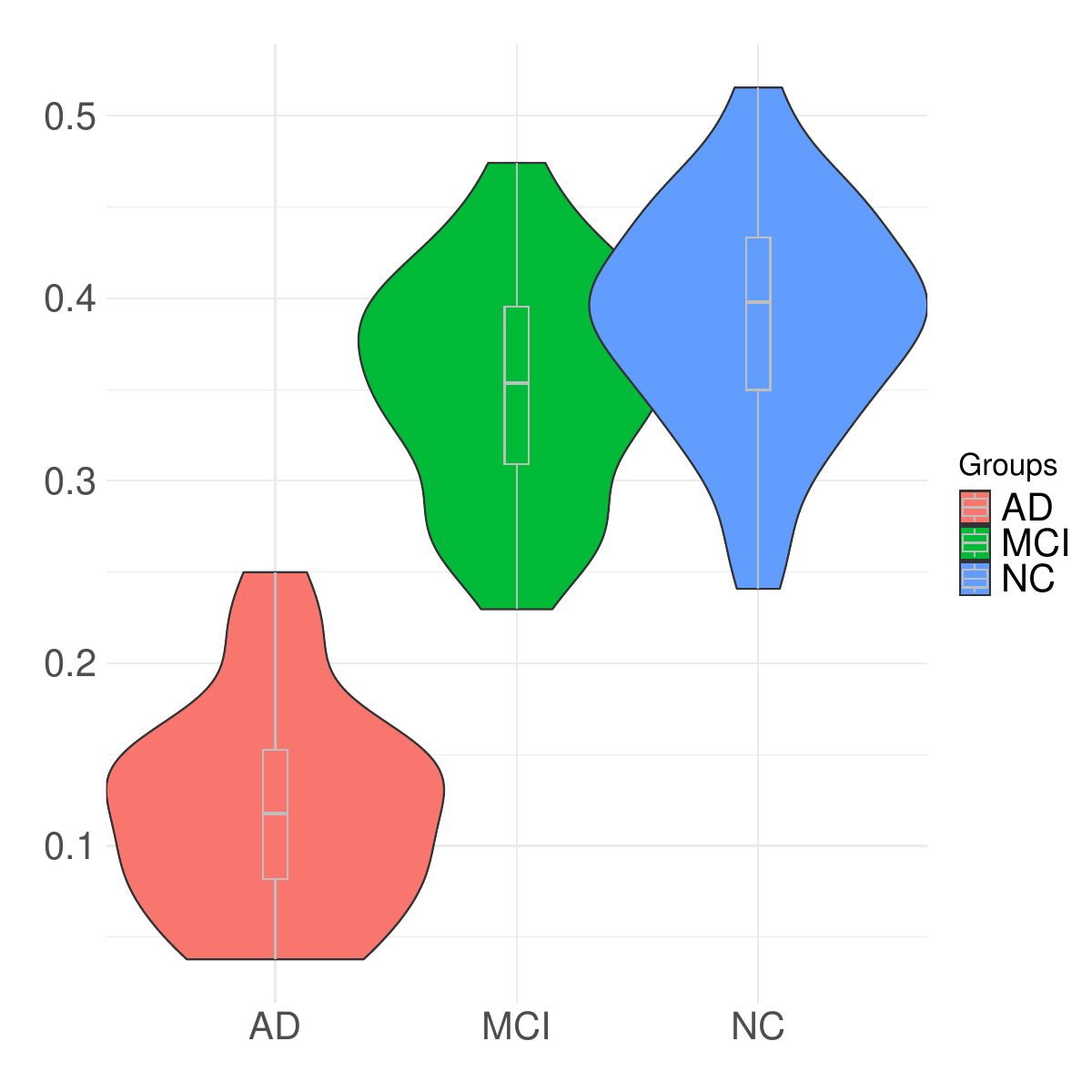}\quad\quad
\includegraphics[width = 0.45\textwidth, height = 6cm, trim=0 0 3cm 0, clip=true]{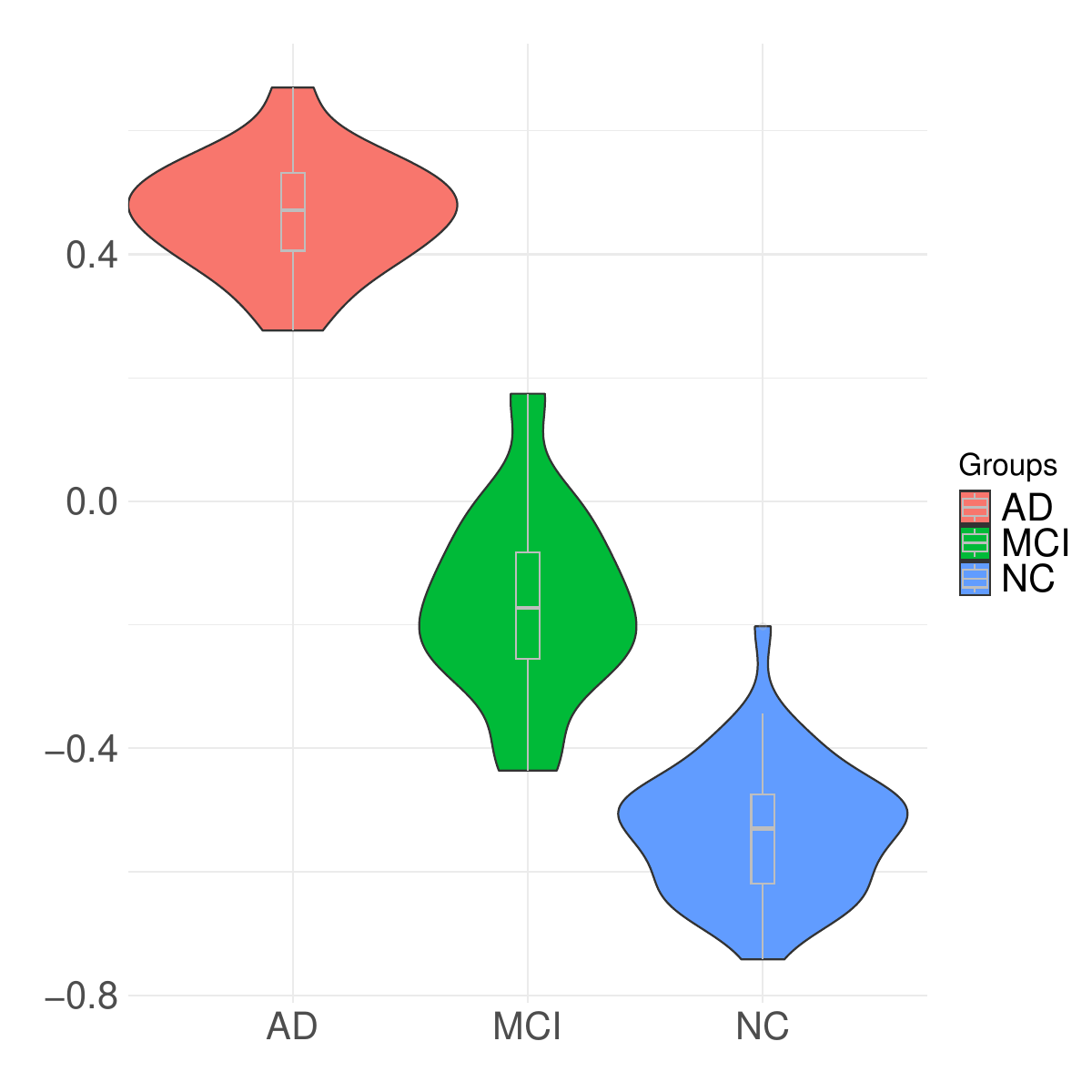}
\vspace*{-5pt}
\caption{Inference for $\alpha$'s: Violin plots of the estimated mean baseline FA (left panel) and ODI (right panel) in the corpus callosum (CC) region for normal cognition (NC, blue),  mild cognitive impairment (MCI, green), and Alzheimer's (AD, red) subjects. 
\vspace*{-10pt}
}
\label{fig: realbaseFA2 and realbaseODI2}
\end{figure}

\begin{figure}[htbp]
\centering
\includegraphics[width = 0.45\textwidth, height = 6cm, trim=0 0 3cm 0, clip=true]{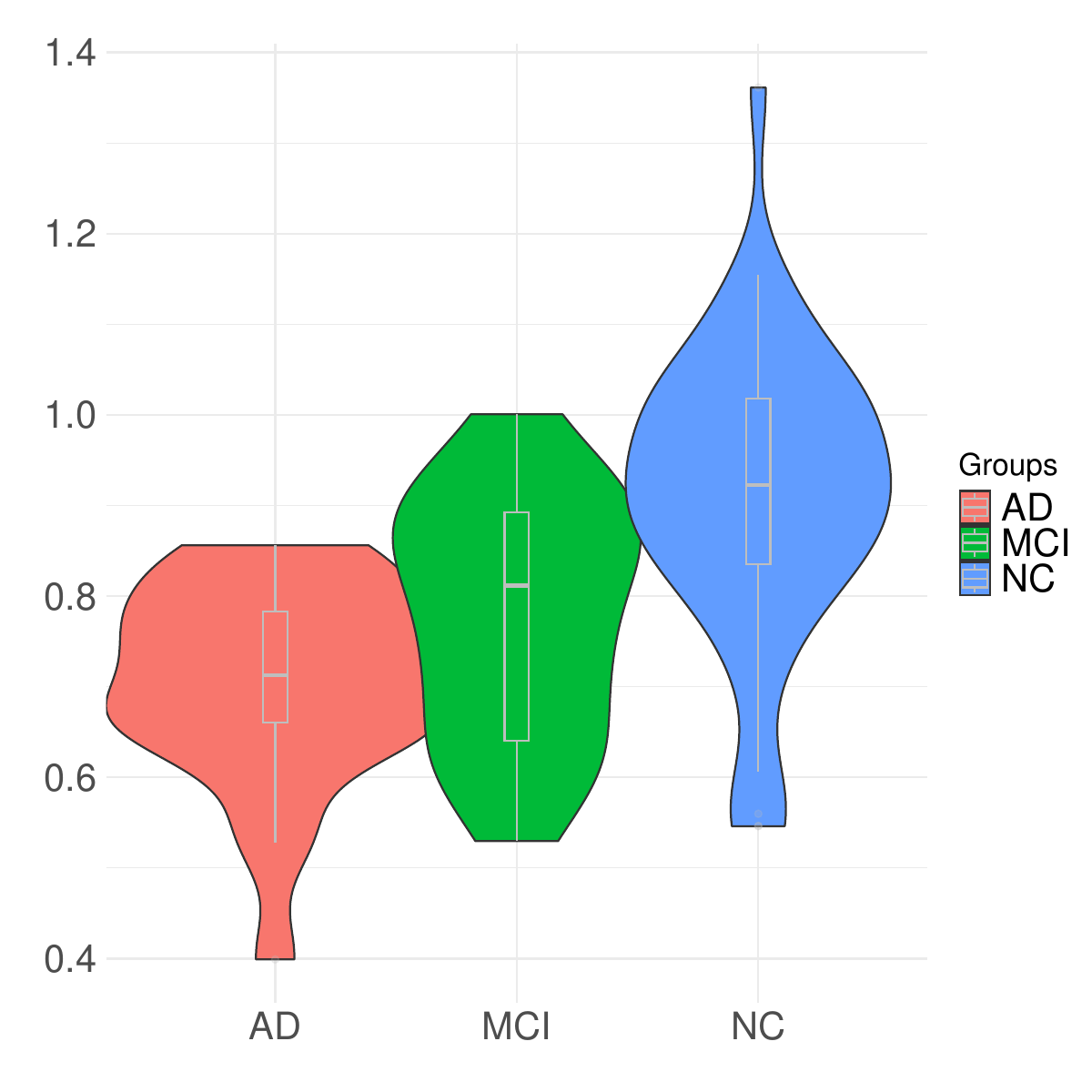}\quad\quad
\includegraphics[width = 0.45\textwidth, height = 6cm, trim=0 0 3cm 0, clip=true]{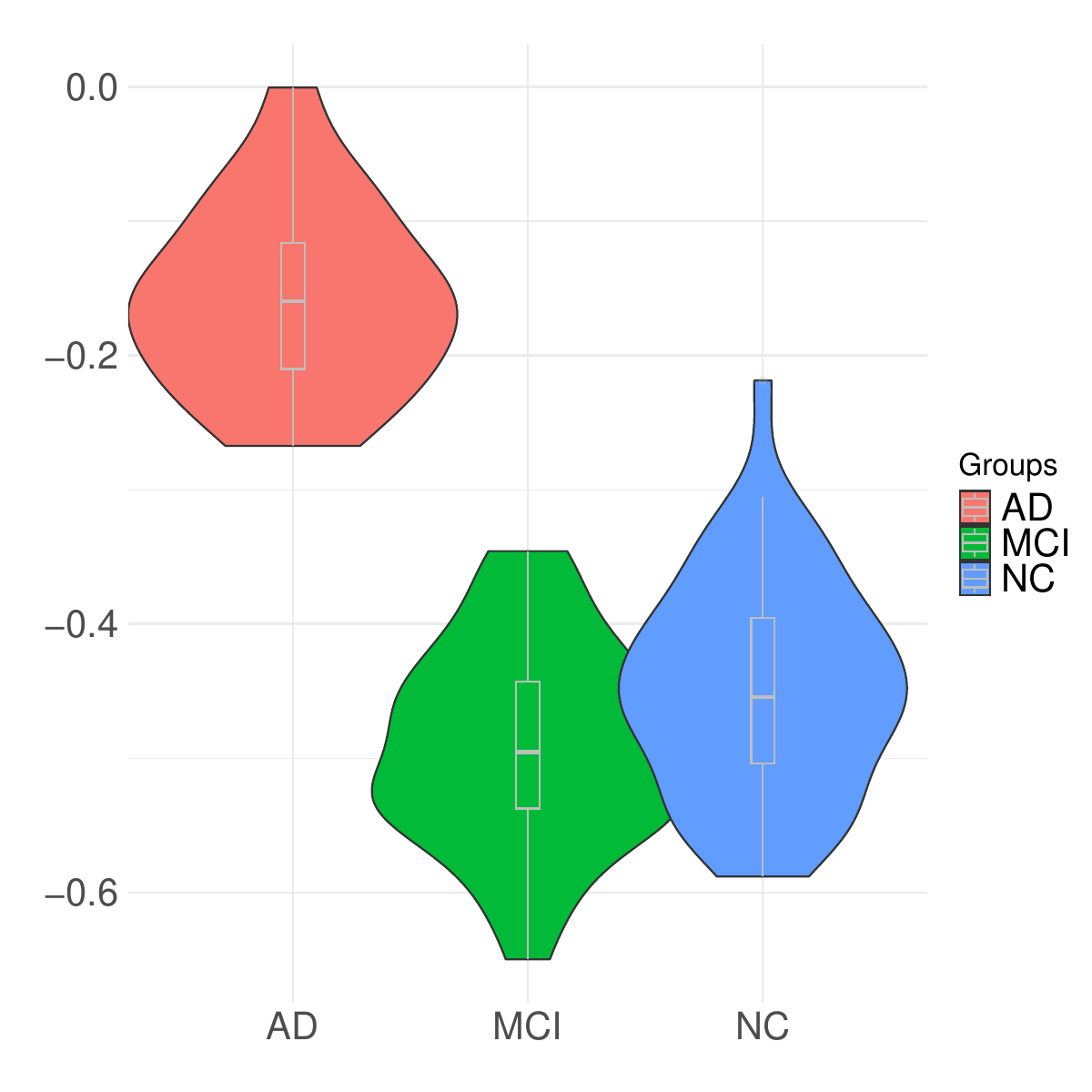}
\caption{Inference for $\alpha$'s: Violin plots of the estimated mean baseline FWF (left panel) and NDI (right panel) in the corpus callosum (CC) region for normal cognition (NC, blue),  mild cognitive impairment (MCI, green), and Alzheimer's (AD, red) subjects. 
}
\label{fig: realbaseFWF2 and realbaseNDI2}
\end{figure}

\begin{figure}[!htbp]
\centering
\includegraphics[width = 0.45\textwidth, clip=true]{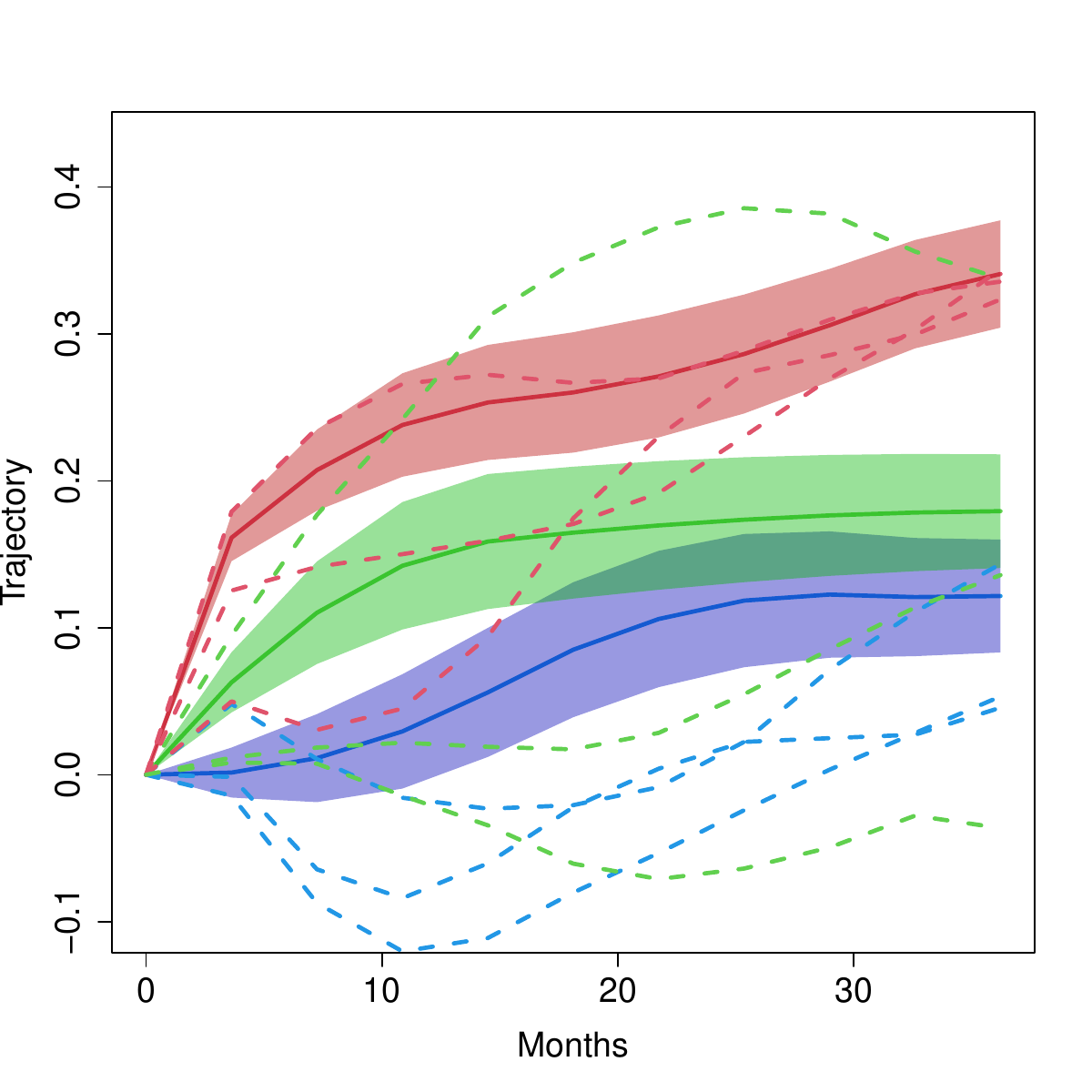}\quad\quad
\includegraphics[width = 0.45\textwidth, clip=true]{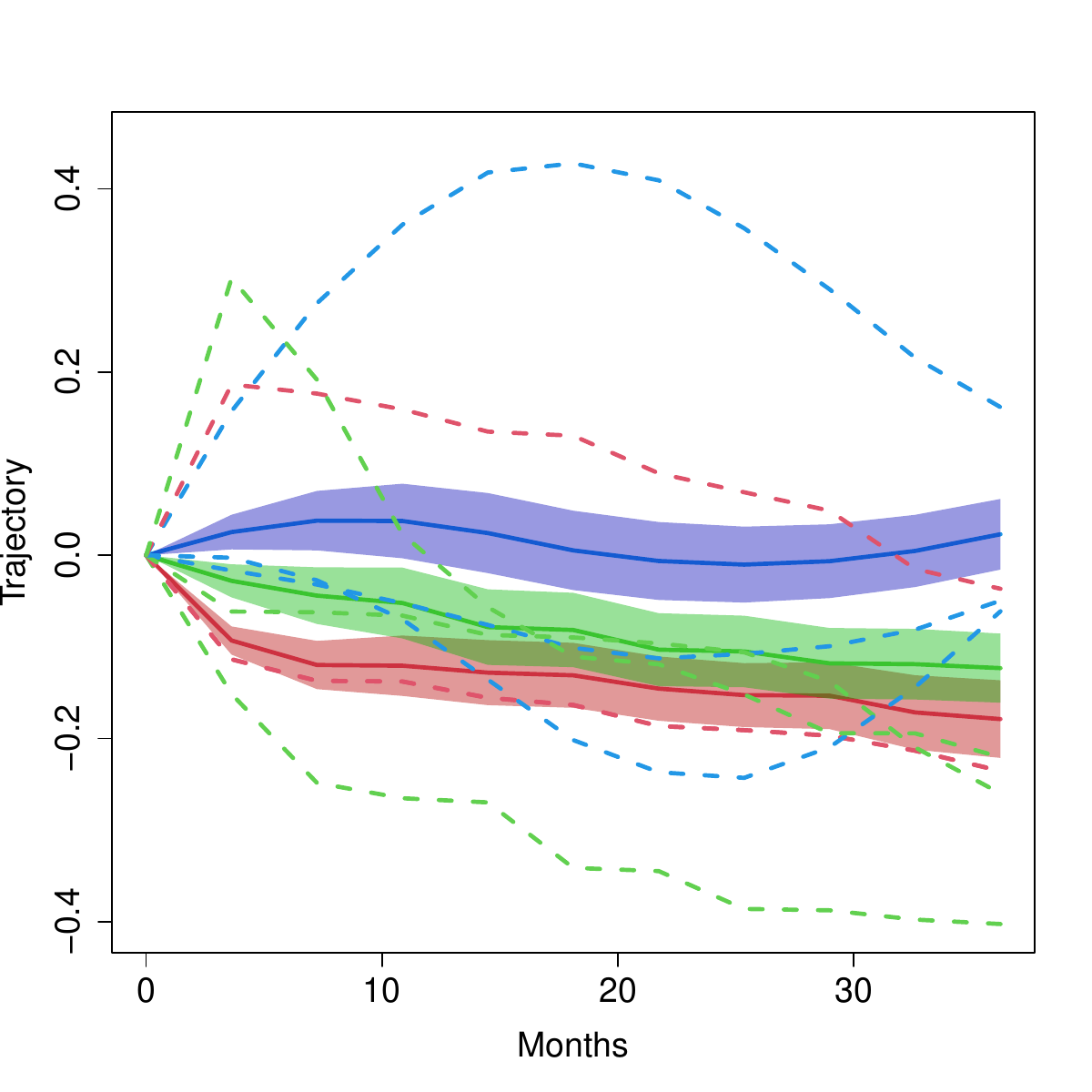}
\caption{Inference for $\beta$'s: Estimated median trajectories of FWF (left panel) and NDI (right panel) in the corpus callosum (CC) region for normal cognition (NC, blue),  mild cognitive impairment (MCI, green), and Alzheimer's (AD, red) subjects. 
The solid lines show the population-level trajectories; 
the shaded regions show the corresponding 90\% point-wise credible intervals; 
the dotted lines show three randomly selected subjects from each group. 
}
\label{fig: realFWF2 and realNDI2}
\end{figure}


\begin{figure}[htbp]
\centering
\begin{minipage}{0.47\textwidth}
\subfigure[AD vs NC]{\includegraphics[width = 0.4\textwidth, trim=1cm 1cm 1cm 2cm, clip=true]{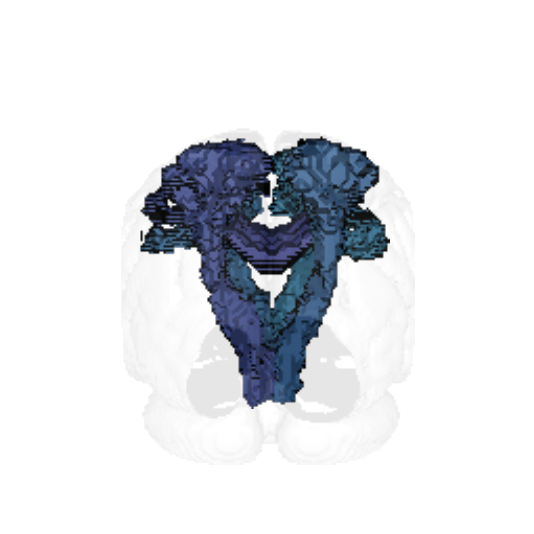}}
\subfigure[AD vs NC]{\includegraphics[width = 0.35\textwidth, trim=1cm 1cm 1cm 1cm, clip=true]{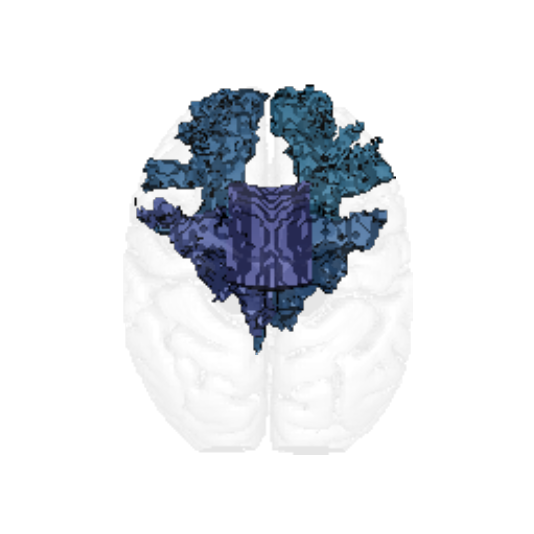}}\\
\vspace{-0pt} 
\subfigure[AD vs MCI]{\includegraphics[width = 0.4\textwidth, trim=1cm 1cm 1cm 2cm, clip=true]{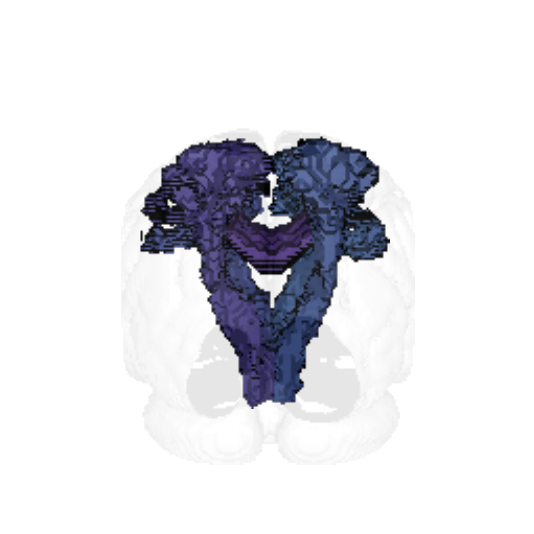}}
\subfigure[AD vs MCI]{\includegraphics[width = 0.35\textwidth, trim=1cm 1cm 1cm 1cm, clip=true]{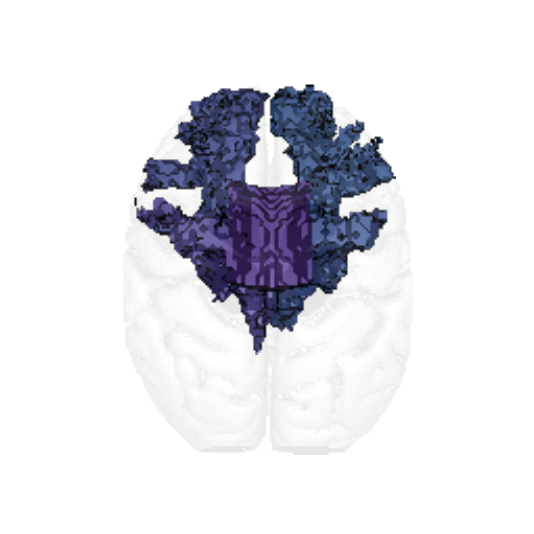}}\\
\vspace{-0pt} 
\subfigure[MCI vs NC]{\includegraphics[width = 0.4\textwidth, trim=1cm 1cm 1cm 2cm, clip=true]{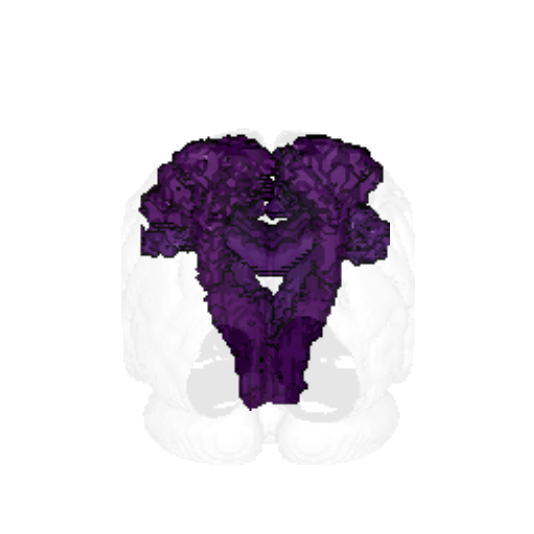}}
\subfigure[MCI vs NC]{\includegraphics[width = 0.35\textwidth, trim=1cm 1cm 1cm 1cm, clip=true]{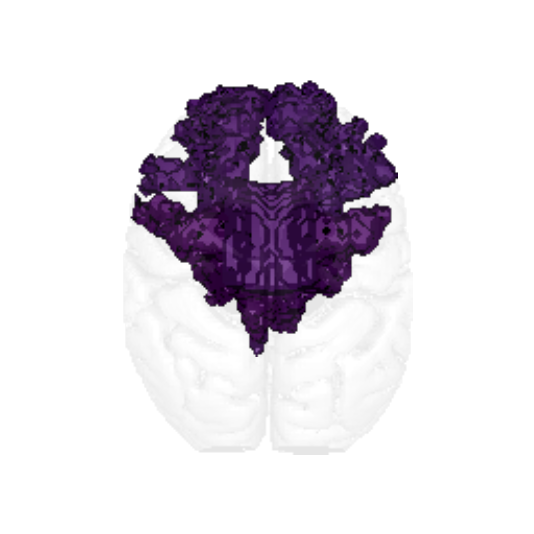}}
\end{minipage}
\begin{minipage}{0.04\textwidth}
\hspace*{-1.25cm}\subfigure{\includegraphics[width=3\textwidth, trim=10cm 0cm 0cm 0cm, clip=true]{Figures/Legend_2.pdf}}
\end{minipage}
\begin{minipage}{0.47\textwidth}
\subfigure[AD vs NC]{\includegraphics[width = 0.4\textwidth, trim=1cm 1cm 1cm 2cm, clip=true]{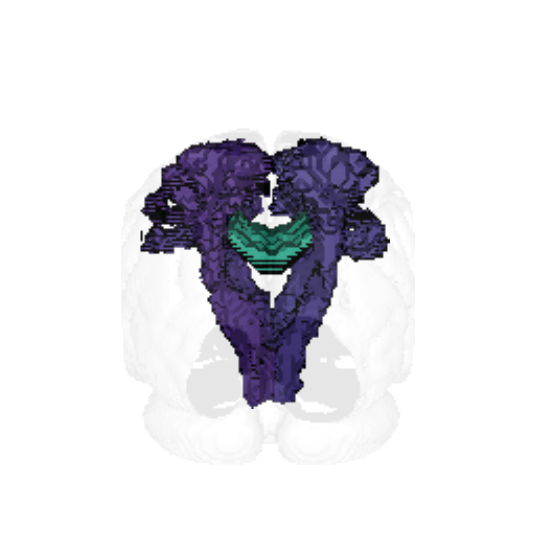}}
\subfigure[AD vs NC]{\includegraphics[width = 0.35\textwidth, trim=1cm 1cm 1cm 1cm, clip=true]{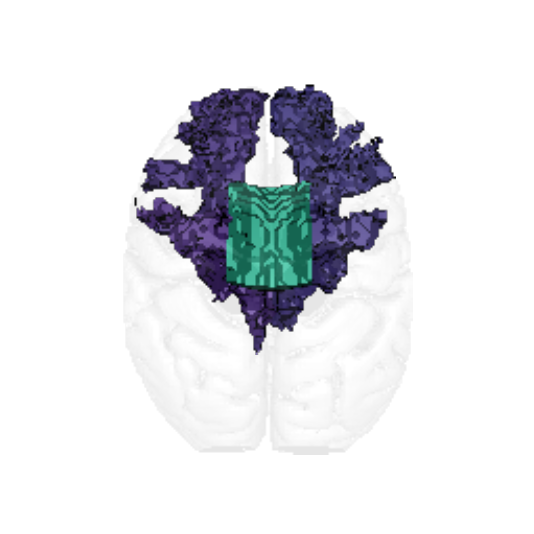}}\\
\vspace{-0pt} 
\subfigure[AD vs MCI]{\includegraphics[width = 0.4\textwidth, trim=1cm 1cm 1cm 2cm, clip=true]{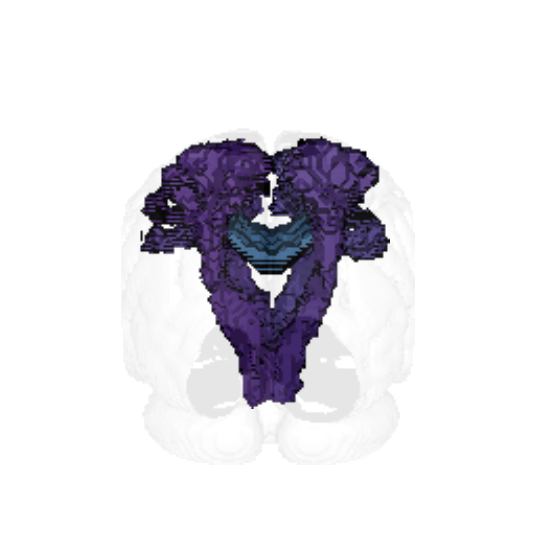}}
\subfigure[AD vs MCI]{\includegraphics[width = 0.35\textwidth, trim=1cm 1cm 1cm 1cm, clip=true]{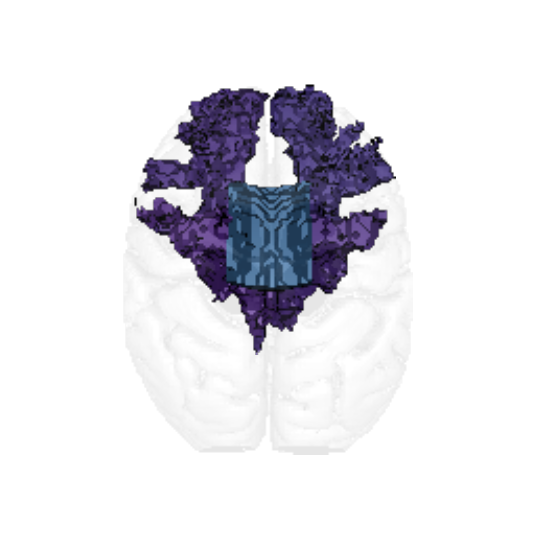}}\\
\vspace{-0pt} 
\subfigure[MCI vs NC]{\includegraphics[width = 0.4\textwidth, trim=1cm 1cm 1cm 2cm, clip=true]{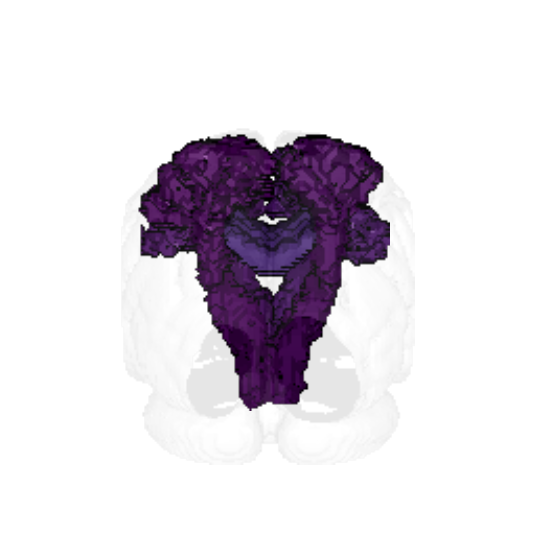}}
\subfigure[MCI vs NC]{\includegraphics[width = 0.35\textwidth, trim=1cm 1cm 1cm 1cm, clip=true]{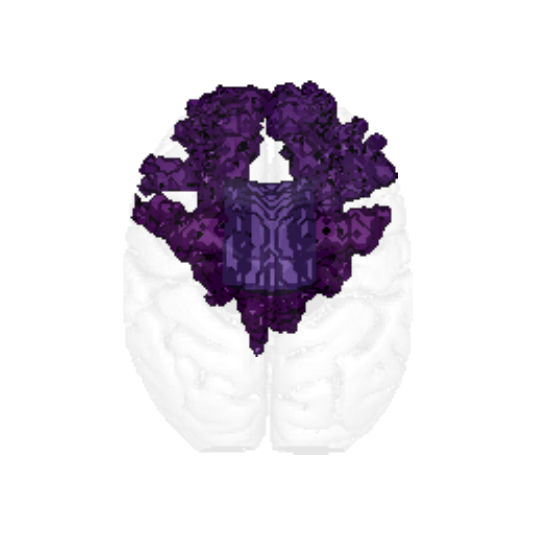}}
\end{minipage}
\caption{Inference for $\beta$'s: Summarized differences across five different tracts (left and right corticospinal tracts, left and right frontopontine tracts, and corpus callosum) in FWF (sub-panels (a)-(f) to the left of the scale) and NDI (sub-panels (h)-(m) to the right of the scale) from two different angles.}
\label{fig: realFWF_and_NDI}
\end{figure}

\clearpage\newpage

We now present results for $\text{Arc-Length}_{\beta}(h_{g})$ and $\text{CGD}_{\beta}(h_{g},h_{g}^{\prime})$ associated with 21 white matter fiber tracts.

\begin{table}[!htbp]
\centering
\caption{$\text{Arc-Length}_{\beta}(h_{g})$ for different disease status groups $h_{g}$ across the  left CST} 
\begin{tabular}{rrrr}
  \hline
 & Normal & MCI & AD \\ 
  \hline
FA & 0.584 & 0.693 & 1.051 \\ 
  ODI & 1.137 & 1.332 & 1.347 \\ 
  FWF & 0.568 & 0.608 & 0.813 \\ 
  NDI & 0.498 & 0.546 & 0.923 \\ 
   \hline
\end{tabular}
\end{table}

\vskip -10pt
\begin{table}[!htbp]
\centering
\caption{$\text{Arc-Length}_{\beta}(h_{g})$ for different disease status groups $h_{g}$ across the  left DLF} 
\begin{tabular}{rrrr}
  \hline
 & Normal & MCI & AD \\ 
  \hline
FA & 1.105 & 1.137 & 1.367 \\ 
  ODI & 0.609 & 0.614 & 0.643 \\ 
  FWF & 0.930 & 0.903 & 1.341 \\ 
  NDI & 0.888 & 0.798 & 1.038 \\ 
   \hline
\end{tabular}
\end{table}

\vskip -10pt
\begin{table}[!htbp]
\centering
\caption{$\text{Arc-Length}_{\beta}(h_{g})$ for different disease status groups $h_{g}$ across the  left FPT} 
\begin{tabular}{rrrr}
  \hline
 & Normal & MCI & AD \\ 
  \hline
FA & 0.723 & 0.861 & 1.119 \\ 
  ODI & 0.978 & 1.165 & 1.383 \\ 
  FWF & 0.860 & 0.825 & 1.199 \\ 
  NDI & 0.608 & 0.741 & 0.811 \\ 
   \hline
\end{tabular}
\end{table}

\vskip -10pt
\begin{table}[!htbp]
\centering
\caption{$\text{Arc-Length}_{\beta}(h_{g})$ for different disease status groups $h_{g}$ across the  left MFT} 
\begin{tabular}{rrrr}
  \hline
 & Normal & MCI & AD \\ 
  \hline
FA & 1.021 & 0.978 & 1.368 \\ 
  ODI & 1.101 & 0.979 & 1.282 \\ 
  FWF & 0.927 & 0.940 & 1.363 \\ 
  NDI & 0.700 & 0.696 & 0.945 \\ 
   \hline
\end{tabular}
\end{table}

\vskip -10pt
\begin{table}[!htbp]
\centering
\caption{$\text{Arc-Length}_{\beta}(h_{g})$ for different disease status groups $h_{g}$ across the  left MLF} 
\begin{tabular}{rrrr}
  \hline
 & Normal & MCI & AD \\ 
  \hline
FA & 1.400 & 1.386 & 1.837 \\ 
  ODI & 0.829 & 0.689 & 0.724 \\ 
  FWF & 0.616 & 0.680 & 1.060 \\ 
  NDI & 0.573 & 0.439 & 1.061 \\ 
   \hline
\end{tabular}
\end{table}
\begin{table}[!htbp]
\centering
\caption{$\text{Arc-Length}_{\beta}(h_{g})$ for different disease status groups $h_{g}$ across the  left NST} 

\begin{tabular}{rrrr}
  \hline
 & Normal & MCI & AD \\ 
  \hline
FA & 1.080 & 1.010 & 1.258 \\ 
  ODI & 1.146 & 0.753 & 1.147 \\ 
  FWF & 0.730 & 0.708 & 1.254 \\ 
  NDI & 0.756 & 0.776 & 1.125 \\ 
   \hline
\end{tabular}
\end{table}
\begin{table}[!htbp]
\centering
\caption{$\text{Arc-Length}_{\beta}(h_{g})$ for different disease status groups $h_{g}$ across the  left PPT} 
\begin{tabular}{rrrr}
  \hline
 & Normal & MCI & AD \\ 
  \hline
FA & 0.458 & 0.594 & 0.785 \\ 
  ODI & 0.741 & 0.796 & 0.967 \\ 
  FWF & 0.452 & 0.533 & 0.681 \\ 
  NDI & 0.362 & 0.514 & 0.572 \\ 
   \hline
\end{tabular}

\end{table}
\begin{table}[!htbp]
\centering
\caption{$\text{Arc-Length}_{\beta}(h_{g})$ for different disease status groups $h_{g}$ across the  left SCP} 
\begin{tabular}{rrrr}
  \hline
 & Normal & MCI & AD \\ 
  \hline
FA & 0.920 & 1.002 & 1.121 \\ 
  ODI & 0.722 & 0.773 & 0.835 \\ 
  FWF & 0.742 & 0.707 & 1.066 \\ 
  NDI & 0.565 & 0.602 & 0.899 \\ 
   \hline
\end{tabular}

\end{table}
\begin{table}[!htbp]
\centering
\caption{$\text{Arc-Length}_{\beta}(h_{g})$ for different disease status groups $h_{g}$ across the  left STT} 
\begin{tabular}{rrrr}
  \hline
 & Normal & MCI & AD \\ 
  \hline
FA & 0.778 & 0.726 & 1.044 \\ 
  ODI & 1.032 & 1.135 & 1.190 \\ 
  FWF & 0.582 & 0.674 & 0.935 \\ 
  NDI & 0.490 & 0.583 & 0.807 \\ 
   \hline
\end{tabular}

\end{table}
\begin{table}[!htbp]
\centering
\caption{$\text{Arc-Length}_{\beta}(h_{g})$ for different disease status groups $h_{g}$ across the  left TPT} 
\begin{tabular}{rrrr}
  \hline
 & Normal & MCI & AD \\ 
  \hline
FA & 0.776 & 0.823 & 1.309 \\ 
  ODI & 0.908 & 0.852 & 1.061 \\ 
  FWF & 0.510 & 0.573 & 0.664 \\ 
  NDI & 0.568 & 0.719 & 0.966 \\ 
   \hline
\end{tabular}

\end{table}
\begin{table}[!htbp]
\centering
\caption{$\text{Arc-Length}_{\beta}(h_{g})$ for different disease status groups $h_{g}$ across the  right CST} 
\begin{tabular}{rrrr}
  \hline
 & Normal & MCI & AD \\ 
  \hline
FA & 0.616 & 0.528 & 0.989 \\ 
  ODI & 1.267 & 1.063 & 1.365 \\ 
  FWF & 0.573 & 0.589 & 0.808 \\ 
  NDI & 0.500 & 0.503 & 0.830 \\ 
   \hline
\end{tabular}

\end{table}
\begin{table}[!htbp]
\centering
\caption{$\text{Arc-Length}_{\beta}(h_{g})$ for different disease status groups $h_{g}$ across the  right DLF} 

\begin{tabular}{rrrr}
  \hline
 & Normal & MCI & AD \\ 
  \hline
FA & 0.974 & 0.981 & 1.213 \\ 
  ODI & 0.609 & 0.661 & 0.690 \\ 
  FWF & 0.926 & 0.871 & 1.329 \\ 
  NDI & 0.914 & 0.775 & 1.124 \\ 
   \hline
\end{tabular}
\end{table}
\begin{table}[!htbp]
\centering
\caption{$\text{Arc-Length}_{\beta}(h_{g})$ for different disease status groups $h_{g}$ across the  right FPT} 

\begin{tabular}{rrrr}
  \hline
 & Normal & MCI & AD \\ 
  \hline
FA & 0.783 & 0.724 & 1.051 \\ 
  ODI & 1.133 & 0.968 & 1.325 \\ 
  FWF & 0.831 & 0.788 & 1.161 \\ 
  NDI & 0.586 & 0.658 & 0.717 \\ 
   \hline
\end{tabular}
\end{table}
\begin{table}[!htbp]
\centering
\caption{$\text{Arc-Length}_{\beta}(h_{g})$ for different disease status groups $h_{g}$ across the  right MFT} 
\begin{tabular}{rrrr}
  \hline
 & Normal & MCI & AD \\ 
  \hline
FA & 0.963 & 0.920 & 1.244 \\ 
  ODI & 1.068 & 0.976 & 1.295 \\ 
  FWF & 0.916 & 0.916 & 1.329 \\ 
  NDI & 0.695 & 0.697 & 0.906 \\ 
   \hline
\end{tabular}
\end{table}
\begin{table}[!htbp]
\centering
\caption{$\text{Arc-Length}_{\beta}(h_{g})$ for different disease status groups $h_{g}$ across the  right MLF} 

\begin{tabular}{rrrr}
  \hline
 & Normal & MCI & AD \\ 
  \hline
FA & 1.309 & 1.286 & 1.725 \\ 
  ODI & 0.769 & 0.649 & 0.725 \\ 
  FWF & 0.578 & 0.658 & 1.036 \\ 
  NDI & 0.570 & 0.442 & 1.002 \\ 
   \hline
\end{tabular}
\end{table}
\begin{table}[!htbp]
\centering
\caption{$\text{Arc-Length}_{\beta}(h_{g})$ for different disease status groups $h_{g}$ across the  right NST} 

\begin{tabular}{rrrr}
  \hline
 & Normal & MCI & AD \\ 
  \hline
FA & 1.018 & 1.124 & 1.290 \\ 
  ODI & 0.759 & 1.182 & 1.289 \\ 
  FWF & 0.749 & 0.741 & 1.231 \\ 
  NDI & 0.684 & 0.755 & 1.055 \\ 
   \hline
\end{tabular}
\end{table}
\begin{table}[!htbp]
\centering
\caption{$\text{Arc-Length}_{\beta}(h_{g})$ for different disease status groups $h_{g}$ across the  right PPT} 

\begin{tabular}{rrrr}
  \hline
 & Normal & MCI & AD \\ 
  \hline
FA & 0.433 & 0.497 & 0.686 \\ 
  ODI & 0.710 & 0.850 & 0.958 \\ 
  FWF & 0.454 & 0.534 & 0.660 \\ 
  NDI & 0.376 & 0.477 & 0.557 \\ 
   \hline
\end{tabular}
\end{table}
\begin{table}[!htbp]
\centering
\caption{$\text{Arc-Length}_{\beta}(h_{g})$ for different disease status groups $h_{g}$ across the  right SCP} 

\begin{tabular}{rrrr}
  \hline
 & Normal & MCI & AD \\ 
  \hline
FA & 0.915 & 0.927 & 1.097 \\ 
  ODI & 0.683 & 0.780 & 0.862 \\ 
  FWF & 0.822 & 0.755 & 1.164 \\ 
  NDI & 0.663 & 0.689 & 1.008 \\ 
   \hline
\end{tabular}
\end{table}
\begin{table}[!htbp]
\centering
\caption{$\text{Arc-Length}_{\beta}(h_{g})$ for different disease status groups $h_{g}$ across the  right STT} 

\begin{tabular}{rrrr}
  \hline
 & Normal & MCI & AD \\ 
  \hline
FA & 0.670 & 0.628 & 0.947 \\ 
  ODI & 1.133 & 0.971 & 1.270 \\ 
  FWF & 0.565 & 0.657 & 0.889 \\ 
  NDI & 0.485 & 0.519 & 0.751 \\ 
   \hline
\end{tabular}
\end{table}
\begin{table}[!htbp]
\centering
\caption{$\text{Arc-Length}_{\beta}(h_{g})$ for different disease status groups $h_{g}$ across the  right TPT} 

\begin{tabular}{rrrr}
  \hline
 & Normal & MCI & AD \\ 
  \hline
FA & 0.736 & 0.765 & 1.227 \\ 
  ODI & 0.934 & 0.942 & 1.201 \\ 
  FWF & 0.498 & 0.544 & 0.643 \\ 
  NDI & 0.524 & 0.705 & 0.915 \\ 
   \hline
\end{tabular}
\end{table}
\begin{table}[!htbp]
\centering
\caption{$\text{Arc-Length}_{\beta}(h_{g})$ for different disease status groups $h_{g}$ across the corpus callosum} 
\begin{tabular}{rrrr}
  \hline
 & Normal & MCI & AD \\ 
  \hline
FA & 0.532 & 0.622 & 0.771 \\ 
  ODI & 1.224 & 1.489 & 2.076 \\ 
  FWF & 1.035 & 1.353 & 1.504 \\ 
  NDI & 0.651 & 0.618 & 0.966 \\ 
   \hline
\end{tabular}

\end{table}

\begin{table}[!htbp]
\centering
\caption{$\text{CGD}_{\beta}(h_{g},h_{g}^{\prime})$ between pairs of disease status groups $(h_{g},h_{g}^{\prime})$  across the left CST} 
\begin{tabular}{rrrr}
  \hline
 & Normal-MCI & AD-MCI & Normal-AD \\ 
  \hline
FA & 0.041 & 0.080 & 0.084 \\ 
  ODI & 0.015 & 0.029 & 0.038 \\ 
  FWF & 0.012 & 0.061 & 0.073 \\ 
  NDI & 0.010 & 0.035 & 0.044 \\ 
   \hline
\end{tabular}

\end{table}
\begin{table}[!htbp]
\centering
\caption{$\text{CGD}_{\beta}(h_{g},h_{g}^{\prime})$ between pairs of disease status groups $(h_{g},h_{g}^{\prime})$  across the left DLF} 
\begin{tabular}{rrrr}
  \hline
 & Normal-MCI & AD-MCI & Normal-AD \\ 
  \hline
FA & 0.013 & 0.041 & 0.054 \\ 
  ODI & 0.026 & 0.096 & 0.117 \\ 
  FWF & 0.038 & 0.066 & 0.095 \\ 
  NDI & 0.044 & 0.059 & 0.100 \\ 
   \hline
\end{tabular}

\end{table}
\begin{table}[!htbp]
\centering
\caption{$\text{CGD}_{\beta}(h_{g},h_{g}^{\prime})$ between pairs of disease status groups $(h_{g},h_{g}^{\prime})$  across the left FPT} 
\begin{tabular}{rrrr}
  \hline
 & Normal-MCI & AD-MCI & Normal-AD \\ 
  \hline
FA & 0.030 & 0.081 & 0.092 \\ 
  ODI & 0.029 & 0.063 & 0.067 \\ 
  FWF & 0.020 & 0.066 & 0.086 \\ 
  NDI & 0.019 & 0.042 & 0.046 \\ 
   \hline
\end{tabular}

\end{table}
\begin{table}[!htbp]
\centering
\caption{$\text{CGD}_{\beta}(h_{g},h_{g}^{\prime})$ between pairs of disease status groups $(h_{g},h_{g}^{\prime})$  across the left MFT} 
\begin{tabular}{rrrr}
  \hline
 & Normal-MCI & AD-MCI & Normal-AD \\ 
  \hline
FA & 0.033 & 0.083 & 0.100 \\ 
  ODI & 0.028 & 0.073 & 0.083 \\ 
  FWF & 0.041 & 0.108 & 0.180 \\ 
  NDI & 0.023 & 0.070 & 0.073 \\ 
   \hline
\end{tabular}

\end{table}
\begin{table}[!htbp]
\centering
\caption{$\text{CGD}_{\beta}(h_{g},h_{g}^{\prime})$ between pairs of disease status groups $(h_{g},h_{g}^{\prime})$  across the left MLF} 
\begin{tabular}{rrrr}
  \hline
 & Normal-MCI & AD-MCI & Normal-AD \\ 
  \hline
FA & 0.035 & 0.038 & 0.033 \\ 
  ODI & 0.017 & 0.104 & 0.074 \\ 
  FWF & 0.042 & 0.098 & 0.157 \\ 
  NDI & 0.012 & 0.100 & 0.121 \\ 
   \hline
\end{tabular}

\end{table}
\begin{table}[!htbp]
\centering
\caption{$\text{CGD}_{\beta}(h_{g},h_{g}^{\prime})$ between pairs of disease status groups $(h_{g},h_{g}^{\prime})$  across the left NST} 
\begin{tabular}{rrrr}
  \hline
 & Normal-MCI & AD-MCI & Normal-AD \\ 
  \hline
FA & 0.030 & 0.071 & 0.077 \\ 
  ODI & 0.026 & 0.050 & 0.054 \\ 
  FWF & 0.031 & 0.061 & 0.099 \\ 
  NDI & 0.012 & 0.071 & 0.055 \\ 
   \hline
\end{tabular}

\end{table}
\begin{table}[!htbp]
\centering
\caption{$\text{CGD}_{\beta}(h_{g},h_{g}^{\prime})$ between pairs of disease status groups $(h_{g},h_{g}^{\prime})$  across the left PPT} 
\begin{tabular}{rrrr}
  \hline
 & Normal-MCI & AD-MCI & Normal-AD \\ 
  \hline
FA & 0.024 & 0.058 & 0.060 \\ 
  ODI & 0.009 & 0.018 & 0.023 \\ 
  FWF & 0.010 & 0.039 & 0.051 \\ 
  NDI & 0.009 & 0.020 & 0.023 \\ 
   \hline
\end{tabular}

\end{table}
\begin{table}[!htbp]
\centering
\caption{$\text{CGD}_{\beta}(h_{g},h_{g}^{\prime})$ between pairs of disease status groups $(h_{g},h_{g}^{\prime})$  across the left SCP} 
\begin{tabular}{rrrr}
  \hline
 & Normal-MCI & AD-MCI & Normal-AD \\ 
  \hline
FA & 0.016 & 0.034 & 0.047 \\ 
  ODI & 0.017 & 0.051 & 0.045 \\ 
  FWF & 0.021 & 0.113 & 0.152 \\ 
  NDI & 0.012 & 0.049 & 0.054 \\ 
   \hline
\end{tabular}

\end{table}
\begin{table}[!htbp]
\centering
\caption{$\text{CGD}_{\beta}(h_{g},h_{g}^{\prime})$ between pairs of disease status groups $(h_{g},h_{g}^{\prime})$  across the left STT} 
\begin{tabular}{rrrr}
  \hline
 & Normal-MCI & AD-MCI & Normal-AD \\ 
  \hline
FA & 0.034 & 0.070 & 0.076 \\ 
  ODI & 0.013 & 0.037 & 0.038 \\ 
  FWF & 0.014 & 0.072 & 0.083 \\ 
  NDI & 0.011 & 0.031 & 0.036 \\ 
   \hline
\end{tabular}

\end{table}
\begin{table}[!htbp]
\centering
\caption{$\text{CGD}_{\beta}(h_{g},h_{g}^{\prime})$ between pairs of disease status groups $(h_{g},h_{g}^{\prime})$  across the left TPT} 
\begin{tabular}{rrrr}
  \hline
 & Normal-MCI & AD-MCI & Normal-AD \\ 
  \hline
FA & 0.013 & 0.041 & 0.050 \\ 
  ODI & 0.010 & 0.013 & 0.017 \\ 
  FWF & 0.016 & 0.065 & 0.071 \\ 
  NDI & 0.012 & 0.055 & 0.058 \\ 
   \hline
\end{tabular}

\end{table}

\begin{table}[!htbp]
\centering
\caption{$\text{CGD}_{\beta}(h_{g},h_{g}^{\prime})$ between pairs of disease status groups $(h_{g},h_{g}^{\prime})$  across the right CST} 
\begin{tabular}{rrrr}
  \hline
 & Normal-MCI & AD-MCI & Normal-AD \\ 
  \hline
FA & 0.039 & 0.090 & 0.097 \\ 
  ODI & 0.015 & 0.030 & 0.041 \\ 
  FWF & 0.010 & 0.047 & 0.058 \\ 
  NDI & 0.008 & 0.030 & 0.034 \\ 
   \hline
\end{tabular}

\end{table}
\begin{table}[!htbp]
\centering
\caption{$\text{CGD}_{\beta}(h_{g},h_{g}^{\prime})$ between pairs of disease status groups $(h_{g},h_{g}^{\prime})$  across the right DLF} 
\begin{tabular}{rrrr}
  \hline
 & Normal-MCI & AD-MCI & Normal-AD \\ 
  \hline
FA & 0.012 & 0.048 & 0.058 \\ 
  ODI & 0.026 & 0.096 & 0.128 \\ 
  FWF & 0.031 & 0.060 & 0.089 \\ 
  NDI & 0.052 & 0.076 & 0.139 \\ 
   \hline
\end{tabular}

\end{table}
\begin{table}[!htbp]
\centering
\caption{$\text{CGD}_{\beta}(h_{g},h_{g}^{\prime})$ between pairs of disease status groups $(h_{g},h_{g}^{\prime})$  across the right FPT} 
\begin{tabular}{rrrr}
  \hline
 & Normal-MCI & AD-MCI & Normal-AD \\ 
  \hline
FA & 0.038 & 0.071 & 0.093 \\ 
  ODI & 0.028 & 0.061 & 0.062 \\ 
  FWF & 0.018 & 0.057 & 0.077 \\ 
  NDI & 0.015 & 0.037 & 0.040 \\ 
   \hline
\end{tabular}

\end{table}
\begin{table}[!htbp]
\centering
\caption{$\text{CGD}_{\beta}(h_{g},h_{g}^{\prime})$ between pairs of disease status groups $(h_{g},h_{g}^{\prime})$  across the right MFT} 
\begin{tabular}{rrrr}
  \hline
 & Normal-MCI & AD-MCI & Normal-AD \\ 
  \hline
FA & 0.034 & 0.099 & 0.119 \\ 
  ODI & 0.025 & 0.078 & 0.093 \\ 
  FWF & 0.040 & 0.106 & 0.187 \\ 
  NDI & 0.019 & 0.070 & 0.081 \\ 
   \hline
\end{tabular}

\end{table}
\begin{table}[!htbp]
\centering
\caption{$\text{CGD}_{\beta}(h_{g},h_{g}^{\prime})$ between pairs of disease status groups $(h_{g},h_{g}^{\prime})$  across the right MLF} 
\begin{tabular}{rrrr}
  \hline
 & Normal-MCI & AD-MCI & Normal-AD \\ 
  \hline
FA & 0.032 & 0.034 & 0.038 \\ 
  ODI & 0.016 & 0.075 & 0.113 \\ 
  FWF & 0.032 & 0.089 & 0.134 \\ 
  NDI & 0.013 & 0.087 & 0.125 \\ 
   \hline
\end{tabular}

\end{table}
\begin{table}[!htbp]
\centering
\caption{$\text{CGD}_{\beta}(h_{g},h_{g}^{\prime})$ between pairs of disease status groups $(h_{g},h_{g}^{\prime})$  across the right NST} 
\begin{tabular}{rrrr}
  \hline
 & Normal-MCI & AD-MCI & Normal-AD \\ 
  \hline
FA & 0.038 & 0.084 & 0.094 \\ 
  ODI & 0.024 & 0.047 & 0.069 \\ 
  FWF & 0.027 & 0.066 & 0.103 \\ 
  NDI & 0.015 & 0.071 & 0.055 \\ 
   \hline
\end{tabular}

\end{table}
\begin{table}[!htbp]
\centering
\caption{$\text{CGD}_{\beta}(h_{g},h_{g}^{\prime})$ between pairs of disease status groups $(h_{g},h_{g}^{\prime})$  across the right PPT} 
\begin{tabular}{rrrr}
  \hline
 & Normal-MCI & AD-MCI & Normal-AD \\ 
  \hline
FA & 0.029 & 0.058 & 0.061 \\ 
  ODI & 0.008 & 0.018 & 0.022 \\ 
  FWF & 0.008 & 0.024 & 0.032 \\ 
  NDI & 0.008 & 0.019 & 0.021 \\ 
   \hline
\end{tabular}

\end{table}
\begin{table}[!htbp]
\centering
\caption{$\text{CGD}_{\beta}(h_{g},h_{g}^{\prime})$ between pairs of disease status groups $(h_{g},h_{g}^{\prime})$  across the right SCP} 
\begin{tabular}{rrrr}
  \hline
 & Normal-MCI & AD-MCI & Normal-AD \\ 
  \hline
FA & 0.017 & 0.042 & 0.049 \\ 
  ODI & 0.018 & 0.057 & 0.058 \\ 
  FWF & 0.019 & 0.104 & 0.126 \\ 
  NDI & 0.018 & 0.062 & 0.067 \\ 
   \hline
\end{tabular}

\end{table}
\begin{table}[!htbp]
\centering
\caption{$\text{CGD}_{\beta}(h_{g},h_{g}^{\prime})$ between pairs of disease status groups $(h_{g},h_{g}^{\prime})$  across the right STT} 
\begin{tabular}{rrrr}
  \hline
 & Normal-MCI & AD-MCI & Normal-AD \\ 
  \hline
FA & 0.038 & 0.087 & 0.089 \\ 
  ODI & 0.012 & 0.037 & 0.037 \\ 
  FWF & 0.011 & 0.050 & 0.060 \\ 
  NDI & 0.010 & 0.027 & 0.030 \\ 
   \hline
\end{tabular}
\end{table}
\begin{table}[!htbp]
\centering
\caption{$\text{CGD}_{\beta}(h_{g},h_{g}^{\prime})$ between pairs of disease status groups $(h_{g},h_{g}^{\prime})$  across the right TPT} 
\begin{tabular}{rrrr}
  \hline
 & Normal-MCI & AD-MCI & Normal-AD \\ 
  \hline
FA & 0.014 & 0.059 & 0.066 \\ 
  ODI & 0.008 & 0.014 & 0.015 \\ 
  FWF & 0.011 & 0.046 & 0.054 \\ 
  NDI & 0.010 & 0.053 & 0.045 \\ 
   \hline
\end{tabular}
\end{table}
\begin{table}[!htbp]
\centering
\caption{$\text{CGD}_{\beta}(h_{g},h_{g}^{\prime})$ between pairs of disease status groups $(h_{g},h_{g}^{\prime})$  across the corpus callosum} 
\begin{tabular}{rrrr}
  \hline
 & Normal-MCI & AD-MCI & Normal-AD \\ 
  \hline
FA & 0.071 & 0.109 & 0.149 \\ 
  ODI & 0.029 & 0.198 & 0.229 \\ 
  FWF & 0.016 & 0.036 & 0.053 \\ 
  NDI & 0.031 & 0.076 & 0.137 \\ 
   \hline
\end{tabular}
\end{table}

\newpage

\section{Additional Details and Figures for the Simulation Experiments}
\label{sec: sm sim add details and figs}

As outlined in Section \ref{sec: sim study} in the main paper, our simulation study evaluates the benefits of our flexible and scalable compact HOSVD-based model with smoothly varying functional mode matrices, compared to both smooth and non-smooth CP-based models and a non-smooth compact HOSVD alternative.

The compact HOSVD, an equivalent representation of Tucker decomposition with semi-orthogonal mode matrices (Lemma \ref{lem: Tucker-HOSVD}), enhances posterior computation (Section \ref{sec: sm post comp}). For a fair comparison, we also consider a mixed-effects longitudinal CP factorized model with varying ranks and a voxel-wise mixed model (VMW).

First, to compare our approach with flexible CP decomposition-based alternatives (CP model), we again use mixtures of B-splines with time-varying coefficients $\wt\beta(h_{g},h_{1},h_{2},h_{3},h_{t})$ to model the longitudinal changes, 
but now decompose them using a CP approach instead of a compact HOSVD. 
To our knowledge, a model of this exact form has not appeared in the literature, but 
the closest is \cite{niyogi2024tensor} who considered a similar approach in an image-on-image regression setting but with no random effects. 
However, to make our comparisons fair and strict, we still allow for multiplicative random effects in the (now diagonal) core coefficients, 
and further apply the Laplacian regularization on $\{a^{(s)}_{\beta,z_{r}}(h_{s})\}_{1 \leq h_{s}\leq d_{s}}$. 

Finally, we consider a similar CP-based model for the $\alpha$'s as well. 
Specifically, we let  
\bse 
& \alpha^{(i)}(h_{g},h_{1},h_{2},h_{3}) = 
\sum_{z_{r}=1}^{r_{\alpha}} \eta_{\alpha,z_{r}}^{(i)} \prod_{s \in \S_{g}} a^{(s)}_{\alpha,z_{r}}(h_{s}), ~~~
& \eta_{\alpha,z_{r}}^{(i)} \simind \Normal(c_{\alpha,z_{r}}, \tau^{2}_{\alpha}\sigma_{\alpha,z_{r}}^{2}),\\ 
& \wt\beta^{(i)}(h_{g},h_{1},h_{2},h_{3},h_{t}) = \sum_{z_{r}=1}^{r_{\beta}} \eta_{\beta,z_{r}}^{(i)} \prod_{s \in \S_{g,t}} a_{\beta,z_{r}}^{(s)}(h_{s}), ~~~
&\eta_{\beta,z_{r}}^{(i)} \simind \Normal(c_{\beta,z_{r}},\tau_{\beta}^{2} \sigma_{\beta,z_{r}}^{2}),
\ese
other model components remaining similar. 
The cumulative shrinkage prior on the 
the mode matrices (Section \ref{sec: priors}) still allows for semi-automated rank selection and safeguards against overfitting.
We implemented CP models with starting ranks $r_{\alpha}=r_{\beta}\in\{5,10\}$.
Unlike the Tucker case, the mode matrices in CP are not set to be semi-orthogonal as the CP rank is often greater than the dimensions of several of the directions.

Second, to highlight the importance of functional tensor components in modeling smoothly varying images, 
we compare our graph Laplacian-enforced smooth compact HOSVD model (Graph-Laplacian-Smooth, GLS)
with a version that still has the graph Laplacian but is non-smooth (Graph-Laplacian-Non-Smooth, GLNS), keeping the other elements unchanged. 

We also include a voxel-wise mixed (VWM) model in our comparisons. 
The VWM model fits $y^{(i)}(h_{g},h_{1},h_{2},h_{3},t)=f_{h_{g},h_{1},h_{2},h_{3}}(t)+\tau^{(i)}_{h_{g},h_{1},h_{2},h_{3}}+\epsilon^{(i)}_{h_{g},h_{1},h_{2},h_{3},t}$ at each voxel $(h_{1},h_{2},h_{3})$, 
setting the shape constraint $f_{h_{g},h_{1},h_{2},h_{3}}(0)=0$ by invoking the same spline-based model described in Section \ref{sec: models}.

As our scientific focus is on estimating longitudinal changes modeled by B-spline mixtures parameterized by $\bbeta$, we prioritize the accuracy of $\bbeta$ estimation in our comparisons.



To generate our synthetic datasets, 
we randomly select 20 subjects each from 3 groups and use their longitudinal visit times 
{(Figure \ref{fig: obs_time})}. 
The image dimensions are set to be $35 \times 35 \times 35$. 
The true $\alpha^{(i)}$'s are set at the corresponding real data estimates obtained in Section \ref{sec: adni analysis}. 
We form the true $\beta(h_{g},h_{1},h_{2},h_{3},t)=\sum_{\ell=1}^{6} 4t^{2}\gamma_{\ell,h_{g}}\exp\{-(h_{1}-n_{1}u_{1,\ell})^{2}/5-(h_{2}-n_{2}u_{2,\ell})^{2}/5-(h_{3}-n_{3}u_{3,\ell})^{2}/5\}$, where $\gamma_{\ell,h_{g}}=1$ for $\ell \leq 2 h_{g}$ otherwise 0, 
and $u_{m,\ell}$'s are set as 
$u_{1} = [0.50,0.50,0.30,0.80,0.30,0.60]\trans$, 
$u_{2} = [0.40,0.80,0.40,0.80,0.60,0.40]\trans$, 
and $u_{3} = [0.20,0.20,0.60,0.60,0.50,0.50]\trans$ 
This allows us to compare the methods under different signal strengths simultaneously. 
For example, $\beta_0(1,:,:,h_{3},0.6)$ varies in magnitude for various values of $h_{3}$ even at a given time point $t=0.6$ and group id $h_{g}=1$. 
We then set $\beta^{(i)}(\cdot)=\beta(\cdot)+e_{\beta}^{(i)}$, where $e_{\beta}^{(i)}$'s are independent $\Normal(0,0.1^{2})$. 
Finally, the $y^{(i)}(h_{g},h_{1},h_{2},h_{3},t)$'s are generated from $\Normal\{\alpha^{(i)}(h_{g},h_{1},h_{2},h_{3},t)+\beta^{(i)}(h_{g},h_{1},h_{2},h_{3},t), 0.5^{2}\}$. 



\begin{figure}[htbp]
     \centering
     \includegraphics[width=0.8\textwidth]{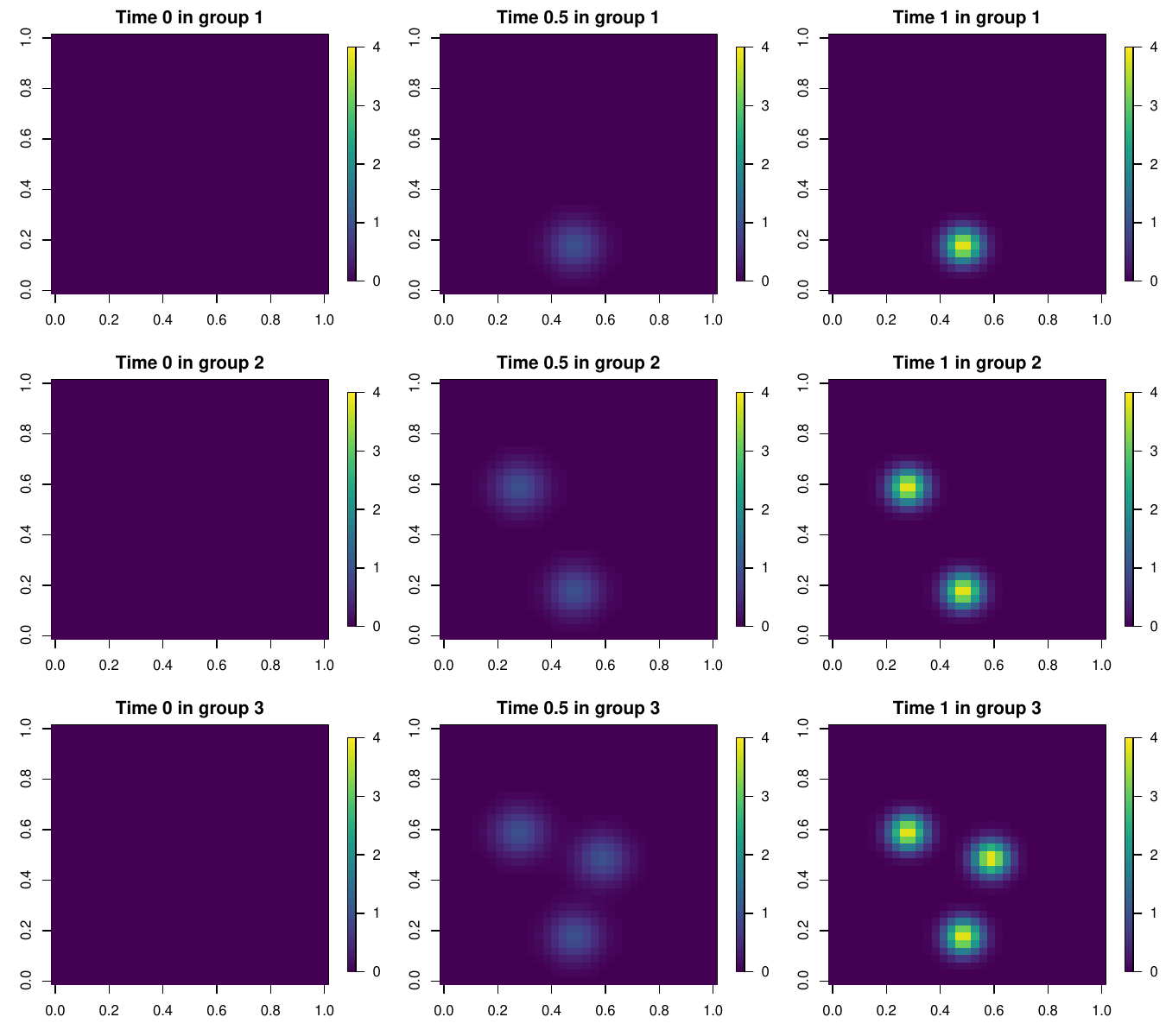}
     \caption{Inference for $\beta$'s: True $\beta(h_{g},h_{1},h_{2},h_{3},t)$ in slice $h_{3}=7$ for three groups $h_{g}=1,2,3$ at three-time points $t=0.0,0.5,1.0$ in simulated data. 
     See Section \ref{sec: sim study} in the main paper for additional details.
     }
     \label{fig: sim betas true}
\end{figure}

\begin{figure}[htbp]
     \centering
     \includegraphics[width=0.8\textwidth]{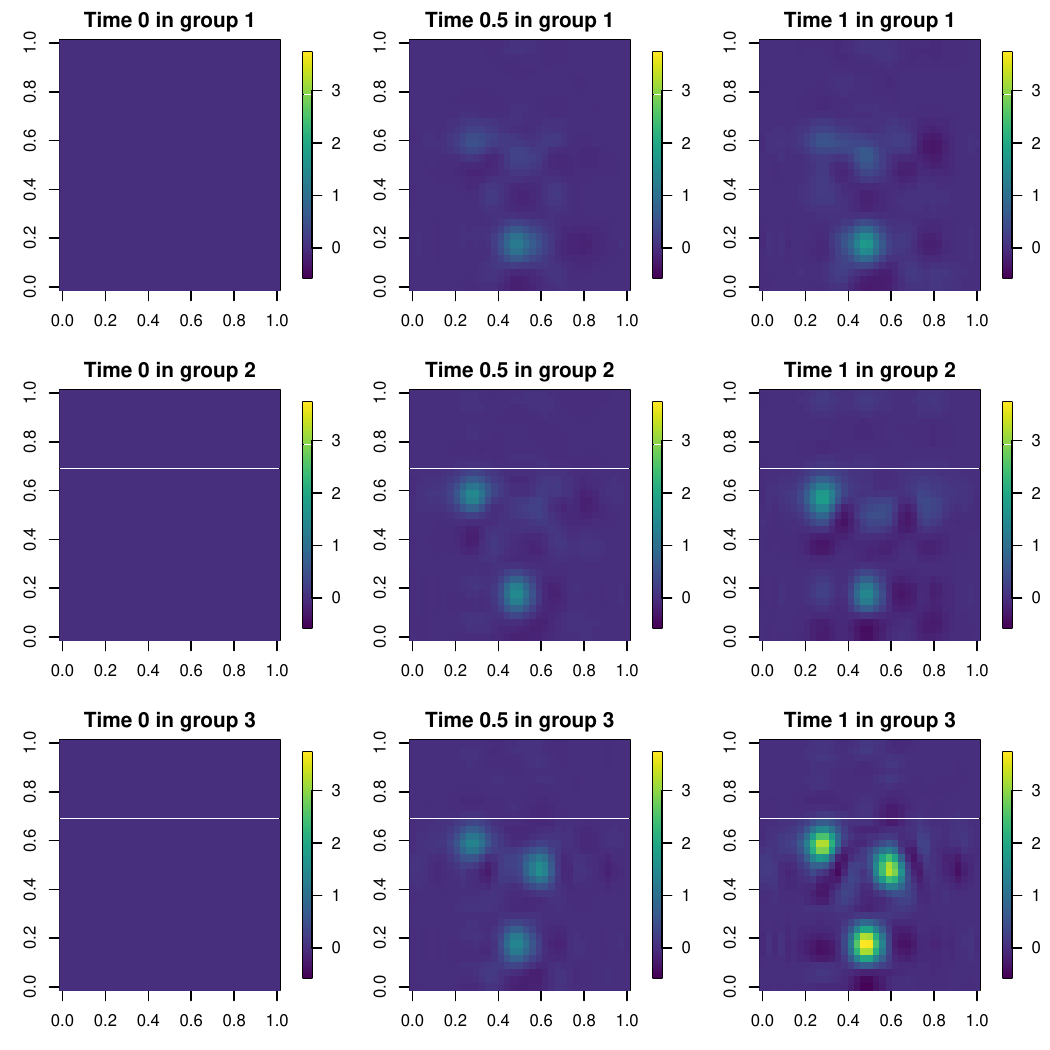}
     \caption{Inference for $\beta$'s: Estimated $\beta(h_{g},h_{1},h_{2},h_{3},t)$ in slice $h_{3}=7$ for three groups $h_{g}=1,2,3$ at three-time points $t=0.0,0.5,1.0$ from simulated data  obtained by our proposed approach.
     See Section \ref{sec: sim study} in the main paper for additional details.
     }
     \label{fig: sim betas esti}
\end{figure}

%
%

\clearpage
\bibliographystyle{natbib}
\bibliography{main}

\end{document}